\documentclass{article}

\usepackage[english]{babel}
\usepackage[utf8x]{inputenc}
\usepackage[T1]{fontenc}
\usepackage[charter]{mathdesign}
\usepackage{eucal}
\usepackage{microtype}

\usepackage[margin=1in,marginparwidth=1in]{geometry}

\usepackage{amsthm}
\usepackage{amsmath}
\usepackage{graphicx}
\usepackage[linesnumbered]{algorithm2e}
\SetArgSty{}
\usepackage[colorlinks=true, allcolors=blue]{hyperref}
\usepackage{enumitem}

\usepackage{tikz}
\usetikzlibrary{calc,positioning,backgrounds,intersections,patterns}
\usepackage{subcaption}

\newtheorem{theorem}{Theorem}
\newtheorem{lemma}[theorem]{Lemma}
\newtheorem{observation}[theorem]{Observation}

\theoremstyle{definition}
\newtheorem{definition}[theorem]{Definition}

\DeclareMathOperator{\OPT}{OPT}

\begin{document}

\title{Polynomial-Time Algorithms for Phylogenetic Inference Problems involving duplication and reticulation}
\author{Leo van Iersel\footnote{Delft Institute of Applied Mathematics, Delft
    University of Technology, Van Mourik Broekmanweg 6, 2628 XE, Delft, The
    Netherlands, \{L.J.J.vanIersel, R.Janssen-2, M.E.L.Jones,
      Y.Murakami\}@tudelft.nl. Research funded in part by the Netherlands
    Organization for Scientific Research (NWO), including Vidi grant
    639.072.602, and partly by the 4TU Applied Mathematics Institute.} \and
  Remie Janssen\footnotemark[1] \and
  Mark Jones\footnotemark[1] \and
  Yukihiro Murakami\footnotemark[1] \and
  Norbert Zeh\footnote{Faculty of Computer Science, Dalhousie University, 6050
    University Ave, Halifax, NS B3H 1W5, Canada, nzeh@cs.dal.ca.  Research
    funded in part by the Natural Sciences and Engineering Research Council of
    Canada.}} 
\maketitle

\begin{abstract}\noindent
  A common problem in phylogenetics is to try to infer a species phylogeny from
  gene trees. We consider different variants of this problem. The first
  variant, called \textsc{Unrestricted Minimal Episodes Inference}, aims at
  inferring a species tree based on a model with speciation and duplication where
  duplications are clustered in duplication episodes. The goal is to minimize
  the number of such episodes. The second variant, \textsc{Parental
    Hybridization}, aims at inferring a species \emph{network} based on a model
  with speciation and reticulation. The goal is to minimize the number of
  reticulation events. It is a variant of the well-studied
  \textsc{Hybridization Number} problem with a more generous view on which gene
  trees are consistent with a given species network.
  We show that these seemingly different problems are in fact closely related
  and can, surprisingly, both be solved in polynomial time, using a structure
  we call ``beaded trees''. However, we also show that methods based on these
  problems have to be used with care because the optimal species phylogenies
  always have a restricted form. To mitigate this problem, we introduce a
  new variant of \textsc{Unrestricted Minimal Episodes Inference} that
  minimizes the duplication episode depth. We prove that this new variant of
  the problem can also be solved in polynomial time.
\end{abstract}

\section{Introduction}\label{sec:introduction}

\emph{Phylogenetic trees} are commonly used to represent the evolutionary
history of a set of taxa.  The leaves represent extant taxa; internal nodes
represent speciation events that caused lineages to diverge. If we assume that
the only process is speciation and that no incomplete
lineage sorting occurs, then any gene will have a gene tree that is consistent
with the species phylogeny. In such cases, there exist efficient algorithms to
reconstruct a species tree from gene trees.  There are, however, evolutionary
processes beyond vertical inheritance of genetic material and speciation events
that make it more challenging to reconstruct the real evolutionary history.
Examples of such processes are hybridization, horizontal gene transfer, and
duplication. Each of these processes can result in discordance between gene
trees.

This leads to a number of problems in which the task is to minimize the number
of such complicating events.  In \emph{reconciliation problems}, we are given
the gene trees together with the species phylogeny, and the task is to find
optimal embeddings of the gene trees into the species phylogeny.  Such methods
are for example used to estimate dates of duplications, to discover
relationships between duplicate
genes~\cite{chan2013reconciliation,vernot2008reconciliation}, and to
reconstruct the infection history of parasites~\cite{page1994maps}. In
\emph{inference problems}, only the gene trees are given and we aim to find a
species phylogeny that minimizes the discordance with the gene trees. Such
problems are  relevant when the species phylogeny is not yet known with
certainty.

\paragraph{Duplication minimization problems.}

Gene duplications happen as a consequence of errors in the DNA replication
process.  This leads to a species having multiple copies of the same gene.
There exist many types of gene duplication, which depend on the positions of
errors within the replication
process~\cite{panchy2016evolution,reams2015mechanisms}. The scale of gene
duplications is determined by the number of genes that get duplicated.  An
extreme example of a large-scale duplication is \emph{Whole Genome Duplication
  (WGD)}, in which every gene in the genome is duplicated.  This process, also
known as polyploidization, occurs as a result of an error in separation of
chromosomes during gamete production.  It is most common in plants (see, e.g.,
Figure~\ref{fig:WGDExample}) but has also occurred in
animals~\cite{zhang2003evolution}, and there are two WGD events even in the
evolutionary history leading to humans \cite{dehal2005two,ohno1968evolution}.
Large-scale duplications provide species with diversification potential, giving
them the ability to quickly adapt to a changing
environment~\cite{ohno1970book,glasauer2014whole,zhang2003evolution}.

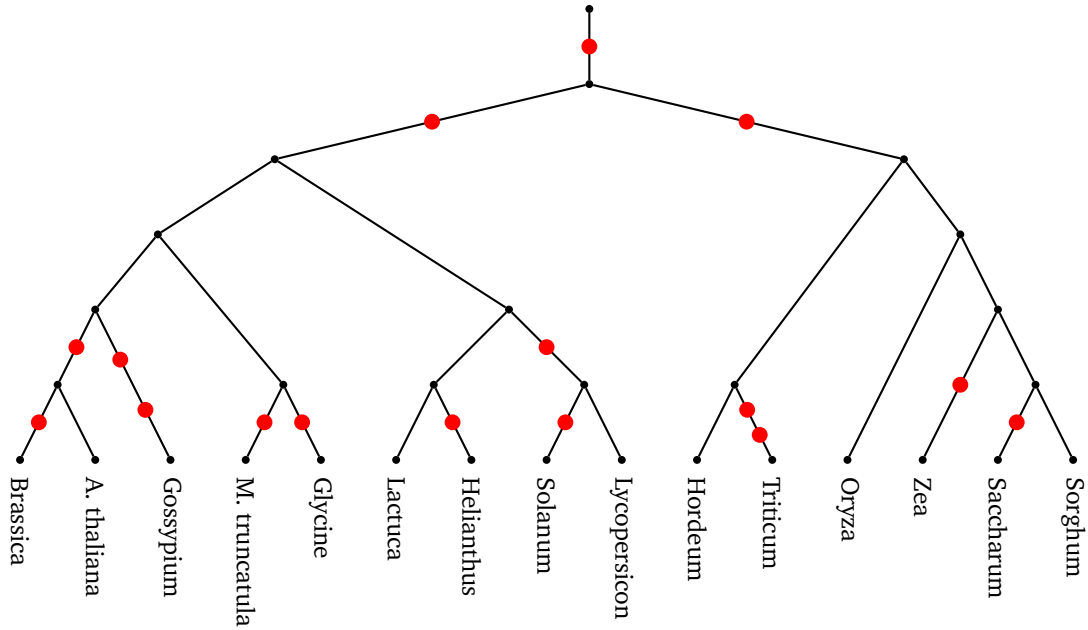
\begin{figure}[t]
  \centering
  \begin{tikzpicture}[
    node/.style={fill=black,circle,inner sep=0pt,minimum size=3pt,outer sep=0pt},
    dup/.style={fill=red,circle,inner sep=0pt,minimum size=6pt},
    edge/.style={draw,thick}]
    \foreach \i in {0,...,14} {
      \node [node] (l\i) at (\i,0) {};
    }
    \node [anchor=west,rotate=270,xshift=3pt] at (l0) {Brassica};
    \node [anchor=west,rotate=270,xshift=3pt] at (l1) {A. thaliana};
    \node [anchor=west,rotate=270,xshift=3pt] at (l2) {Gossypium};
    \node [anchor=west,rotate=270,xshift=3pt] at (l3) {M. truncatula};
    \node [anchor=west,rotate=270,xshift=3pt] at (l4) {Glycine};
    \node [anchor=west,rotate=270,xshift=3pt] at (l5) {Lactuca};
    \node [anchor=west,rotate=270,xshift=3pt] at (l6) {Helianthus};
    \node [anchor=west,rotate=270,xshift=3pt] at (l7) {Solanum};
    \node [anchor=west,rotate=270,xshift=3pt] at (l8) {Lycopersicon};
    \node [anchor=west,rotate=270,xshift=3pt] at (l9) {Hordeum};
    \node [anchor=west,rotate=270,xshift=3pt] at (l10) {Triticum};
    \node [anchor=west,rotate=270,xshift=3pt] at (l11) {Oryza};
    \node [anchor=west,rotate=270,xshift=3pt] at (l12) {Zea};
    \node [anchor=west,rotate=270,xshift=3pt] at (l13) {Saccharum};
    \node [anchor=west,rotate=270,xshift=3pt] at (l14) {Sorghum};
    \coordinate (x01) at (barycentric cs:l0=0.5,l1=0.5);
    \node [node] (p01) at (x01 |- 0,1) {}; \path [edge] (l0) -- (p01) -- (l1);
    \coordinate (x02) at (barycentric cs:p01=0.667,l2=0.333);
    \node [node] (p02) at (x02 |- 0,2) {}; \path [edge] (p01) -- (p02) -- (l2);
    \coordinate (x34) at (barycentric cs:l3=0.5,l4=0.5);
    \node [node] (p34) at (x34 |- 0,1) {}; \path [edge] (l3) -- (p34) -- (l4);
    \coordinate (x04) at (barycentric cs:p02=0.667,p34=0.333);
    \node [node] (p04) at (x04 |- 0,3) {}; \path [edge] (p02) -- (p04) -- (p34);
    \coordinate (x56) at (barycentric cs:l5=0.5,l6=0.5);
    \node [node] (p56) at (x56 |- 0,1) {}; \path [edge] (l5) -- (p56) -- (l6);
    \coordinate (x78) at (barycentric cs:l7=0.5,l8=0.5);
    \node [node] (p78) at (x78 |- 0,1) {}; \path [edge] (l7) -- (p78) -- (l8);
    \coordinate (x58) at (barycentric cs:p56=0.5,p78=0.5);
    \node [node] (p58) at (x58 |- 0,2) {}; \path [edge] (p56) -- (p58) -- (p78);
    \coordinate (x08) at (barycentric cs:p04=0.667,p58=0.333);
    \node [node] (p08) at (x08 |- 0,4) {}; \path [edge] (p04) -- (p08) -- (p58);
    \coordinate (x910) at (barycentric cs:l9=0.5,l10=0.5);
    \node [node] (p910) at (x910 |- 0,1) {}; \path [edge] (l9) -- (p910) -- (l10);
    \coordinate (x1314) at (barycentric cs:l13=0.5,l14=0.5);
    \node [node] (p1314) at (x1314 |- 0,1) {}; \path [edge] (l13) -- (p1314) -- (l14);
    \coordinate (x1214) at (barycentric cs:l12=0.333,p1314=0.667);
    \node [node] (p1214) at (x1214 |- 0,2) {}; \path [edge] (l12) -- (p1214) -- (p1314);
    \coordinate (x1114) at (barycentric cs:l11=0.25,p1214=0.75);
    \node [node] (p1114) at (x1114 |- 0,3) {}; \path [edge] (l11) -- (p1114) -- (p1214);
    \coordinate (x914) at (barycentric cs:p910=0.25,p1114=0.75);
    \node [node] (p914) at (x914 |- 0,4) {}; \path [edge] (p910) -- (p914) -- (p1114);
    \coordinate (x014) at (barycentric cs:p08=0.5,p914=0.5);
    \node [node] (p014) at (x014 |- 0,5) {}; \path [edge] (p08) -- (p014) -- (p914);
    \node [node] (r) at (p014 |- 0,6) {}; \path [edge] (p014) -- (r);
    \node [dup] at (barycentric cs:r=0.5,p014=0.5) {};
    \node [dup] at (barycentric cs:l0=0.5,p01=0.5) {};
    \node [dup] at (barycentric cs:p01=0.5,p02=0.5) {};
    \node [dup] at (barycentric cs:l2=0.333,p02=0.667) {};
    \node [dup] at (barycentric cs:l2=0.667,p02=0.333) {};
    \node [dup] at (barycentric cs:l3=0.5,p34=0.5) {};
    \node [dup] at (barycentric cs:l4=0.5,p34=0.5) {};
    \node [dup] at (barycentric cs:l6=0.5,p56=0.5) {};
    \node [dup] at (barycentric cs:l7=0.5,p78=0.5) {};
    \node [dup] at (barycentric cs:p78=0.5,p58=0.5) {};
    \node [dup] at (barycentric cs:l10=0.333,p910=0.667) {};
    \node [dup] at (barycentric cs:l10=0.667,p910=0.333) {};
    \node [dup] at (barycentric cs:l12=0.5,p1214=0.5) {};
    \node [dup] at (barycentric cs:l13=0.5,p1314=0.5) {};
    \node [dup] at (barycentric cs:p08=0.5,p014=0.5) {};
    \node [dup] at (barycentric cs:p914=0.5,p014=0.5) {};
  \end{tikzpicture}
  \caption{An example of a plant phylogeny published
    in~\cite{adams2005polyploidy}. Red dots indicate known large-scale
    duplication events.}
  \label{fig:WGDExample}
\end{figure}

In their seminal paper \cite{goodman1979fitting}, Goodman et al.\ pioneered the
parsimony approach to reconciling gene trees with species trees. This has
motivated researchers to explore reconciliation through different models
whilst optimizing some measure of the number of duplication events.

The studied problems can be categorized according to how duplication events are
clustered to form duplication episodes and which restrictions are put on the
possible locations of duplications~\cite{paszek2017efficient}. We focus on the
\emph{minimal episodes (ME)} clustering where duplications can be clustered if they
occur on the same branch of the species phylogeny and have no
ancestor-descendant relationship in any gene tree (see
Figure~\ref{fig:DNADuplication}). We believe this way of clustering to be most relevant
since it can cluster duplications that can be part of a single (large-scale)
duplication event. We consider the \emph{unrestricted ME} (called the FHS-model in \cite{paszek2017efficient}),
which does not put any restrictions on the locations of gene duplications.

\begin{figure}[t]
  \hspace*{\stretch{1}}%
  \subcaptionbox{}{%
    \begin{tikzpicture}[
      node/.style={fill=black,circle,inner sep=0pt,minimum size=3pt,outer sep=0pt},
      gene/.style={fill=black!35,rectangle,draw=black,minimum height=0.6cm},
      edge/.style={draw,thick}]
      \path [fill=black!20] (-0.5,0) rectangle (3.5,6);
      \path [draw] (-0.5,0) -- (-0.5,6) (3.5,0) -- (3.5,6);
      \node [node,label=below:$A_{11}$] (a11) at (0,0) {};
      \node [node,label=below:$A_{12}$] (a12) at (1,0) {};
      \node [node,label=below:$A_{21}$] (a21) at (2,0) {};
      \node [node,label=below:$A_{22}$] (a22) at (3,0) {};
      \node [node] (a1) at (0.5,2) {};
      \node [node] (a2) at (2.5,2) {};
      \node [node] (a) at (1.5,4) {};
      \node [node] (r) at (1.5,6) {};
      \path [edge] (a11) -- (a1) -- (a12) (a21) -- (a2) -- (a22)
      (a1) -- (a) -- (a2) (a) -- (r);
      \path [draw] (0.2,1.8) rectangle (2.8,2.2);
      \path [draw] (1.2,3.8) rectangle (1.8,4.2);
      \path (0,-1);
    \end{tikzpicture}}%
  \hspace*{\stretch{1}}%
  \subcaptionbox{}{%
    \begin{tikzpicture}[
      node/.style={fill=black,circle,inner sep=0pt,minimum size=3pt,outer sep=0pt},
      gene/.style={fill=black!35,rectangle,draw=black,minimum height=0.6cm},
      edge/.style={draw,thick}]
      \begin{scope}[local bounding box=bb1]
        \node [gene] (g11) at (0,0) {\parbox{1.25cm}{\raggedright$A_{11}$}}; \node [node] (a11) at (0,0) {};
        \node [gene] (g21) at (2,0) {\parbox{1.25cm}{\raggedright$A_{21}$}}; \node [node] (a21) at (2,0) {};
      \end{scope}
      \begin{scope}[on background layer]
        \coordinate [xshift=-0.2cm,yshift=0.2cm] (nw) at (bb1.north west);
        \coordinate [xshift=0.2cm,yshift=-0.2cm] (se) at (bb1.south east);
        \path [draw=black,fill=black!20] (nw) rectangle (se);
        \node [anchor=north] at (1,0 |- se) {$B_2$};
      \end{scope}
      \begin{scope}[local bounding box=bb2]
        \node [gene] (g12) at (4.5,0) {\parbox{1.3cm}{\raggedleft$A_{12}$}}; \node [node] (a12) at (4.5,0) {};
        \node [gene] (g22) at (6.5,0) {\parbox{1.3cm}{\raggedleft$A_{22}$}}; \node [node] (a22) at (6.5,0) {};
      \end{scope}
      \begin{scope}[on background layer]
        \coordinate [xshift=-0.2cm,yshift=0.2cm] (nw) at (bb2.north west);
        \coordinate [xshift=0.2cm,yshift=-0.2cm] (se) at (bb2.south east);
        \path [draw=black,fill=black!20] (nw) rectangle (se);
        \node [anchor=north] at (5.5,0 |- se) {$B_2$};
      \end{scope}
      \begin{scope}[local bounding box=bb3]
        \node [gene] (g1) at (2.25,3) {\parbox{1.3cm}{\raggedright$A_1$}}; \node [node] (a1) at (2.25,3) {};
        \node [gene] (g2) at (4.25,3) {\parbox{1.3cm}{\raggedleft$A_2$}}; \node [node] (a2) at (4.25,3) {};
      \end{scope}
      \begin{scope}[on background layer]
        \coordinate [xshift=-0.2cm,yshift=0.2cm] (nw) at (bb3.north west);
        \coordinate [xshift=0.2cm,yshift=-0.2cm] (se) at (bb3.south east);
        \path [draw=black,fill=black!20] (nw) rectangle (se);
        \node [anchor=north] at (3.25,0 |- se) {$B$};
      \end{scope}
      \node [gene] (g) at (3.25,6) {\parbox{1.3cm}{\raggedright$A$}}; \node [node] (a) at (3.25,6) {};
      \coordinate (y1) at (barycentric cs:a=0.4,a1=0.6);
      \node [node] (v1) at (y1 -| a) {};
      \node [node] (v2) at (barycentric cs:a1=1,v1=-0.5) {};
      \node [node] (v3) at (barycentric cs:a2=1,v1=-0.5) {};
      \path [edge] (a) -- (v1) -- (v2) -- (a11) (v2) -- (a12) (v1) -- (v3) -- (a21) (v3) -- (a22);
      \path [draw] (g11) +(-1.5,0) -- (g11) -- (g21) -- (g12) -- (g22) -- +(1.5,0);
      \path [draw] (g1) +(-1.5,0) -- (g1) -- (g2) -- +(1.5,0);
      \path [draw] (a) +(-1.5,0) -- (g) -- +(1.5,0);
    \end{tikzpicture}}%
  \hspace*{\stretch{1}}%
  \caption{(a) A gene tree embedded into a (branch of) a species tree with
    duplications clustered as in ME clustering. Duplication clusters are shown
    as rectangles. (b) A representation of the DNA of the species at
    different points in the species tree (at corresponding heights). In the
    first duplication, the gene $A$ (dark rectangle) is duplicated, forming
    $A_1$ and $A_2$. In the second duplication, the block $B$ (light rectangle)
    comprising $A_1$ and $A_2$ is duplicated. This results in four homologous
    copies of gene $A$ using only two duplication episodes. The gene tree is
    also drawn through the depictions of the DNA.}
  \label{fig:DNADuplication}
\end{figure}
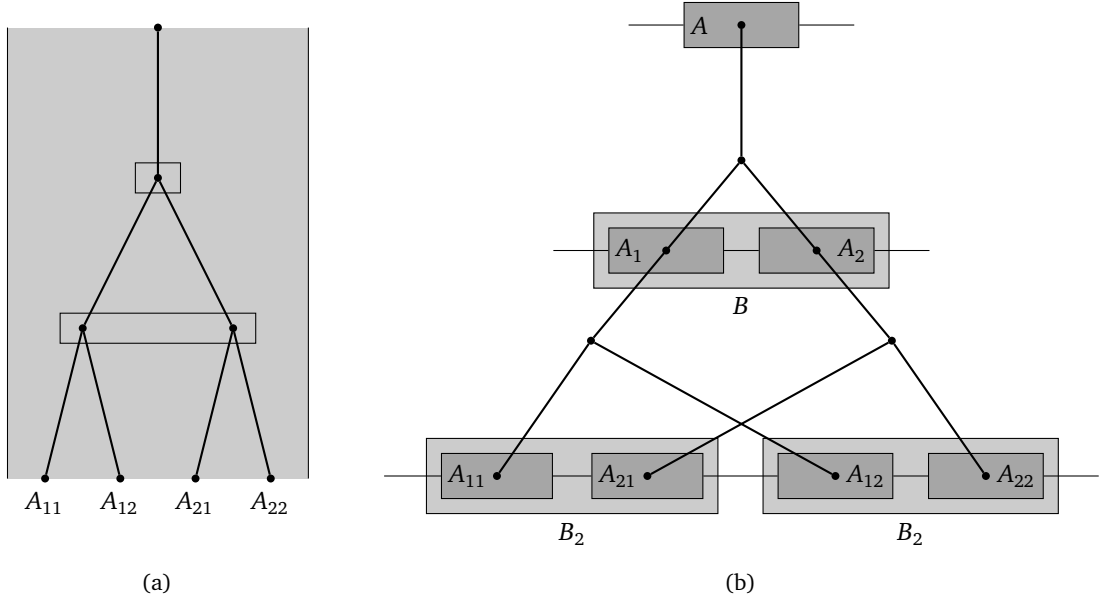

Reconciliation problems have been studied intensively, especially 
without clustering of duplication
events~\cite{page1994maps,page1997gene,doyon2009space}. Several reconciliation
problems with clustering of duplication events have been proven to be
computationally intractable \cite{fellows1998multiple,ma2000gene}, whereas for
others there are polynomial-time~\cite{bansal2008multiple,burleigh2008locating}
or even
linear-time~\cite{mettanant2008linear,luo2011linear,paszek2017efficient}
algorithms. For unrestricted ME reconciliation, which was recently shown to be NP-hard~\cite{dondi2019reconciling},
there only exists an
exponential-time algorithm~\cite{paszek2017efficient}.

It has also been attempted to use reconciliation as a basis for inferring
species phylogenies. For the unrestricted ME inference problem,
\cite{burleigh2010inferring} used a brute-force approach on all possible
species phylogenies. It was observed that unrestricted ME fails to
rank the true species tree among the top third of all topologies (for real data with a well accepted species phylogeny). It was
suggested that a possible reason for this anomaly is that duplication episodes
near the root are overly powerful under this criterion. A similar observation was
made in a more recent reconciliation study~\cite{paszek2017efficient}. However,
neither article gives a mathematical explanation for this phenomenon. It should
also be noted that, since the number of possible species phylogenies grows
extremely quickly with the number of species, brute-force approaches are only
feasible for very small data sets.

Inference problems are generally assumed to be computationally intractable.
However, NP-hardness has been proven only for some restricted inference problem
without clustering of duplication events~\cite{ma2000gene}. For an inference
problem with restricted clustering (called gene duplication (GD) clustering
in~\cite{paszek2017efficient}), NP-hardness was suggested
in~\cite{fellows1998multiple} but not proven. Because of the suspected
intractability of these problems, some heuristic inference approaches have been
attempted using efficient algorithms for reconciliation (see,
e.g.,~\cite{guigo1996reconstruction}).

\paragraph{Reticulation minimization problems.}

Another possible cause of discordance between gene trees is \emph{reticulate
  evolution}, such as hybridization or horizontal gene transfer.  In such
cases, the evolutionary history is represented by a \emph{phylogenetic network}
rather than a tree. 

Reticulate evolution can occur in nature when genetic material from one species
is transmitted to some other species. In asexual species, such transfers are
called \emph{horizontal gene transfers (HGT)}. In bacteria, for example, this
happens in nature by transformation (take-up from the environment) or
conjugation (transmission from another bacterium). In sexual species, a cause
for such transmissions can be \emph{hybridization}, where individuals from
different but related taxa mate. There is also evidence that horizontal gene
transfers occur between multicellular sexual species.  An example is the
transfer of a phototropin gene from Hornworts to Ferns
(Figure~\ref{fig:NetworkExample}). HGT can even happen between more distant
species.

\begin{figure}[t]
  \hspace*{\stretch{1}}%
  \subcaptionbox{}{%
    \begin{tikzpicture}[
      node/.style={fill=black,circle,inner sep=0pt,minimum size=3pt,outer sep=0pt},
      wide node/.style={fill=black,rectangle,inner sep=0pt,minimum height=3pt,rounded corners=1.5pt},
      edge/.style={draw,thick}]
      \foreach \i in {0,...,6} {
        \node [node] (l\i) at (\i,0) {};
      }
      \coordinate (x12) at (barycentric cs:l1=0.5,l2=0.5);
      \node [node] (p12) at (x12 |- 0,1) {}; \path [edge] (l1) -- (p12) -- (l2);
      \coordinate (x02) at (barycentric cs:l0=0.333,p12=0.667);
      \node [node] (p02) at (x02 |- 0,2) {}; \path [edge] (l0) -- (p02) -- (p12);
      \coordinate (x34) at (barycentric cs:l3=0.5,l4=0.5);
      \node [node] (p34) at (x34 |- 0,1) {}; \path [edge] (l3) -- (p34) -- (l4);
      \coordinate (x56) at (barycentric cs:l5=0.5,l6=0.5);
      \node [node] (p56) at (x56 |- 0,1) {}; \path [edge] (l5) -- (p56) -- (l6);
      \coordinate (x36) at (barycentric cs:p34=0.5,p56=0.5);
      \node [node] (p36) at (x36 |- 0,2) {}; \path [edge] (p34) -- (p36) -- (p56);
      \coordinate (x06) at (barycentric cs:p02=0.5,p36=0.5);
      \node [node] (p06) at (x06 |- 0,3) {}; \path [edge] (p02) -- (p06) -- (p36);
      \node [node] (r) at (p06 |- 0,4) {}; \path [edge] (p06) -- (r);
      \node [anchor=west,rotate=270,xshift=3pt] at (l0) {H PHOT};
      \node [anchor=west,rotate=270,xshift=3pt] at (l1) {H NEO};
      \node [anchor=west,rotate=270,xshift=3pt] at (l2) {F NEO};
      \node [anchor=west,rotate=270,xshift=3pt] at (l3) {F PHOT1};
      \node [anchor=west,rotate=270,xshift=3pt] at (l4) {F PHOT2};
      \node [anchor=west,rotate=270,xshift=3pt] at (l5) {S PHOT1};
      \node [anchor=west,rotate=270,xshift=3pt] at (l6) {S PHOT2};
    \end{tikzpicture}}%
  \hspace*{\stretch{1}}%
  \subcaptionbox{}{%
    \begin{tikzpicture}[
      node/.style={fill=black,circle,inner sep=0pt,minimum size=3pt,outer sep=0pt},
      wide node/.style={fill=black,rectangle,inner sep=0pt,minimum height=3pt,rounded corners=1.5pt},
      edge/.style={draw,thick}]
      \coordinate (l12) at (0,0);
      \coordinate (l22) at (2,0);
      \coordinate (l32) at (4,0);
      \path (l12) +(60:0.6) coordinate (p12);
      \path (l12) +(60:2.3) coordinate (pp12);
      \path (l12) +(60:4) coordinate (p132);
      \path (l22) +(60:0.6) coordinate (p22);
      \path (l22) +(60:1.1) coordinate (pp22);
      \path (l22) +(60:2) coordinate (p232);
      \path (l32) +(120:1) coordinate (p32);
      \path (p132) +(90:1) coordinate (r2);
      \foreach \r in {r,l1,p1,pp1,l2,p2,pp2,l3,p3,p23,p13} {
        \path (\r 2)
        +(180:0.125) coordinate (\r 1)
        +(0:0.125) coordinate (\r 3)
        +(180:0.25) coordinate (\r l)
        +(0:0.25) coordinate (\r r);
      }
      \coordinate (pp2m) at (barycentric cs:pp22=0.5,pp23=0.5);
      \path [name path=path1] (p132) -- (l32);
      \path [name path=path2] (pp2m) -- +(60:1);
      \path [name intersections={of=path1 and path2}] coordinate (p23m) at (intersection-1);
      \path [name path=path1] (l1r) -- +(60:4);
      \path [name path=path2] (l3l) -- +(120:4);
      \path [name intersections={of=path1 and path2}] coordinate (p13b) at (intersection-1);
      \path [name path=path1] (l2r) -- +(60:2);
      \path [name intersections={of=path1 and path2}] coordinate (p23b) at (intersection-1);
      \path [name path=path1] (p13) ++(270:0.2) -- +(0:2);
      \path [name path=path2] (l1r) -- +(60:1);
      \path [name intersections={of=path1 and path2}] coordinate (p1br) at (intersection-1);
      \path [name path=path2] (l2l) -- +(60:1);
      \path [name intersections={of=path1 and path2}] coordinate (p2bl) at (intersection-1);
      \path [name path=path1] (p13) ++(90:0.2) -- +(0:2);
      \path [name path=path2] (l1r) -- +(60:1);
      \path [name intersections={of=path1 and path2}] coordinate (p1tr) at (intersection-1);
      \path [name path=path2] (l2l) -- +(60:1);
      \path [name intersections={of=path1 and path2}] coordinate (p2tl) at (intersection-1);
      \path [fill=black!20] (rl) -- (p13l) -- (l1l) -- (l1r) -- (p1br) -- (p2bl) -- (l2l)
      -- (l2r) -- (p23b) -- (l3l) -- (l3r) -- (p13r) -- (rr) -- cycle;
      \path [draw] (rl) -- (p13l) -- (l1l) (rr) -- (p13r) -- (l3r) (l2r) -- (p23b) -- (l3l)
      (l1r) -- (p1br) -- (p2bl) -- (l2l);
      \path [draw,fill=white] (p1tr) -- (p2tl) -- (p23l) -- (p13b) -- cycle;
      \path [edge] (r2) -- (p132) -- (p32) (p31) -- (l31) (p33) -- (l33)
      (p132) -- (pp12) (pp11) -- (l11) (pp13) -- (l13) (p13) -- (p21) -- (l21)
      (pp22) -- (l22) (pp23) -- (l23) (pp2m) -- (p23m);
      \node [node] at (r2) {};
      \node [node] at (p132) {};
      \node [node] at (l11) {};
      \node [node] at (l13) {};
      \node [node] at (p13) {};
      \node [node] at (l21) {};
      \node [node] at (l22) {};
      \node [node] at (l23) {};
      \node [node] at (l31) {};
      \node [node] at (l33) {};
      \node [node] at (p23m) {};
      \node [wide node,minimum width=0.35cm] at (pp12) {};
      \node [wide node,minimum width=0.225cm] at (pp2m) {};
      \node [wide node,minimum width=0.35cm] at (p32) {};
      \node [anchor=north] at (l12) {Hornwort (H)};
      \node [anchor=north] at (l22) {Ferns (F)};
      \node [anchor=north] at (l32) {Seedplants (S)};
      \path (0,-1);
    \end{tikzpicture}}%
  \hspace*{\stretch{1}}%
  \caption{A simplified real example of a gene tree in which multiple
    complicating factors play a role, and a likely reconciliation with the
    phylogeny \cite{li2014horizontal,li2015origin}. (a) A gene tree for the
    phototropin gene in plants. (b) the phylogeny for those plant groups.
    Note that the horizontal branch of the phylogeny from Hornworts to Ferns
    represents a horizontal gene transfer.}
  \label{fig:NetworkExample}
\end{figure}
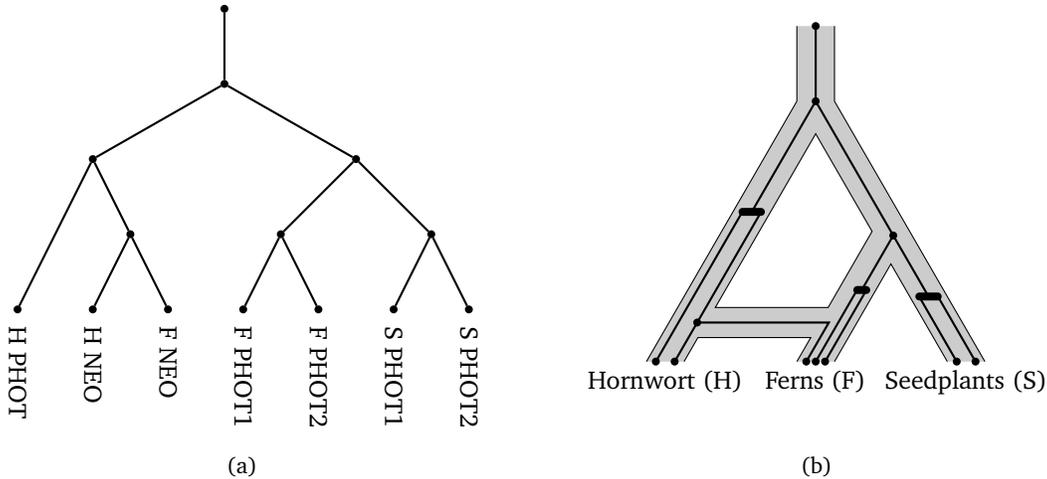

Gene trees that appear to be inconsistent may in fact simply take different
paths through the network.  This leads to a family of inference problems in
which the aim is to find a phylogenetic network that is consistent with the
gene trees and has the minimum number of \emph{reticulation events} (nodes in
the network with two ancestral branches).  A phylogenetic network is often
taken to be consistent with a gene tree if that tree is \emph{displayed} by the
network, which, roughly speaking, means that the gene tree can be drawn inside
the network in such a way that each network branch contains at most one lineage
of the gene tree.  A more generous definition is to count a network as
consistent with a gene tree if the tree is \emph{weakly displayed} by the
network \cite{huber2016folding,zhu2016light}.  Roughly speaking, this means
that different lineages of the gene tree may ``travel down'' the same branch of
the network, as long as any branching node in the tree coincides with a
branching node in the network.  In this case, the tree is also called a
\emph{parental tree} of the network.  This models situations where a species
has individuals carrying multiple homologous copies of a gene.

The \textsc{Hybridization Number} problem, in which we seek a network with the
minimum number of reticulations displaying all input trees, has been
well-studied. It has been shown that \textsc{Hybridization Number} is NP-hard
already when the input consists of only two gene
trees~\cite{bordewich2007computing}. Furthermore, there are theoretical FPT
algorithms for any fixed number of gene trees, but there are no practical
algorithms that can handle instances with more than two input trees unless the
number of taxa is extremely small
\cite{bordewich2007reduction,van2016hybridization}.

In contrast, the \textsc{Parental Hybridization} problem, in which we seek a
network with the minimum number of reticulations that weakly displays each
input tree, was introduced only recently~\cite{zhu2016light} and its
computational complexity was open prior to this article.  Our motivation for
studying this problem is threefold:
\begin{enumerate}[label=(\roman{*})]
  \item Since \textsc{Hybridization Number} is NP-hard, it is interesting
    whether relaxing the notion of a tree displayed by a network leads to an
    easier problem.
  \item Since reticulation can lead to multiple homologous copies of a gene in
    a species, requiring that each gene tree is displayed by the network may
    lead us to overestimate the number of reticulations.
  \item The problem of finding an optimal network that weakly displays a set of
    phylogenies arises as a crucial subproblem in a recent heuristic approach
    for constructing phylogenetic networks in the presence of hybridization and
    incomplete lineage sorting \cite{zhu2016light}.
\end{enumerate}

\paragraph{Structural assumptions.}

In this paper, as is common in the literature, we assume that all networks and
trees are \emph{binary}, that is, every node except the root and the leaves has
total degree exactly~$3$. Our results should easily generalize to nonbinary
trees and networks, but we do not verify this here.

We note that, unlike many papers in this area, we allow a network to contain
\emph{parallel arcs}, that is, pairs of arcs that join the same pair of nodes.
Parallel arcs are normally omitted because, for most problems, it can either be
shown that there exists an optimal solution without parallel arcs or it can be
assumed that a realistic solution contains no parallel arcs.  For example, any
set of gene trees is displayed by an optimal hybridization network without
parallel arcs.  For the problems studied in this paper, however, an optimal
solution may require parallel arcs.  Considering this problem with the added
restriction that parallel arcs are forbidden may be an interesting mathematical
challenge; however, we do not believe it is biologically meaningful.

Explicit reasons to allow parallel arcs in networks are abundant. We give
three: First, if one restricts a large network to a subset of the taxa, the
natural restriction could have parallel arcs. Second, phylogenetic Markov
models for character evolution behave differently if parallel arcs are
suppressed. Third, polyploidization events often result from a sort of
interspecific or intraspecific hybridization \cite{albertin2012polyploidy}; an
intraspecific hybridization is most naturally represented by parallel arcs in
the network.

Throughout this paper, we allow input trees to be multi-labelled, that is, each
species may appear as a label of multiple leaves in a tree.  This is natural
for the problems we study, as gene duplication and reticulation can both lead
to multiple homologous genes appearing in the genome of a single species.

\paragraph{Our contributions.}

We show that both \textsc{Unrestricted Minimal Episodes Inference} and
\textsc{Parental Hybridization} reduce to the problem \textsc{Beaded Tree},
which we introduce in this paper.  Using this reduction, we show that both
problems can be solved in polynomial time by adapting Aho et al.'s classic
algorithm for testing gene tree consistency~\cite{Aho1981InferringAT}.
Thereby, we provide the first polynomial-time algorithm for an inference
problem with duplication clustering. Furthermore, we provide the first
polynomial-time algorithm for constructing a phylogenetic \emph{network} with a minimum number of reticulations from
gene trees.

We also show that optimal solutions to \textsc{Beaded Tree} have a restricted
structure and this has corresponding implications for the optimal solutions to
\textsc{Unrestricted Minimal Episodes Inference} and \textsc{Parental
  Hybridization} that our algorithms produce. Moreover, we show that, in fact,
\emph{all} optimal solutions to \textsc{Unrestricted Minimal Episodes
  Inference} have a particular structure. Therefore, this problem should be used
with care. For this reason, we introduce a variation of \textsc{Unrestricted
  Minimal Episodes Inference}, in which the aim is not to minimize the total
number of duplication episodes but to minimize instead the maximum number of
duplication episodes on any path from the root to a leaf in the output tree.
We show that this problem can also be solved in polynomial time via reduction
to a variant of \textsc{Beaded Tree}, which we call \textsc{Beaded Tree Depth}.

\paragraph{Structure of the paper.}

In Section~\ref{sec:prelim}, we introduce the main definitions, including
formal problem definitions.  In Section~\ref{sec:reduction}, we show that both
\textsc{Unrestricted Minimal Episodes Inference} and \textsc{Parental
  Hybridization} reduce to the problem \textsc{Beaded Tree}.  In
Section~\ref{sec:structure}, we prove structural properties of optimal
solutions to \textsc{Beaded Tree}.  In Section~\ref{sec:algorithm}, we provide
a polynomial-time algorithm for \textsc{Beaded Tree} and prove its correctness
and running time.  In Section~\ref{sec:depth}, we provide a polynomial-time
algorithm for \textsc{Beaded Tree Depth}.  Finally, in
Section~\ref{sec:conclusion}, we discuss our results and possibilities for
further research.

\section{Preliminaries and Definitions}\label{sec:prelim}

We begin by defining \emph{multi-labelled trees}, which form the input for all
problems considered in this paper.

\begin{definition}
  Let $X$ be a set of species.  A \emph{multi-labelled tree (MUL-tree)} on $X$
  is a directed acyclic graph with one node of in-degree 0 and out-degree 1
  (the \emph{root}) and with all other nodes having either in-degree 1 and
  out-degree~2 (\emph{tree nodes}) or in-degree 1 and out-degree 0
  (\emph{leaves}).  Each leaf is labelled with an element of $X$.  If each
  element of $X$ labels at most one leaf, we call the MUL-tree a \emph{tree}.
\end{definition}

Note that we will often refer to a labelled node by its label; for example, we
may say that $x \in X$ is a leaf in a MUL-tree $T$ if one of the leaves of $T$
is labelled with $x$.

The notation introduced in the next definition is common to all structures
considered in this paper, that is, not just to MUL-trees but also to
duplication trees, phylogenetic networks, and beaded trees, defined later in
this section.

\begin{definition}
  Given a directed acyclic graph $G$, let $V(G)$ denote the nodes, and $E(G)$
  the edges of $G$.  Let $L(G)$ denote the leaves (i.e., nodes of out-degree
  $0$) of~$G$.  We refer to the non-leaf nodes of $G$ as the \emph{internal
    nodes} of~$G$.  Given an edge $xy$ in $G$, we say that $x$ is a
  \emph{parent} of $y$ and $y$ is a \emph{child} of $x$.  We say a node $x$ is
  an \emph{ancestor} of a node $y$ (and $y$ is a \emph{descendant} of $x$) if
  there is a path from $x$ to $y$ in $G$ (including if $x=y$). If in addition
  $x \neq y$, we say $x$ is a \emph{strict} ancestor of $y$ (and $y$ is a
  \emph{strict} descendant of $x$).  A node $x$ is a \emph{least common
    ancestor} of two nodes $y$ and $z$ if it is an ancestor of both $y$ and $z$
  and no strict descendant of $x$ is an ancestor of both $y$ and $z$.  If $G$
  is a tree, then the LCA of any two nodes is unique; otherwise, it may not be
  unique.
\end{definition}

\subsection{Duplication Episodes}

The evolutionary history of a set of species, including points at which
duplication events occurred, can be modelled by a duplication tree, defined as
follows:

\begin{definition}
  Let $X$ be a set of species.  A \emph{duplication tree} on $X$ is a directed
  acyclic graph $D$ with one node of in-degree 0 and out-degree 1 (the
  \emph{root}), $|X|$ nodes of in-degree 1 and out-degree 0 (\emph{leaves}),
  and all other nodes having either in-degree 1 and out-degree 2 (\emph{tree
    nodes}) or in-degree 1 and out-degree 1 (\emph{duplication nodes}).  The
  leaves are bijectively labelled with the elements of $X$.  The
  \emph{duplication number} of $D$ is the number of duplication nodes it
  contains.
\end{definition}

We note that, in contrast to MUL-trees, each species in $X$ appears as the
label of exactly one leaf in a duplication tree.  Informally, a
  MUL-tree~$T$ is \emph{consistent} with a duplication tree~$D$ if~$T$
  can be drawn inside~$D$ so that branches of $T$ duplicate only at duplication
  nodes of $D$, in the sense that both out-edges of a node of~$T$ may follow
  the same out-edge of the duplication node (see
  Figure~\ref{fig:duplication_tree}). We formalize this as follows:

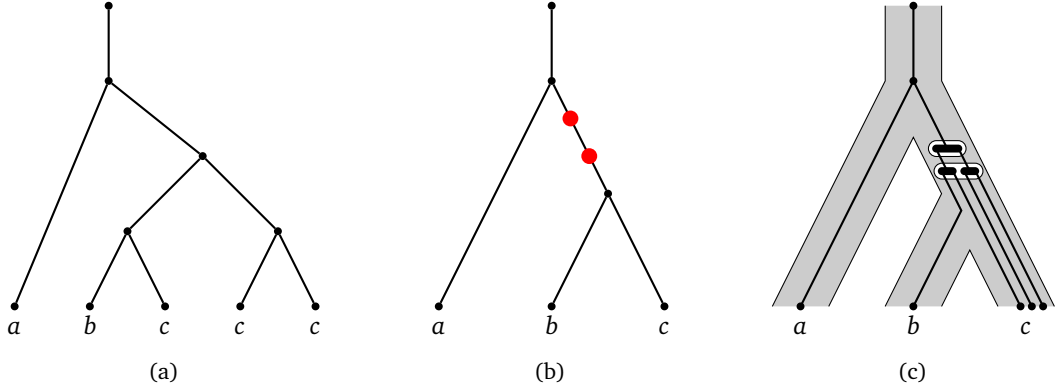
\begin{figure}[t]
  \hspace*{\stretch{1}}%
  \subcaptionbox{}{%
    \begin{tikzpicture}[
      node/.style={fill=black,circle,inner sep=0pt,minimum size=3pt,outer sep=0pt},
      dup/.style={fill=red,circle,inner sep=0pt,minimum size=6pt},
      wide node/.style={fill=black,rectangle,inner sep=0pt,minimum height=3pt,rounded corners=1.5pt},
      episode/.style={draw=black,fill=white,rectangle,inner sep=0pt,minimum height=6pt,rounded corners=3pt},
      edge/.style={draw,thick}]
      \foreach \i in {0,...,4} {
        \node [node] (l\i) at (\i,0) {};
      }
      \coordinate (x12) at (barycentric cs:l1=0.5,l2=0.5);
      \node [node] (p12) at (x12 |- 0,1) {}; \path [edge] (l1) -- (p12) -- (l2);
      \coordinate (x34) at (barycentric cs:l3=0.5,l4=0.5);
      \node [node] (p34) at (x34 |- 0,1) {}; \path [edge] (l3) -- (p34) -- (l4);
      \coordinate (x14) at (barycentric cs:p12=0.5,p34=0.5);
      \node [node] (p14) at (x14 |- 0,2) {}; \path [edge] (p12) -- (p14) -- (p34);
      \coordinate (x04) at (barycentric cs:l0=0.5,p14=0.5);
      \node [node] (p04) at (x04 |- 0,3) {}; \path [edge] (l0) -- (p04) -- (p14);
      \path (p04) +(90:1) node [node] (r) {}; \path [edge] (r) -- (p04);
      \node [anchor=north,text height=height("$b$")] at (l0) {$a$};
      \node [anchor=north,text height=height("$b$")] at (l1) {$b$};
      \node [anchor=north,text height=height("$b$")] at (l2) {$c$};
      \node [anchor=north,text height=height("$b$")] at (l3) {$c$};
      \node [anchor=north,text height=height("$b$")] at (l4) {$c$};
    \end{tikzpicture}}%
  \hspace*{\stretch{1}}%
  \subcaptionbox{}{%
    \begin{tikzpicture}[x=1.5cm,y=1.5cm,
      node/.style={fill=black,circle,inner sep=0pt,minimum size=3pt,outer sep=0pt},
      dup/.style={fill=red,circle,inner sep=0pt,minimum size=6pt},
      wide node/.style={fill=black,rectangle,inner sep=0pt,minimum height=3pt,rounded corners=1.5pt},
      episode/.style={draw=black,fill=white,rectangle,inner sep=0pt,minimum height=6pt,rounded corners=3pt},
      edge/.style={draw,thick}]
      \foreach \i in {0,1,2} {
        \node [node] (l\i) at (\i,0) {};
      }
      \coordinate (x12) at (barycentric cs:l1=0.5,l2=0.5);
      \node [node] (p12) at (x12 |- 0,1) {}; \path [edge] (l1) -- (p12) -- (l2);
      \coordinate (x02) at (barycentric cs:l0=0.333,p12=0.667);
      \node [node] (p02) at (x02 |- 0,2) {}; \path [edge] (l0) -- (p02) -- (p12);
      \path (p02) +(90:0.667) node [node] (r) {}; \path [edge] (r) -- (p02);
      \node [dup] at (barycentric cs:p02=0.333,p12=0.667) {};
      \node [dup] at (barycentric cs:p02=0.667,p12=0.333) {};
      \node [anchor=north,text height=height("$b$")] at (l0) {$a$};
      \node [anchor=north,text height=height("$b$")] at (l1) {$b$};
      \node [anchor=north,text height=height("$b$")] at (l2) {$c$};
    \end{tikzpicture}}%
  \hspace*{\stretch{1}}%
  \subcaptionbox{}{%
    \begin{tikzpicture}[x=1.5cm,y=1.5cm,
      node/.style={fill=black,circle,inner sep=0pt,minimum size=3pt,outer sep=0pt},
      dup/.style={fill=red,circle,inner sep=0pt,minimum size=6pt},
      wide node/.style={fill=black,rectangle,inner sep=0pt,minimum height=3pt,rounded corners=1.5pt},
      episode/.style={draw=black,fill=white,rectangle,inner sep=0pt,minimum height=6pt,rounded corners=3pt},
      edge/.style={draw,thick}]
      \foreach \i in {0,...,2} {
        \coordinate (l\i) at (\i,0);
      }
      \coordinate (x12) at (barycentric cs:l1=0.5,l2=0.5);
      \coordinate (p12) at (x12 |- 0,1);
      \coordinate (x02) at (barycentric cs:l0=0.333,p12=0.667);
      \coordinate (p02) at (x02 |- 0,2);
      \path (p02) +(90:0.667) coordinate (r);
      \coordinate (pp12) at (barycentric cs:p12=0.8,p02=0.2);
      \coordinate (ppp12) at (barycentric cs:p12=0.6,p02=0.4);
      \foreach \v in {l0,l1,l2,p12,pp12,ppp12,p02,r} {
        \path (\v)
        +(180:0.05) coordinate (\v 2)
        +(180:0.15) coordinate (\v 1)
        +(180:0.25) coordinate (\v l)
        +(0:0.05) coordinate (\v 3)
        +(0:0.15) coordinate (\v 4)
        +(0:0.25) coordinate (\v r);
      }
      \path [name path=path1] (l0r) -- +(1,2);
      \path [name path=path2] (l2l) -- +(-1,2);
      \path [name intersections={of=path1 and path2}] coordinate (p02b) at (intersection-1);
      \path [name path=path1] (l1r) -- +(0.5,1);
      \path [name path=path2] (l2l) -- +(-0.5,1);
      \path [name intersections={of=path1 and path2}] coordinate (p12b) at (intersection-1);
      \path [name path=path1] (l1) -- +(0.5,1);
      \path [name path=path2] (l21) -- +(-0.5,1);
      \path [name intersections={of=path1 and path2}] coordinate (p1) at (intersection-1);
      \coordinate (ppp12a) at (barycentric cs:ppp121=0.5,ppp122=0.5);
      \coordinate (ppp12b) at (barycentric cs:ppp123=0.5,ppp124=0.5);
      \coordinate (pp12a) at (barycentric cs:pp121=0.5,pp122=0.5);
      \coordinate (pp12b) at (barycentric cs:pp123=0.5,pp124=0.5);
      \path [fill=black!20] (rl) -- (p02l) -- (l0l) -- (l0r) -- (p02b) -- (p12l) --
      (l1l) -- (l1r) -- (p12b) -- (l2l) -- (l2r) -- (p02r) -- (rr) -- cycle;
      \node [episode,minimum width=0.5cm] at (ppp12) {};
      \node [episode,minimum width=0.65cm] at (pp12) {};
      \path [draw] (rl) -- (p02l) -- (l0l) (rr) -- (p02r) -- (l2r)
      (l0r) -- (p02b) -- (p12l) -- (l1l) (l1r) -- (p12b) -- (l2l);
      \path [edge] (r) -- (p02) (l0) -- (p02) -- (ppp12) (pp122) -- (l22) (pp123) -- (l23)
      (pp124) -- (l24) (pp121) -- (p1) -- (l1) (ppp12a) -- (pp12a) (ppp12b) -- (pp12b);
      \node [node] at (l0) {};
      \node [node] at (l1) {};
      \node [node] at (l22) {};
      \node [node] at (l23) {};
      \node [node] at (l24) {};
      \node [node] at (p02) {};
      \node [node] at (r) {};
      \node [wide node,minimum width=0.25cm] at (pp12a) {};
      \node [wide node,minimum width=0.25cm] at (pp12b) {};
      \node [wide node,minimum width=0.4cm] at (ppp12) {};
      \node [anchor=north,text height=height("$b$")] at (l0) {$a$};
      \node [anchor=north,text height=height("$b$")] at (l1) {$b$};
      \node [anchor=north,text height=height("$b$")] at (l2) {$c$};
    \end{tikzpicture}}%
  \hspace*{\stretch{1}}%
  \caption{(a) A MUL-tree $T$ on $X = \{a, b, c\}$. (b) A duplication tree $D$
    that is consistent with $T$. (c) An illustration showing how $T$ can be
    drawn inside $D$. This shows how two or more incoming branches may
    duplicate simultaneously at a duplication node (according to the Minimal
    Episodes clustering).}
  \label{fig:duplication_tree}
\end{figure}

\begin{definition}
 Given a MUL-tree $T$ on $X$ and a duplication tree $D$ on $X$, a
 \emph{duplication mapping} from $T$ to $D$ is a function $M:V(T) \rightarrow
 V(D)$ such that
 \begin{itemize}
   \item For each leaf $l \in L(T)$, $M(l)$ is a leaf of $D$ labelled with the
     same species as $l$,
   \item For each edge $uv \in E(T)$, $M(u)$ is a strict ancestor of $M(v)$,
     and
   \item For each internal node $u$ of $T$ with children $v, v'$, either $M(u)$
     is the least common ancestor of $M(v)$ and $M(v')$, or $M(u)$ is a
     duplication node.
 \end{itemize}
 This is illustrated in Figure~\ref{fig:duplication_tree}.
 We say that $D$ \emph{is consistent with} $T$ if there is a duplication
 mapping from $T$ to~$D$.
\end{definition}

Let $S$ be the species tree derived from $D$ by suppressing duplication nodes.
Then a duplication mapping from $T$ to $D$ represents a reconciliation of $T$
with $S$ with Minimal Episodes clustering.  Each duplication node in $D$
represents a cluster of duplications, which is called a \emph{duplication
  episode}.  Internal nodes in $T$ are treated as duplications if they are
mapped to duplication nodes of $D$, and as speciations otherwise.  Duplications
are clustered together if they are mapped to the same duplication node of $D$.
The properties of a duplication tree and duplication mapping ensure that
duplications that are clustered occur on the same branch of the species
phylogeny and have no ancestor-descendant relationship in a gene tree, as
required by Minimal Episodes clustering.  We are now ready to define the
following problem:

\medskip

\noindent \textsc{Unrestricted Minimal Episodes Inference}\\
\noindent \textbf{Input}: A set $\mathcal{T} = \{T_1, \ldots, T_t\}$ of
MUL-trees with label sets $X_1, \ldots, X_t \subseteq X$.\\
\noindent \textbf{Output}:  A duplication tree $D$ on $X$ with minimum
duplication number and such that $D$ is consistent with each tree
in~$\mathcal{T}$.

\medskip

For this and other optimization problems, we use the term \emph{solution} to refer to an object that satisfies the requirements specified in the description of the output except that it does not necessarily need to optimize the optimization criterion. An \emph{optimal solution} is a solution that optimizes the optimization criterion. For example, for \textsc{Unrestricted Minimal Episodes Inference}, a solution is a duplication tree on $X$ that is consistent with each tree in~$\mathcal{T}$. It is an optimal solution if, in addition, it has minimum duplication number over all such duplication trees.

\medskip

We note that for any MUL-tree $T$ on $X$ and any duplication tree $D$ on $X$
that has at least $|V(T)|$ duplication nodes as ancestors of every tree node,
$D$ is consistent with $T$. It follows that every instance of
\textsc{Unrestricted Minimal Episodes Inference} has a solution (and therefore an optimal solution).

\subsection{Parental Hybridization}

\emph{Phylogenetic networks} are an appropriate mathematical model used for describing evolutionary histories
that include reticulation events and are central to the problem
\textsc{Parental Hybridization}, defined below.

\begin{definition}
  Let $X$ be a set of species.  A \emph{(rooted binary) phylogenetic network}
  $N$ on $X$ is a directed acyclic multigraph with one node of in-degree 0 and
  out-degree 1 (the \emph{root}), $|X|$ nodes of in-degree 1 and out-degree~0
  (\emph{leaves}), and all other nodes having either in-degree 1 and out-degree
  2 (\emph{tree nodes}) or in-degree~2 and out-degree 1 (\emph{reticulation
    nodes}).  The leaves are bijectively labelled with the elements of $X$.
  The \emph{reticulation number} of $N$ is the number of reticulation nodes it
  contains.  If $N$ contains no reticulation nodes, then $N$ is a \emph{tree}.
\end{definition}

We note that the key distinctions between a phylogenetic network and a MUL-tree
are that a phylogenetic network may contain reticulation nodes but each label
in $X$ may appear only once, whereas a MUL-tree has no reticulations but each
label can appear multiple times. Also note that, due to the degree
restrictions, there can be at most two edges between any pair of nodes in a
phylogenetic network, and there are no loops.

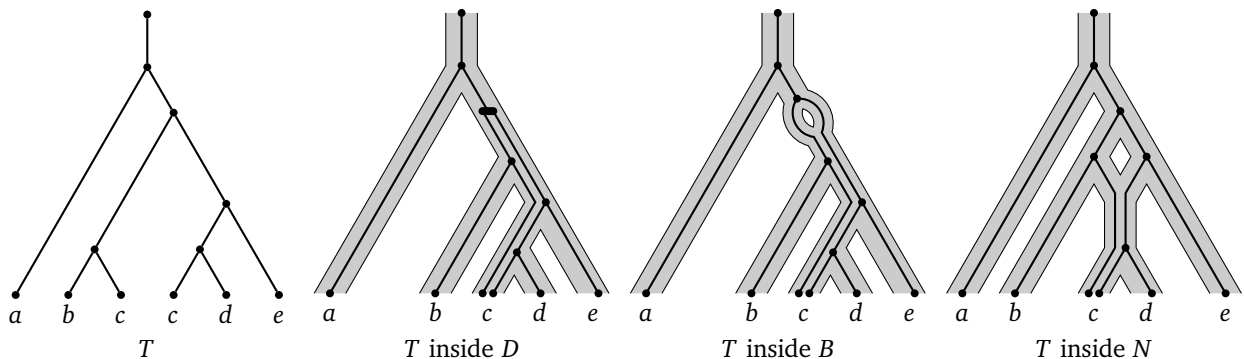
\begin{figure}[t]
  \begin{tikzpicture}[
    node/.style={fill=black,circle,inner sep=0pt,minimum size=3pt,outer sep=0pt},
    dup/.style={fill=red,circle,inner sep=0pt,minimum size=6pt},
    wide node/.style={fill=black,rectangle,inner sep=0pt,minimum height=3pt,rounded corners=1.5pt},
    edge/.style={draw,thick},
    x=0.7cm,y=0.7cm]
    \begin{scope}[local bounding box=T]
      \foreach \i in {0,...,5} {
        \node [node] (l\i) at (\i,0) {};
      }
      \path (l0) +(60:5) node [node] (p05) {};
      \path (l1) +(60:1) node [node] (p12) {};
      \path (l1) +(60:4) node [node] (p15) {};
      \path (l3) +(60:1) node [node] (p34) {};
      \path (l3) +(60:2) node [node] (p35) {};
      \path (p05) +(90:1) node [node] (r) {};
      \path [edge] (l1) -- (p12) -- (l2) (l3) -- (p34) -- (l4) (p34) -- (p35) -- (l5)
      (p12) -- (p15) -- (p35) (l0) -- (p05) -- (p15) (p05) -- (r);
      \node [anchor=north,text height=height("$d$")] at (l0) {$a$};
      \node [anchor=north,text height=height("$d$")] at (l1) {$b$};
      \node [anchor=north,text height=height("$d$")] at (l2) {$c$};
      \node [anchor=north,text height=height("$d$")] at (l3) {$c$};
      \node [anchor=north,text height=height("$d$")] at (l4) {$d$};
      \node [anchor=north,text height=height("$d$")] at (l5) {$e$};
    \end{scope}
    \node [anchor=north] at (T.south) {$T$};
  \end{tikzpicture}%
  \hspace*{\stretch{1}}%
  \begin{tikzpicture}[
    node/.style={fill=black,circle,inner sep=0pt,minimum size=3pt,outer sep=0pt},
    dup/.style={fill=red,circle,inner sep=0pt,minimum size=6pt},
    wide node/.style={fill=black,rectangle,inner sep=0pt,minimum height=3pt,rounded corners=1.5pt},
    edge/.style={draw,thick},
    x=0.7cm,y=0.7cm]
    \begin{scope}[local bounding box=D]
      \foreach \i/\x in {0/0,1/2,2/3,3/4,4/5} {
        \coordinate (l\i) at (\x,0);
      }
      \path (l0) +(60:5) coordinate (p04);
      \path (l1) +(60:3) coordinate (p14);
      \path (l2) +(60:1) coordinate (p23);
      \path (l2) +(60:2) coordinate (p24);
      \path (p04) +(90:1) coordinate (r);
      \coordinate (pp14) at (barycentric cs:p14=0.5,p04=0.5);
      \foreach \v in {l0,l1,l2,l3,l4,p04,p14,pp14,p23,p24,r} {
        \path (\v)
        +(180:0.1) coordinate (\v 1)
        +(180:0.3) coordinate (\v l)
        +(0:0.1) coordinate (\v 2)
        +(0:0.3) coordinate (\v r);
      }
      \path [name path=path1] (l0r) -- +(60:5);
      \path [name path=path2] (l4l) -- +(120:5);
      \path [name intersections={of=path1 and path2}] coordinate (p04b) at (intersection-1);
      \path [name path=path1] (l1r) -- +(60:3);
      \path [name intersections={of=path1 and path2}] coordinate (p14b) at (intersection-1);
      \path [name path=path1] (l2r) -- +(60:2);
      \path [name intersections={of=path1 and path2}] coordinate (p24b) at (intersection-1);
      \path [name path=path2] (l3l) -- +(120:1);
      \path [name intersections={of=path1 and path2}] coordinate (p23b) at (intersection-1);
      \path [name path=path1] (l22) -- +(60:1);
      \path [name path=path2] (l3) -- +(120:1);
      \path [name intersections={of=path1 and path2}] coordinate (p23m) at (intersection-1);
      \path [name path=path1] (l1) -- +(60:3);
      \path [name path=path2] (l41) -- +(120:3);
      \path [name intersections={of=path1 and path2}] coordinate (p14m) at (intersection-1);
      \path [fill=black!20] (rl) -- (p04l) -- (l0l) -- (l0r) -- (p04b) -- (p14l) -- (l1l) -- (l1r)
      -- (p14b) -- (p24l) -- (l2l) -- (l2r) -- (p23b) -- (l3l) -- (l3r) -- (p23r) -- (p24b)
      -- (l4l) -- (l4r) -- (p04r) -- (rr) -- cycle;
      \path [draw] (rl) -- (p04l) -- (l0l) (l0r) -- (p04b) -- (p14l) -- (l1l)
      (l1r) -- (p14b) -- (p24l) -- (l2l) (l2r) -- (p23b) -- (l3l)
      (l3r) -- (p23r) -- (p24b) -- (l4l) (rr) -- (p04r) -- (l4r);
      \path [edge] (l22) -- (p23m) -- (l3) (p23m) -- (p242) -- (l42) (l0) -- (p04) -- (pp14)
      (p04) -- (r) (l1) -- (p14m) -- (p241) -- (l21) (pp141) -- (p14m) (pp142) -- (p242);
      \node [node] at (l0) {};
      \node [node] at (l1) {};
      \node [node] at (l21) {};
      \node [node] at (l22) {};
      \node [node] at (l3) {};
      \node [node] at (l42) {};
      \node [node] at (p242) {};
      \node [node] at (p23m) {};
      \node [node] at (p04) {};
      \node [node] at (r) {};
      \node [node] at (p14m) {};
      \node [wide node,minimum width=0.25cm] at (pp14) {};
      \node [anchor=north,text height=height("$d$")] at (l0) {$a$};
      \node [anchor=north,text height=height("$d$")] at (l1) {$b$};
      \node [anchor=north,text height=height("$d$")] at (l2) {$c$};
      \node [anchor=north,text height=height("$d$")] at (l3) {$d$};
      \node [anchor=north,text height=height("$d$")] at (l4) {$e$};
    \end{scope}
    \node [anchor=north] at (D.south) {$T$ inside~$D$};
  \end{tikzpicture}%
  \hspace*{\stretch{1}}%
  \begin{tikzpicture}[
    node/.style={fill=black,circle,inner sep=0pt,minimum size=3pt,outer sep=0pt},
    dup/.style={fill=red,circle,inner sep=0pt,minimum size=6pt},
    wide node/.style={fill=black,rectangle,inner sep=0pt,minimum height=3pt,rounded corners=1.5pt},
    edge/.style={draw,thick},
    x=0.7cm,y=0.7cm]
    \begin{scope}[local bounding box=B]
      \foreach \i/\x in {0/0,1/2,2/3,3/4,4/5} {
        \coordinate (l\i) at (\x,0);
      }
      \path (l0) +(60:5) coordinate (p04);
      \path (l1) +(60:3) coordinate (p14);
      \path (l2) +(60:1) coordinate (p23);
      \path (l2) +(60:2) coordinate (p24);
      \path (p04) +(90:1) coordinate (r);
      \coordinate (pp14) at (barycentric cs:p14=0.575,p04=0.425);
      \foreach \v in {l0,l1,l2,l3,l4,p04,p14,pp14,p23,p24,r} {
        \path (\v)
        +(180:0.1) coordinate (\v 1)
        +(180:0.3) coordinate (\v l)
        +(0:0.1) coordinate (\v 2)
        +(0:0.3) coordinate (\v r);
      }
      \path [name path=path1] (l0r) -- +(60:5);
      \path [name path=path2] (l4l) -- +(120:5);
      \path [name intersections={of=path1 and path2}] coordinate (p04b) at (intersection-1);
      \path [name path=path1] (l1r) -- +(60:3);
      \path [name intersections={of=path1 and path2}] coordinate (p14b) at (intersection-1);
      \path [name path=path1] (l2r) -- +(60:2);
      \path [name intersections={of=path1 and path2}] coordinate (p24b) at (intersection-1);
      \path [name path=path2] (l3l) -- +(120:1);
      \path [name intersections={of=path1 and path2}] coordinate (p23b) at (intersection-1);
      \path [name path=path1] (l22) -- +(60:1);
      \path [name path=path2] (l3) -- +(120:1);
      \path [name intersections={of=path1 and path2}] coordinate (p23m) at (intersection-1);
      \path [name path=path1] (l1) -- +(60:3);
      \path [name path=path2] (l41) -- +(120:3);
      \path [name intersections={of=path1 and path2}] coordinate (p14m) at (intersection-1);
      \path (pp14) +(120:0.2) coordinate (pp14it) +(-60:0.2) coordinate (pp14ib);
      \path [name path=path1] (pp14it) -- +(-15:1);
      \path [name path=path2] (pp14ib) -- +(75:1);
      \path [name intersections={of=path1 and path2}] coordinate (beadcr) at (intersection-1);
      \path [name path=path1] (pp14it) -- +(255:1);
      \path [name path=path2] (pp14ib) -- +(165:1);
      \path [name intersections={of=path1 and path2}] coordinate (beadcl) at (intersection-1);
      \path [name path=path1] (beadcr) -- +(165:1.5);
      \path [name path=path2] (l4l) -- +(120:5);
      \path [name intersections={of=path1 and path2}] coordinate (pp14otl) at (intersection-1);
      \path [name path=path1] (beadcl) -- +(-15:1.5);
      \path [name path=path2] (l4r) -- +(120:5);
      \path [name intersections={of=path1 and path2}] coordinate (pp14obr) at (intersection-1);
      \path [name path=path1]
      let \p1 = ($(pp14it) - (beadcr)$), \n1 = {veclen(\x1,\y1)},
          \p2 = ($(pp14otl) - (beadcr)$), \n2 = {veclen(\x2,\y2)},
          \n3 = {(\n2 + \n1) / 2}
      in (beadcr) circle (\n3);
      \path [name path=path2] (l4) -- +(120:5);
      \path [name intersections={of=path1 and path2}] coordinate (pp14t) at (intersection-1);
      \path [name path=path2] (l41) -- +(120:5);
      \path [name intersections={of=path1 and path2}] coordinate (pp14bl) at (intersection-2);
      \path [name path=path1]
      let \p1 = ($(pp14it) - (beadcl)$), \n1 = {veclen(\x1,\y1)},
          \p2 = ($(pp14obr) - (beadcl)$), \n2 = {veclen(\x2,\y2)},
          \n3 = {(\n2 + \n1) / 2}
      in (beadcl) circle (\n3);
      \path [name path=path2] (l42) -- +(120:5);
      \path [name intersections={of=path1 and path2}] coordinate (pp14br) at (intersection-2);
      \path [fill=black!20] let \p1 = ($(beadcr) - (pp14otl)$), \n1 = {veclen(\x1,\y1)} in
      (rl) -- (p04l) -- (l0l) -- (l0r) -- (p04b) -- (pp14otl) arc (165:255:\n1)
      -- (p14l) -- (l1l) -- (l1r) -- (p14b) -- (p24l) -- (l2l) -- (l2r) -- (p23b) -- (l3l) -- (l3r)
      -- (p23r) -- (p24b) -- (l4l) -- (l4r) -- (pp14obr) arc (-15:75:\n1) -- (p04r) -- (rr)
      -- cycle;
      \path [draw,fill=white] let \p1 = ($(beadcr) - (pp14it)$), \n1 = {veclen(\x1,\y1)}
      in (pp14it) arc (165:255:\n1) arc (-15:75:\n1) -- cycle;
      \path [draw] let \p1 = ($(beadcr) - (pp14otl)$), \n1 = {veclen(\x1,\y1)} in
      (rl) -- (p04l) -- (l0l) (l0r) -- (p04b) -- (pp14otl) arc (165:255:\n1) -- (p14l) -- (l1l)
      (l1r) -- (p14b) -- (p24l) -- (l2l) (l2r) -- (p23b) -- (l3l)
      (l3r) -- (p23r) -- (p24b) -- (l4l) (l4r) -- (pp14obr) arc (-15:75:\n1) -- (p04r) -- (rr);
      \path [edge] (l22) -- (p23m) -- (l3) (p23m) -- (p242) -- (l42) (l0) -- (p04) -- (pp14t)
      (p04) -- (r) (l1) -- (p14m) -- (p241) -- (l21);
      \path [edge] let \p1 = ($(pp14t) - (beadcr)$), \n1 = {veclen(\x1,\y1)}, \n2 = {atan2(\y1,\x1)},
      \p2 = ($(pp14bl) - (beadcr)$), \n3 = {360 + atan2(\y2,\x2)}
      in (pp14t) arc (\n2:\n3:\n1) -- (p14m);
      \path [edge] let \p1 = ($(pp14t) - (beadcl)$), \n1 = {veclen(\x1,\y1)}, \n2 = {atan2(\y1,\x1)},
      \p2 = ($(pp14br) - (beadcl)$), \n3 = {atan2(\y2,\x2)}
      in (pp14t) arc (\n2:\n3:\n1) -- (p242);
      \node [node] at (l0) {};
      \node [node] at (l1) {};
      \node [node] at (l21) {};
      \node [node] at (l22) {};
      \node [node] at (l3) {};
      \node [node] at (l42) {};
      \node [node] at (p242) {};
      \node [node] at (p23m) {};
      \node [node] at (p04) {};
      \node [node] at (r) {};
      \node [node] at (p14m) {};
      \node [node] at (pp14t) {};
      \node [anchor=north,text height=height("$d$")] at (l0) {$a$};
      \node [anchor=north,text height=height("$d$")] at (l1) {$b$};
      \node [anchor=north,text height=height("$d$")] at (l2) {$c$};
      \node [anchor=north,text height=height("$d$")] at (l3) {$d$};
      \node [anchor=north,text height=height("$d$")] at (l4) {$e$};
    \end{scope}
    \node [anchor=north] at (B.south) {$T$ inside~$B$};
  \end{tikzpicture}%
  \hspace*{\stretch{1}}%
  \begin{tikzpicture}[
    node/.style={fill=black,circle,inner sep=0pt,minimum size=3pt,outer sep=0pt},
    dup/.style={fill=red,circle,inner sep=0pt,minimum size=6pt},
    wide node/.style={fill=black,rectangle,inner sep=0pt,minimum height=3pt,rounded corners=1.5pt},
    edge/.style={draw,thick},
    x=0.7cm,y=0.7cm]
    \begin{scope}[local bounding box=N]
      \foreach \i/\x in {0/0,1/1,2/2.5,3/3.5,4/5} {
        \coordinate (l\i) at (\x,0);
      }
      \path (l0) +(60:5) coordinate (p04);
      \path (l1) +(60:4) coordinate (p14);
      \path (l1) +(60:3) coordinate (p13);
      \path (l2) +(60:1) coordinate (p23);
      \path (l4) +(120:3) coordinate (p24);
      \path (p13) +(-60:1) coordinate (pp23);
      \path (p04) +(90:1) coordinate (r);
      \foreach \v in {l0,l1,l2,l3,l4,p04,p13,p14,p23,pp23,p24,r} {
        \path (\v)
        +(180:0.1) coordinate (\v 1)
        +(180:0.3) coordinate (\v l)
        +(0:0.1) coordinate (\v 2)
        +(0:0.3) coordinate (\v r);
      }
      \path [name path=path1] (p13) -- +(-60:1);
      \path [name path=path2] (pp231) -- +(90:1);
      \path [name intersections={of=path1 and path2}] coordinate (pp231m) at (intersection-1);
      \path [name path=path1] (pp232) -- +(90:1);
      \path [name path=path2] (p24) -- +(-120:1);
      \path [name intersections={of=path1 and path2}] coordinate (pp232m) at (intersection-1);
      \path [name path=path1] (l0r) -- +(60:5);
      \path [name path=path2] (l4l) -- +(120:5);
      \path [name intersections={of=path1 and path2}] coordinate (p04b) at (intersection-1);
      \path [name path=path1] (pp23r) -- +(60:1);
      \path [name path=path2] (l4l) -- +(120:3);
      \path [name intersections={of=path1 and path2}] coordinate (p24b) at (intersection-1);
      \path [name path=path1] (l2r) -- +(60:1);
      \path [name path=path2] (l3l) -- +(120:1);
      \path [name intersections={of=path1 and path2}] coordinate (p23b) at (intersection-1);
      \path [name path=path1] (l1r) -- +(60:3);
      \path [name path=path2] (pp23l) -- +(120:1);
      \path [name intersections={of=path1 and path2}] coordinate (p13b) at (intersection-1);
      \path [name path=path1] (p13r) -- +(60:1);
      \path [name path=path2] (p24l) -- +(120:1);
      \path [name intersections={of=path1 and path2}] coordinate (p14b) at (intersection-1);
      \path [name path=path1] (p13r) -- +(-60:1);
      \path [name path=path2] (p24l) -- +(-120:1);
      \path [name intersections={of=path1 and path2}] coordinate (pp23t) at (intersection-1);
      \path [fill=black!20] (rl) -- (p04l) -- (l0l) -- (l0r) -- (p04b) -- (p14l) -- (l1l) -- (l1r)
      -- (p13b) -- (pp23l) -- (p23l) -- (l2l) -- (l2r) -- (p23b) -- (l3l) -- (l3r) -- (p23r)
      -- (pp23r) -- (p24b) -- (l4l) -- (l4r) -- (p04r) -- (rr) -- cycle;
      \path [draw] (rl) -- (p04l) -- (l0l) (rr) -- (p04r) -- (l4r)
      (l0r) -- (p04b) -- (p14l) -- (l1l) (l3r) -- (p23r) -- (pp23r) -- (p24b) -- (l4l)
      (l2r) -- (p23b) -- (l3l) (l1r) -- (p13b) -- (pp23l) -- (p23l) -- (l2l);
      \path [draw,fill=white] (p13r) -- (p14b) -- (p24l) -- (pp23t) -- cycle;
      \path [edge] (l0) -- (p04) -- (p14) (p04) -- (r) (l1) -- (p13) -- (pp231m) -- (p231) -- (l21)
      (l22) -- (p232) -- (l32) (p232) -- (pp232m) -- (p24) -- (l4) (p13) -- (p14) -- (p24);
      \node [node] at (l0) {};
      \node [node] at (l1) {};
      \node [node] at (l21) {};
      \node [node] at (l22) {};
      \node [node] at (l32) {};
      \node [node] at (l4) {};
      \node [node] at (p04) {};
      \node [node] at (r) {};
      \node [node] at (p232) {};
      \node [node] at (p13) {};
      \node [node] at (p14) {};
      \node [node] at (p24) {};
      \node [anchor=north,text height=height("$d$")] at (l0) {$a$};
      \node [anchor=north,text height=height("$d$")] at (l1) {$b$};
      \node [anchor=north,text height=height("$d$")] at (l2) {$c$};
      \node [anchor=north,text height=height("$d$")] at (l3) {$d$};
      \node [anchor=north,text height=height("$d$")] at (l4) {$e$};
    \end{scope}
    \node [anchor=north] at (N.south) {$T$ inside~$N$};
  \end{tikzpicture}%
  \caption{A MUL-tree~$T$ and illustrations of a duplication mapping from $T$ to a duplication tree
    $D$, and of weak embeddings of $T$ into a beaded tree $B$ and into a
    phylogenetic network $N$ that is not a beaded tree.}
  \label{fig:trees_networks}
\end{figure}

\begin{definition}
  Given a set $X$ of species, let $N$ be a phylogenetic network, and $T$ a
  MUL-tree on $X$.  A \emph{weak embedding} of $T$ into $N$ is a function $h$
  that maps every node of $T$ to a node of $N$, and every edge in $T$ to a
  directed path in $N$ such that
  \begin{itemize}
    \item For each leaf $l \in L(T)$, $h(l)$ is a leaf of $N$ labelled with the
      same species,
    \item For each edge $xy \in E(T)$, the path $h(xy)$ is a path from $h(x)$
      to $h(y)$ in $N$, and
    \item For each internal node $x$ in $T$ with children $y,y'$, the paths
      $h(xy)$ and $h(xy')$ start with different out-edges of $h(x)$.
  \end{itemize}
  This is illustrated in Figure~\ref{fig:trees_networks}.
  We say that $N$ \emph{weakly displays} $T$ if there is a weak embedding of
  $T$ into $N$.
\end{definition}

We note that $N$ weakly displays $T$ if and only if $T$ is a \emph{parental
  tree inside $N$} as defined in~\cite{zhu2016light}, hence the name
\textsc{Parental Hybridization}. The notion of a tree \emph{weakly displayed}
by a network was first introduced in~\cite{huber2016folding}, where it was
shown that $T$ is weakly displayed by $N$ if and only if there exists a
\emph{locally separated reconciliation} from $T$  to $N$, which is equivalent
to our definition of a weak embedding.

We now define the \textsc{Parental Hybridization} problem:

\medskip

\noindent \textsc{Parental Hybridization}\\
\noindent \textbf{Input}: A set $\mathcal{T} = \{T_1, \ldots, T_t\}$ of MUL-trees
with label sets $X_1, \ldots, X_t \subseteq X$.\\
\noindent \textbf{Output}:  A phylogenetic network $N$ on $X$ with minimum
reticulation number and such that $N$ weakly displays all MUL-trees in
$\mathcal{T}$.

\medskip

Even though we do not use it in this paper, it is worth noting the relationship
between weak embeddings, weakly displayed trees, and \textsc{Parental
  Hybridization} on one hand and embeddings, displayed trees, and
\textsc{Hybridization Number} on the other hand.  An \emph{embedding} of a tree
$T$ into a network $N$ is a weak embedding $h$ of $T$ into $N$ with the added
condition that the paths $h(e)$ and $h(e')$ are edge-disjoint for every pair of
edges $e \ne e' \in T$.  (Note that this also implies that the paths are
node-disjoint unless $e$ and $e'$ have a node in common.) If such an embedding
exists, then $N$ \emph{displays} $T$.  Similarly to \textsc{Parental
  Hybridization}, the \textsc{Hybridization Number} problem for a set of
phylogenetic trees $\mathcal{T}$ asks for a phylogenetic network $N$ with the
minimum reticulation number and such that $N$ displays all trees in
$\mathcal{T}$.

\subsection{Beaded Trees}

The key to solving both \textsc{Unrestricted Minimal Episodes Inference}
and \textsc{Parental Hybridization} is the equivalence between these two
problems and the following \textsc{Beaded Tree} problem, which we establish
in this paper.

\begin{definition}
  A \emph{bead} in a phylogenetic network $N$ is a pair of nodes $(u,v)$ such
  that there are two parallel edges from $u$ to $v$.  A \emph{beaded tree} is a
  phylogenetic network $B$ in which every reticulation node is part of a bead
  (see Figure~\ref{fig:trees_networks}).
\end{definition}

The \textsc{Beaded Tree} problem is defined as follows:

\medskip

\noindent \textsc{Beaded Tree}\\
\noindent \textbf{Input}: A set $\mathcal{T} = \{T_1, \ldots, T_t\}$ of
MUL-trees with label sets $X_1, \ldots, X_t \subseteq X$.\\
\noindent \textbf{Output}:  A beaded tree $B$ on $X$ with minimum reticulation
number that weakly displays all MUL-trees in $\mathcal{T}$.

\section{Reductions to \textsc{Beaded Tree}}\label{sec:reduction}

In this section, we show that the two problems \textsc{Unrestricted Minimal
  Episodes Inference} and \textsc{Parental Hybridization} are both reducible to
\textsc{Beaded Tree}, which will allow us to focus on the latter problem in the
rest of the paper.  We begin with the proof for \textsc{Unrestricted Minimal
  Episodes Inference}.

\begin{lemma}\label{lem:WGDequivalence}
  Let $X$ be a set of species and $\mathcal{T} = \{T_1, \ldots, T_t\}$ a set of
  MUL-trees on $X$.  For any integer~$k$, there exists a solution to
  \textsc{Unrestricted Minimal Episodes Inference} on $\mathcal{T}$ with $k$
  duplications if and only if there exists a solution to \textsc{Beaded Tree}
  on $\mathcal{T}$ with $k$ beads.
\end{lemma}

\begin{proof}
  Let the duplication tree $D$ be a solution to \textsc{Unrestricted Minimal
    Episodes Inference} on $\mathcal{T}$ with $k$ duplications.  Then construct a
  beaded tree $B$ from $D$ as follows:  Replace each duplication node $d$ in
  $D$ with a bead $(u_d,v_d)$.  If $p$ is $d$'s parent in $D$, then $u_d$'s
  parent in $B$ is $p$ or, if $p$ is itself a duplication node,~$v_p$; $v_d$'s
  child in $B$ is $d$'s child $c$ in $D$ or, if $c$ is itself a duplication
  node, $u_c$.

  It is easy to observe that $B$ is a beaded tree with $k$ beads. To see that
  $B$ is a solution to \textsc{Beaded Tree} on~$\mathcal{T}$, consider any tree
  $T \in \mathcal{T}$ and let $M$ be a duplication mapping from $T$ to $D$.
  Then we can construct a weak embedding $h$ from $T$ into $B$ as follows.  For
  each node $x$ in $T$, if $M(x)$ is a duplication node $d$, then let $h(x)$ be
  the tree node~$u_d$ (i.e., the top node of the bead $(u_d,v_d)$). Otherwise,
  let $h(x)=M(x)$.  For any edge~$xy$, the node $h(y)$ is by construction a
  strict descendant of $h(x)$, so there exists a path from $h(x)$ to $h(y)$ in
  $B$.  We choose $h(xy)$ to be any such path but ensure that the two paths
  $h(xy)$ and $h(xy')$ start with different edges in the bead $(u_d, v_d)$ if
  $M(x)$ is a duplication node $d$ in $D$ and $y$ and $y'$ are $x$'s children
  in $T$.  This guarantees that the paths $h(xy)$ and $h(xy')$ start with
  different out-edges of $h(x)$ if $M(x)$ is a duplication node. If $M(x)$ is
  not a duplication node, then $M(x)$ is the least common ancestor of $M(y)$
  and $M(y')$, so the paths $h(xy)$ and $h(xy')$ are edge-disjoint and again
  start with different out-edges of $h(x)$.  Thus, $h$ is a weak embedding of
  $T$ into~$B$.

  Conversely, let the beaded tree $B$ be a solution to \textsc{Beaded Tree} on
  $\mathcal{T}$ with $k$ beads.  Then construct a duplication tree $D$ from $B$
  by replacing each bead $(u,v)$ with a duplication node $d_{(u,v)}$.
  $d_{(u,v)}$'s parent in $D$ is $u$'s parent $p$ in $B$ or, if $p$ is itself
  part of a bead $(x,p)$, the duplication node $d_{(x,p)}$; $d_{(u,v)}$'s child
  in $D$ is $v$'s child $c$ in $B$ or, if $c$ is itself part of a bead $(c,y)$,
  the duplication node $d_{(c,y)}$.

  It is easy to observe that $D$ is a duplication tree with $k$ duplications.
  To see that $D$ is a solution to  \textsc{Unrestricted Minimal Episodes
    Inference} on $\mathcal{T}$, consider any tree $T \in \mathcal{T}$ and let
  $h$ be a weak embedding of $T$ into $B$.  Then we can construct a duplication
  mapping from $T$ to $D$ as follows.  For any node $x$ in $T$, if $h(x)$ is
  not in a bead, then set $M(x)=h(x)$.  If $h(x)$ is the top node $u$ of a bead
  $(u,v)$, then let $M(x) = d_{(u,v)}$.  (Note that $h(x)$ cannot be the bottom
  node of a bead, because $x$ is either a leaf or has out-degree~$2$.) By the
  requirements of a weak embedding, $M(x)$ is a strict ancestor of $M(y)$ for
  any edge $xy$ in~$T$. Furthermore, for any internal node $x$  with children
  $y$ and~$y'$, there are paths from $h(x)$ to $h(y)$ and from $h(x)$ to
  $h(y')$ that start with different out-edges of $h(x)$.  It follows that
  either $M(x) = h(x)$ is the least common ancestor of $M(y)$ and $M(y')$ or
  $M(x)$ is a duplication node.
\end{proof}

The next lemma shows that any instance $\mathcal{T}$ of \textsc{Parental
  Hybridization} has an optimal solution that is a beaded tree, that is,
\textsc{Parental Hybridization} can be reduced to \textsc{Beaded Tree}.

\begin{figure}[t]
  \hspace*{\stretch{1}}%
  \begin{tikzpicture}[
    node/.style={fill=black,circle,inner sep=0pt,minimum size=3pt,outer sep=0pt},
    dup/.style={fill=red,circle,inner sep=0pt,minimum size=6pt},
    wide node/.style={fill=black,rectangle,inner sep=0pt,minimum height=3pt,rounded corners=1.5pt},
    edge/.style={draw,thick}]
    \node [node] (q) at (0,0) {};
    \node [node] (r) at (0,1) {};
    \path (r) +(60:1) node [node] (dt) {} +(120:1) node [node] (cs) {};
    \path (cs) +(90:2) node [node] (c1) {};
    \path (dt) +(90:2) node [node] (d1) {};
    \path (c1) +(60:1) node [node] (u) {};
    \path (u) +(90:1) coordinate (root);
    \path (c1) +(225:1) coordinate (cc1);
    \path (cs) +(225:1) coordinate (ccs);
    \path (d1) +(-45:1) coordinate (cd1);
    \path (dt) +(-45:1) coordinate (cdt);
    \path [edge] (cc1) -- (c1) -- (u) -- (d1) -- (cd1)
    (ccs) -- (cs) -- (r) -- (dt) -- (cdt)
    (r) -- (q) (root) -- (u);
    \path [edge,dashed] (c1) -- (cs) (d1) -- (dt);
    \node [anchor=west] at (r) {$r$};
    \node [anchor=west] at (q) {$q$};
    \node [anchor=west] at (u) {$u$};
    \node [anchor=west,yshift=3pt] at (d1) {$d_1$};
    \node [anchor=west,yshift=3pt] at (dt) {$d_t$};
    \node [anchor=east,yshift=3pt] at (c1) {$c_1$};
    \node [anchor=east,yshift=3pt] at (cs) {$c_s$};
  \end{tikzpicture}%
  \hspace*{\stretch{1}}%
  \begin{tikzpicture}[
    node/.style={fill=black,circle,inner sep=0pt,minimum size=3pt,outer sep=0pt},
    dup/.style={fill=red,circle,inner sep=0pt,minimum size=6pt},
    wide node/.style={fill=black,rectangle,inner sep=0pt,minimum height=3pt,rounded corners=1.5pt},
    edge/.style={draw,thick}]
    \begin{scope}[overlay]
      \coordinate (q) at (0,0);
      \coordinate (r) at (0,1);
      \path (r) +(60:1) coordinate (d3) +(120:1) coordinate (c3) {};
      \path (c3) +(90:1) coordinate (c2);
      \path (c2) +(90:1) coordinate (c1);
      \path (d3) +(90:1) coordinate (d2);
      \path (d2) +(90:1) coordinate (d1);
      \path (c1) +(60:1) coordinate (u);
      \path (u) +(90:1) coordinate (root);
      \path (c1) +(225:1) coordinate (cc1);
      \path (c2) +(225:1) coordinate (cc2);
      \path (d1) +(-45:1) coordinate (cd1);
      \path (d2) +(-45:1) coordinate (cd2);
      \path (d3) +(-45:1) coordinate (cd3);
      \foreach \v in {root,u,r,q,c1,c2,c3,d1,d2,d3,cc1,cc2,cd1,cd2,cd3} {
        \path (\v)
        +(180:0.05) coordinate (\v 1)
        +(180:0.15) coordinate (\v l)
        +(0:0.05) coordinate (\v 2)
        +(0:0.15) coordinate (\v r)
        +(90:0.15) coordinate (\v t)
        +(270:0.15) coordinate (\v b);
      }
      \path [name path=path1] (cc1t) -- +(45:1);
      \path [name path=path2] (ul) -- +(240:1.5);
      \path [name intersections={of=path1 and path2}] coordinate (c1l) at (intersection-1);
      \path [name path=path1] (cd1t) -- +(135:1);
      \path [name path=path2] (ur) -- +(-60:1.5);
      \path [name intersections={of=path1 and path2}] coordinate (d1r) at (intersection-1);
      \path [name path=path1] (c1r) -- +(60:1);
      \path [name path=path2] (d1l) -- +(120:1);
      \path [name intersections={of=path1 and path2}] coordinate (ub) at (intersection-1);
      \path [name path=path1] (c3r) -- +(-60:1);
      \path [name path=path2] (d3l) -- +(-120:1);
      \path [name intersections={of=path1 and path2}] coordinate (rt) at (intersection-1);
      \path [name path=path1] (rr) -- +(60:1);
      \path [name path=path2] (cd3b) -- +(135:1);
      \path [name intersections={of=path1 and path2}] coordinate (d3r) at (intersection-1);
      \path [name path=path1] (c3l) -- +(90:2);
      \path [name path=path2] (cc2b) -- +(45:1);
      \path [name intersections={of=path1 and path2}] coordinate (c2b) at (intersection-1);
      \path [name path=path2] (cc2t) -- +(45:1);
      \path [name intersections={of=path1 and path2}] coordinate (c2t) at (intersection-1);
      \path [name path=path2] (cc1b) -- +(45:1);
      \path [name intersections={of=path1 and path2}] coordinate (c1b) at (intersection-1);
      \path [name path=path1] (d2r) +(270:1.5) -- +(90:1);
      \path [name path=path2] (cd3t) -- +(135:1);
      \path [name intersections={of=path1 and path2}] coordinate (d3t) at (intersection-1);
      \path [name path=path2] (cd2b) -- +(135:1);
      \path [name intersections={of=path1 and path2}] coordinate (d2b) at (intersection-1);
      \path [name path=path2] (cd2t) -- +(135:1);
      \path [name intersections={of=path1 and path2}] coordinate (d2t) at (intersection-1);
      \path [name path=path2] (cd1b) -- +(135:1);
      \path [name intersections={of=path1 and path2}] coordinate (d1b) at (intersection-1);
      \path [name path=path1] (q1) -- +(90:2);
      \path [name path=path2] (c3) -- +(300:1);
      \path [name intersections={of=path1 and path2}] coordinate (r1) at (intersection-1);
      \path [name path=path1] (q2) -- +(90:2);
      \path [name path=path2] (d3) -- +(240:1);
      \path [name intersections={of=path1 and path2}] coordinate (r2) at (intersection-1);
    \end{scope}
    \path [fill=black!20] (rootl) -- (ul) -- (c1l) -- (cc1t) -- (cc1b) -- (c1b) -- (c2t)
    -- (cc2t) -- (cc2b) -- (c2b) -- (c3l) -- (rl) -- (ql) -- (qr) -- (rr) -- (d3r)
    -- (cd3b) -- (cd3t) -- (d3t) -- (d2b) -- (cd2b) -- (cd2t) -- (d2t) -- (d1b) -- (cd1b)
    -- (cd1t) -- (d1r) -- (ur) -- (rootr) -- cycle;
    \path [fill=blue!30] (ul) -- (ub) -- (c1r) -- (c3r) -- (rt) -- (rl) -- (c3l) -- (c2b)
    -- (cc2b) -- (cc2t) -- (c2t) -- (c1b) -- (cc1b) -- (cc1t) -- (c1l) -- cycle;
    \path [fill=red!30] (ur) -- (ub) -- (d1l) -- (d3l) -- (rt) -- (rr) -- (d3r) -- (cd3b)
    -- (cd3t) -- (d3t) -- (d2b) -- (cd2b) -- (cd2t) -- (d2t) -- (d1b) -- (cd1b) -- (cd1t)
    -- (d1r) -- cycle;
    \path [pattern=north west lines,pattern color=white] (ul) -- (ub) -- (c1r) -- (c3r) -- (rt) -- (rl) -- (c3l) -- (c2b)
    -- (cc2b) -- (cc2t) -- (c2t) -- (c1b) -- (cc1b) -- (cc1t) -- (c1l) -- cycle;
    \path [fill=red!30] (ur) -- (ub) -- (d1l) -- (d3l) -- (rt) -- (rr) -- (d3r) -- (cd3b)
    -- (cd3t) -- (d3t) -- (d2b) -- (cd2b) -- (cd2t) -- (d2t) -- (d1b) -- (cd1b) -- (cd1t)
    -- (d1r) -- cycle;
    \path [pattern=north east lines,pattern color=white] (ur) -- (ub) -- (d1l) -- (d3l) -- (rt) -- (rr) -- (d3r) -- (cd3b)
    -- (cd3t) -- (d3t) -- (d2b) -- (cd2b) -- (cd2t) -- (d2t) -- (d1b) -- (cd1b) -- (cd1t)
    -- (d1r) -- cycle;
    \path [draw] (rootl) -- (ul) -- (c1l) -- (cc1t) (rootr) -- (ur) -- (d1r) -- (cd1t)
    (qr) -- (rr) -- (d3r) -- (cd3b) (ql) -- (rl) -- (c3l) -- (c2b) -- (cc2b)
    (cc2t) -- (c2t) -- (c1b) -- (cc1b) (cd3t) -- (d3t) -- (d2b) -- (cd2b)
    (cd2t) -- (d2t) -- (d1b) -- (cd1b);
    \path [draw,fill=white] (ub) -- (c1r) -- (c3r) -- (rt) -- (d3l) -- (d1l) -- cycle;
    \path [edge] (cc1) -- (c1) -- (c2) (d2) -- (d1) -- (cd1) (c1) -- (u) -- (d1) (u) -- (root)
    (d3) -- (d2) -- (cd2) (cc2) -- (c2) -- (c3) -- (r1) -- (q1) (cd3) -- (d3) -- (r2) -- (q2);
    \node [node] at (u) {};
    \node [node] at (c1) {};
    \node [node] at (c2) {};
    \node [node] at (d1) {};
    \node [node] at (d2) {};
    \node [node] at (d3) {};
    \node [anchor=west,xshift=-1pt,yshift=2pt] at (d3t) {$d_3$};
    \node [anchor=west,xshift=-1pt,yshift=2pt] at (d2t) {$d_2$};
    \node [anchor=west,xshift=-1pt,yshift=2pt] at (d1 -| d3t) {$d_1$};
    \node [anchor=east,xshift=1pt,yshift=2pt] at (c2t) {$c_2$};
    \node [anchor=east,xshift=1pt,yshift=2pt] at (c1 -| c2t) {$c_1$};
    \node [anchor=west,yshift=-2pt] at (rr) {$r$};
    \node [anchor=west] at (qr) {$q$};
    \node [anchor=west,yshift=2pt] at (ur) {$u$};
  \end{tikzpicture}%
  \hspace*{\stretch{1}}%
  \begin{tikzpicture}[
    node/.style={fill=black,circle,inner sep=0pt,minimum size=3pt,outer sep=0pt},
    dup/.style={fill=red,circle,inner sep=0pt,minimum size=6pt},
    wide node/.style={fill=black,rectangle,inner sep=0pt,minimum height=3pt,rounded corners=1.5pt},
    edge/.style={draw,thick}]
    \node [node] (q) at (0,0) {};
    \path let \n1 = {4+2*cos(30)} in coordinate (root) at (0,\n1);
    \node [node] (dt) at (0,0.75) {};
    \path let \p1 = ($(root) - (dt)$), \n1 = {(\y1 - 2.25cm) / 3},
    \p2 = ($(root) - (0,\n1+1.5cm)$),
    \p3 = ($(dt) + (0,\n1)$)
    in node [node] (d1) at (\p3) {} (d1) +(90:0.75) node [node] (cs) {} node [node] (c1) at (\p2) {};
    \path (root) +(270:0.75) node [node] (u) {};
    \path (c1)   +(90:0.75)  node [node] (v) {};
    \path (c1) +(225:1) coordinate (cc1);
    \path (cs) +(225:1) coordinate (ccs);
    \path (d1) +(315:1) coordinate (cd1);
    \path (dt) +(315:1) coordinate (cdt);
    \path [name path=path1] (u) -- +(315:1);
    \path [name path=path2] (v) -- +(45:1);
    \path [name intersections={of=path1 and path2}] coordinate (beadcr) at (intersection-1);
    \path [edge] (q) -- (dt) -- (cdt) (ccs) -- (cs) -- (d1) -- (cd1) (v) -- (c1) -- (cc1)
    (root) -- (u);
    \path [edge] let \p1 = ($(u) - (beadcr)$), \n1 = {veclen(\x1,\y1)} in
    (u) arc (135:225:\n1) (u) arc (45:-45:\n1);
    \path [edge,dashed] (c1) -- (cs) (d1) -- (dt);
    \node [anchor=west,yshift=3pt] at (d1) {$d_1$};
    \node [anchor=west,yshift=3pt] at (dt) {$d_t$};
    \node [anchor=east,yshift=3pt] at (c1) {$c_1$};
    \node [anchor=east,yshift=3pt] at (cs) {$c_s$};
    \node [anchor=east] at (q) {$q$};
    \node [anchor=west,yshift=3pt] at (u) {$u$};
    \node [anchor=west,yshift=-3pt] at (v) {$v$};
  \end{tikzpicture}%
  \hspace*{\stretch{1}}%
  \begin{tikzpicture}[
    node/.style={fill=black,circle,inner sep=0pt,minimum size=3pt,outer sep=0pt},
    dup/.style={fill=red,circle,inner sep=0pt,minimum size=6pt},
    wide node/.style={fill=black,rectangle,inner sep=0pt,minimum height=3pt,rounded corners=1.5pt},
    edge/.style={draw,thick}]
    \coordinate (q) at (0,0);
    \path let \n1 = {4+2*cos(30)} in coordinate (root) at (0,\n1);
    \path let \p1 = ($(root) - (dt)$), \n1 = {(\y1 - 2.25cm) / 3},
    \p2 = ($(root) - (0,\n1+1.5cm)$)
    in coordinate (c1) at (\p2);
    \path (root) +(270:0.75) coordinate (u);
    \path (c1)   +(90:0.75)  coordinate (v);
    \path (v) +(270:0.25) coordinate (c1);
    \coordinate (c2) at (barycentric cs:c1=0.75,q=0.25);
    \coordinate (d1) at (barycentric cs:c1=0.66,q=0.34);
    \coordinate (d2) at (barycentric cs:c1=0.43,q=0.57);
    \coordinate (d3) at (barycentric cs:c1=0.2,q=0.8);
    \foreach \v in {root,u,v,c1,c2,d1,d2,d3,q} {
      \path (\v)
      +(180:0.15) coordinate (\v 1)
      +(180:0.3) coordinate (\v l)
      +(0:0.15) coordinate (\v 2)
      +(0:0.3) coordinate (\v r);
    }
    \path foreach \v in {c1,c2} {
      (\v 1) +(225:1) coordinate (c\v)
    };
    \path foreach \v in {d1,d2,d3} {
      (\v 2) +(315:1) coordinate (c\v)
    };
    \foreach \v in {cc1,cc2,cd1,cd2,cd3} {
      \path (\v) +(90:0.15) coordinate (\v t) +(270:0.15) coordinate (\v b);
    }
    \path (barycentric cs:u=0.5,v=0.5) +(90:0.2) coordinate (uib) +(270:0.2) coordinate (vit);
    \path [name path=path1] (uib) -- +(315:1);
    \path [name path=path2] (vit) -- +(45:1);
    \path [name intersections={of=path1 and path2}] coordinate (beadcr) at (intersection-1);
    \path [name path=path1] (uib) -- +(225:1);
    \path [name path=path2] (vit) -- +(135:1);
    \path [name intersections={of=path1 and path2}] coordinate (beadcl) at (intersection-1);
    \path [name path=path1] (ql) -- (rootl);
    \path [name path=path2] (cc2b) -- +(45:1);
    \path [name intersections={of=path1 and path2}] coordinate (c2b) at (intersection-1);
    \path [name path=path2] (cc2t) -- +(45:1);
    \path [name intersections={of=path1 and path2}] coordinate (c2t) at (intersection-1);
    \path [name path=path2] (cc1b) -- +(45:1);
    \path [name intersections={of=path1 and path2}] coordinate (c1b) at (intersection-1);
    \path [name path=path2] (cc1t) -- +(45:1);
    \path [name intersections={of=path1 and path2}] coordinate (c1t) at (intersection-1);
    \path [name path=path2] (beadcr) -- +(135:1);
    \path [name intersections={of=path1 and path2}] coordinate (uotl) at (intersection-1);
    \path [name path=path1] (qr) -- (rootr);
    \path [name path=path2] (cd3b) -- +(135:1);
    \path [name intersections={of=path1 and path2}] coordinate (d3b) at (intersection-1);
    \path [name path=path2] (cd3t) -- +(135:1);
    \path [name intersections={of=path1 and path2}] coordinate (d3t) at (intersection-1);
    \path [name path=path2] (cd2b) -- +(135:1);
    \path [name intersections={of=path1 and path2}] coordinate (d2b) at (intersection-1);
    \path [name path=path2] (cd2t) -- +(135:1);
    \path [name intersections={of=path1 and path2}] coordinate (d2t) at (intersection-1);
    \path [name path=path2] (cd1b) -- +(135:1);
    \path [name intersections={of=path1 and path2}] coordinate (d1b) at (intersection-1);
    \path [name path=path2] (cd1t) -- +(135:1);
    \path [name intersections={of=path1 and path2}] coordinate (d1t) at (intersection-1);
    \path [name path=path2] (beadcl) -- +(315:1);
    \path [name intersections={of=path1 and path2}] coordinate (uobr) at (intersection-1);
    \path [name path=path1] let \p1 = ($(uobr) - (beadcl)$), \n1 = {veclen(\x1,\y1)},
    \p2 = ($(vit) - (beadcl)$), \n2 = {veclen(\x2,\y2)}, \n3 = {(\n1 + \n2) / 2} in
    (beadcl) circle (\n3);
    \path [name path=path2] (root) -- +(270:1);
    \path [name intersections={of=path1 and path2}] coordinate (u) at (intersection-1);
    \path [name path=path2] (q2) -- +(90:5);
    \path [name intersections={of=path1 and path2}] coordinate (v2) at (intersection-2);
    \path [name path=path1] let \p1 = ($(uotl) - (beadcr)$), \n1 = {veclen(\x1,\y1)},
    \p2 = ($(uib) - (beadcr)$), \n2 = {veclen(\x2,\y2)}, \n3 = {(\n1 + \n2) / 2} in
    (beadcr) circle (\n3);
    \path [name path=path2] (q1) -- +(90:5);
    \path [name intersections={of=path1 and path2}] coordinate (v1) at (intersection-2);
    \path (d3b) +(180:0.6) coordinate (d3bl) +(180:0.3) coordinate (d3bm);
    \path [fill=black!20] let \p1 = ($(uotl) - (beadcr)$), \n1 = {veclen(\x1,\y1)} in
    (rootl) -- (uotl) arc (135:225:\n1) -- (c1t) -- (cc1t) -- (cc1b) -- (c1b)
    -- (c2t) -- (cc2t) -- (cc2b) -- (c2b) -- (ql) -- (qr) -- (d3b) -- (cd3b) -- (cd3t)
    -- (d3t) -- (d2b) -- (cd2b) -- (cd2t) -- (d2t) -- (d1b) -- (cd1b) -- (cd1t) -- (d1t)
    -- (uobr) arc (-45:45:\n1) -- (rootr) -- cycle;
    \path [fill=blue!30] let \p1 = ($(uotl) - (beadcr)$), \n1 = {veclen(\x1,\y1)} in
    (uotl) arc (135:225:\n1) -- (c1t) -- (cc1t) -- (cc1b) -- (c1b) -- (c2t) -- (cc2t)
    -- (cc2b) -- (c2b) -- (d3bl) -- (d3bm) -- (uib) -- cycle;
    \path [pattern=north west lines,pattern color=white] let \p1 = ($(uotl) - (beadcr)$), \n1 = {veclen(\x1,\y1)} in
    (uotl) arc (135:225:\n1) -- (c1t) -- (cc1t) -- (cc1b) -- (c1b) -- (c2t) -- (cc2t)
    -- (cc2b) -- (c2b) -- (d3bl) -- (d3bm) -- (uib) -- cycle;
    \path [fill=red!30] let \p1 = ($(uobr) - (beadcl)$), \n1 = {veclen(\x1,\y1)} in
    (uib) -- (d3bm) -- (d3b) -- (cd3b) -- (cd3t) -- (d3t) -- (d2b) -- (cd2b) -- (cd2t) -- (d2t)
    -- (d1b) -- (cd1b) -- (cd1t) -- (d1t) -- (uobr) arc (-45:45:\n1) -- cycle;
    \path [pattern=north east lines,pattern color=white] let \p1 = ($(uobr) - (beadcl)$), \n1 = {veclen(\x1,\y1)} in
    (uib) -- (d3bm) -- (d3b) -- (cd3b) -- (cd3t) -- (d3t) -- (d2b) -- (cd2b) -- (cd2t) -- (d2t)
    -- (d1b) -- (cd1b) -- (cd1t) -- (d1t) -- (uobr) arc (-45:45:\n1) -- cycle;
    \path [draw] (cc2b) -- (c2b) -- (ql) (cc2t) -- (c2t) -- (c1b) -- (cc1b)
    (qr) -- (d3b) -- (cd3b) (cd3t) -- (d3t) -- (d2b) -- (cd2b) (cd2t) -- (d2t) -- (d1b) -- (cd1b);
    \path [draw] let \p1 = ($(uotl) - (beadcr)$), \n1 = {veclen(\x1,\y1)} in
    (rootl) -- (uotl) arc (135:225:\n1) -- (c1t) -- (cc1t)
    (cd1t) -- (d1t) -- (uobr) arc (-45:45:\n1) -- (rootr);
    \path [draw,fill=white] let \p1=($(uib) - (beadcr)$), \n1 = {veclen(\x1,\y1)} in
    (uib) arc (135:225:\n1) arc (-45:45:\n1) -- cycle;
    \node [node] at (u) {};
    \node [node] at (c11) {};
    \node [node] at (c21) {};
    \node [node] at (d12) {};
    \node [node] at (d22) {};
    \node [node] at (d32) {};
    \path [edge] (cc2) -- (c21) -- (q1) (cc1) -- (c11) -- (c21)
    (cd1) -- (d12) -- (d22) (cd2) -- (d22) -- (d32) (cd3) -- (d32) -- (q2) (root) -- (u);
    \path [edge] let \p1 = ($(u) - (beadcl)$), \n1 = {atan2(\y1,\x1)},
    \p2 = ($(v2) - (beadcl)$), \n2 = {atan2(\y2,\x2)},
    \n3 = {veclen(\y1,\x1)} in
    (u) arc (\n1:\n2:\n3) -- (d12);
    \path [edge] let \p1 = ($(u) - (beadcr)$), \n1 = {atan2(\y1,\x1)},
    \p2 = ($(v1) - (beadcr)$), \n2 = {360+atan2(\y2,\x2)},
    \n3 = {veclen(\y1,\x1)} in
    (u) arc (\n1:\n2:\n3) -- (c11);
    \node [anchor=west,xshift=-1pt,yshift=2pt] at (d3t) {$d_3$};
    \node [anchor=west,xshift=-1pt,yshift=2pt] at (d2t) {$d_2$};
    \node [anchor=west,xshift=-1pt,yshift=2pt] at (d1 -| d3t) {$d_1$};
    \node [anchor=east,xshift=1pt,yshift=2pt] at (c2t) {$c_2$};
    \node [anchor=east,xshift=1pt,yshift=2pt] at (c1 -| c2t) {$c_1$};
    \node [anchor=east] at (ql) {$q$};
    \node [anchor=west,yshift=2pt] at (uotl -| uobr) {$u$};
    \node [anchor=west,yshift=-2pt] at (uobr) {$v$};
  \end{tikzpicture}%
  \hspace*{\stretch{1}}%
  \caption{The construction of a bead as in Lemma~\ref{lem:BeadedSolution}. The
    embeddings of one tree into the networks before and after are also
    depicted. The colours in the ``fat'' networks indicate how the embeddings
    change.}
  \label{fig:addTree}
\end{figure}
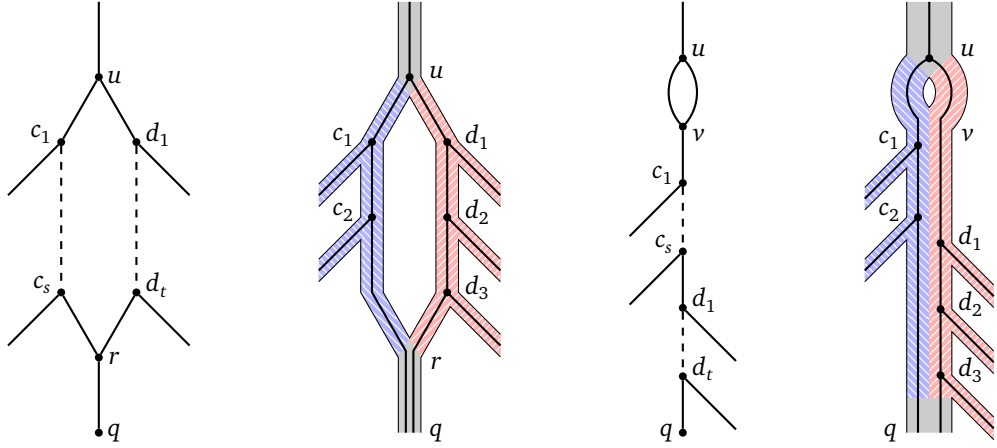

\begin{lemma}\label{lem:BeadedSolution}
  For any set $\mathcal{T}$ of MUL-trees on $X$, there exists a phylogenetic
  network $N$ with $k$ reticulations that weakly displays all MUL-trees in
  $\mathcal{T}$ if and only if there exists a beaded tree $B$ with $k$
  reticulations that weakly displays the MUL-trees in $\mathcal{T}$.
\end{lemma}

\begin{proof}
  The if-direction is trivial because every beaded tree is a phylogenetic
  network.  For the only-if-direction, consider a network $N$ with the maximum
  number of beads among all solutions of \textsc{Parental Hybridization} on
  $\mathcal{T}$ with $k$ reticulations.  If $N$ is a beaded tree, the lemma
  holds.  Otherwise, there is some reticulation node $r$ in $N$ that has two
  different parents $c_s$ and $d_t$.  Let $q$ be the unique child of $r$.  Let
  $u$ be a least common ancestor of $c_s$ and $d_t$ in $N$, let $c_1$ and $d_1$
  be the children of $u$, let $c_1, \ldots, c_s$ be the nodes on a path from
  $c_1$ to $c_s$, and let $d_1, \ldots, d_t$ be the nodes on a path from $d_1$
  to $d_t$.  Note that, by construction, there is no directed path from $d_j$
  to $c_i$, for any $1 \le i \le s$ and $1 \le j \le t$.  These definitions are
  illustrated in Figure~\ref{fig:addTree}.

  We obtain a phylogenetic network $N'$ from $N$ as follows: Delete $r$ and any
  edges incident to it, as well as the edges $uc_1$ and $ud_1$.  Now add a new
  node $v$, a pair of parallel edges from $u$ to $v$, and edges $vc_1, c_sd_1$,
  and $d_tq$.  (Note that this construction assumes that $s,t \geq 1$; if this
  is not the case, then we can produce $N'$ by introducing a ``dummy node''
  $c_1$ or $d_1$ and suppressing it after the construction is complete.)

  Observe that (as there is no path from any node $d_j$ to any node $c_i$ in
  $N$) $N'$ is still an acyclic graph. It follows that $N'$ is a phylogenetic
  network, and it is easy to see that $N'$ has the same number of reticulations
  as $N$ but one more bead than $N$.  We show now that any MUL-tree $T$ weakly
  displayed by $N$ is also weakly displayed by $N'$, from which it follows that
  $N'$ is also a solution to \textsc{Parental Hybridization} on
  $\mathcal{T}$ with $k$ reticulations.  Since $N'$ has one more bead than $N$, this contradicts the
  choice of $N$, that is $N$ must be a beaded tree.

  Let $h$ be a weak embedding of $T$ into $N$.  Then we define a weak embedding
  $h'$ of $T$ into $N'$ as follows.  Since $h(x) \ne r$ for every node $x \in
  T$ and $V(N) \setminus V(N') = \{r\}$, we have $h(x) \in V(N')$ for all $x
  \in T$.  Thus, we can define $h'(x) = h(x)$ for all $x \in T$.  Next observe
  that, for any two nodes $u', v' \in V(N) \setminus \{r\}$, there exists a
  path from $u'$ to $v'$ in $N'$ if there exists such a path in $N$.  Thus,
  since there exists a path $h(xy)$ from $h(x)$ to $h(y)$ in $N$, for every
  edge $xy \in T$, there also exists a path $h'(xy)$ from $h'(x)$ to $h'(y)$ in
  $N'$ for every edge $xy \in T$.  We need to show that we can choose these
  paths such that, for every node $x \in V(T)$ with children $y$ and $y'$, the
  paths $h'(xy)$ and $h'(xy')$ begin with different out-edges of $h'(x)$.

  So consider a node $x$ and its two children $y$ and $y'$ in $T$. If no
  out-edges of $h'(x)$ were deleted in the construction of $N'$, then the
  children of $h'(x)$ are the same in $N'$ as in $N$, and these children are
  still ancestors of $h'(y)$ and $h'(y')$. Thus, the required paths exist.  Now
  assume that at least one out-edge of $h'(x)$ was deleted, from which it
  follows that $h'(x) \in \{u,c_s,d_t\}$.  If $h'(x)=u$, then there are two
  paths from $h'(x)$ to $h'(y)$ and from $h'(x)$ to $h'(y')$ that use different
  out-edges of $h'(x)$, as each path can use a different parallel edge from $u$
  to $v$.  If $h'(x)=c_s$, then one of $\{h'(y),h'(y')\}$ is a descendant of
  $r$ (and therefore a descendant of $q$), and the other is a descendant of the
  other child of $c_s$. Therefore, in $N'$, one of $\{h'(y),h'(y')\}$ is
  descended from $q$, and the other is descended from the child of $c_s$ that
  is not~$d_1$. Thus, the required paths still exist. A similar argument
  applies when $h'(x)=d_t$.  This finishes the proof.
\end{proof}

\section{Structural Properties of Optimal Beaded Trees}\label{sec:structure}

In this section, we prove some of the properties of an optimal solution to an
instance of \textsc{Beaded Tree}.
These properties will both be used in Section~\ref{sec:algorithm} as a basis
for our algorithm for finding an optimal beaded tree for any given instance
and highlight that in fact \emph{every} optimal solution to an instance of
\textsc{Beaded Tree} has a very restrictive structure.

\begin{definition}
  Given a phylogenetic network $N$ on $X$ and a subset $S \subseteq X$, let
  $N\setminus S$ denote the network derived from $N$ by deleting every leaf in
  $S$, and then exhaustively deleting unlabelled nodes of out-degree $0$ and
  suppressing nodes of in-degree $1$ and out-degree $1$.  Let $N|_{S}$ denote
  the network $N \setminus (X \setminus S)$.
\end{definition}

For a set of MUL-trees $\mathcal{T}$, let $F_1(\mathcal{T})$ denote the set of
trees derived by, roughly speaking, deleting the topmost tree node from every
tree.  We make this notion more precise in the following definition.

\begin{definition}
  Given a MUL-tree $T$ with more than one leaf, let $r$ denote the root, $x$
  the child of $r$ and $y_l$ and $y_r$ the children of $x$.  Let $T_l$ be
  derived from $T$ by deleting $y_r$ and all its descendants, and suppressing
  $x$.  Similarly let $T_r$ be derived from $T$ by deleting $y_l$ and all its
  descendants, and suppressing $x$.  Then we call $\{T_l,T_r\}$ the
  \emph{depth-$1$ forest} of $T$, denoted $F_1(T)$. For a set of MUL-trees
  $\mathcal{T}$, we define 
  \begin{equation*}
    F_1(\mathcal{T})=\bigcup_{T\in\mathcal{T}}F_1(T).
  \end{equation*}
\end{definition}

In what follows, we say that a beaded tree $B$ has a \emph{bead at the root} if
the child $u$ of the root node is part of a bead $(u,v)$.

\begin{lemma}\label{lem:topBead}
  Given an instance $\mathcal{T}$ of \textsc{Beaded Tree}, there exists a
  solution $B$ with a bead at the root and reticulation number $k$ if and only
  if $F_1(\mathcal{T})$ has a solution $B'$ with reticulation number $k-1$.
\end{lemma}

\begin{proof}
  Suppose first that $F_1(\mathcal{T})$ has a solution $B'$ with reticulation
  number $k-1$. Let $r$ be the root of $B'$ and $a$ its child.  Construct a
  beaded tree $B$ from $B'$ by deleting the edge $ra$, adding a new bead
  $(u,v)$, and adding edges $ru$ and $va$.  By construction, $B$ is a beaded
  tree with $k$ beads, and it has a bead at the root.

  To see that $B$ is a solution for $\mathcal{T}$, consider any tree $T$ in
  $\mathcal{T}$, and let $\{T_l, T_r\} = F_1(T)$.  Let $r_T$ be the root of
  $T$, $x$~its child, and $y_l$ and $y_r$ the children of $x$, with $y_l \in
  V(T_l)$ and $y_r \in V(T_r)$.  Since $B'$ is a solution for
  $F_1(\mathcal{T})$, there exist weak embeddings $h_l$ and $h_r$ of $T_l$ and
  $T_r$, respectively, into $B'$.  Construct a weak embedding $h$ of $T$ into
  $B$ as follows: Let $h(r_T)=r$, $h(x) = u$, and for all other nodes $x' \in
  V(T)$, $h(x') = h_l(x')$ if $x' \in V(T_l)$, and $h(x') = h_r(x')$ if $x' \in
  V(T_r)$.  Let $h(r_Tx)$ be the path from $r$ to $u$, $h(x'y')=h_l(x'y')$ if
  $x'y' \in E(T_l)$, and $h(x'y')=h_r(x'y')$ if $x'y' \in E(T_r)$.  Finally,
  let $h(xy_l)$ be a path from $u$ to $h(y_l)$, and $h(xy_r)$ a path from $u$
  to $h(y_r)$, letting those two paths start with different out-edges of $u$.
  It is easy to see that $h$ is a weak embedding of $T$ into $B$, so $B$ is a
  solution for $\mathcal{T}$.

  Conversely, suppose that $\mathcal{T}$ has a solution $B$ with a bead $(u,v)$
  at the root and reticulation number $k$; see Figure~\ref{fig:manyFigures}.
  Let $r$ be the root of $B$ and $z$ the child of $v$.  Let $B'$ be the network
  derived from $B$ by deleting $u$ and $v$ and adding an edge $rz$.  By
  construction, $B'$ is a beaded tree with reticulation number $k-1$.

  To see that $B'$ is a solution for $F_1(\mathcal{T})$, consider any tree $T$
  in $\mathcal{T}$, and let $\{T_l, T_r\} = F_1(T)$.  Let $r_T$ be the root of
  $T$, $x$ its child, and $y_l$ and $y_r$ the children of $x$, with $y_l \in
  V(T_l)$ and $y_r \in V(T_r)$.  Since $B$ is a solution for $\mathcal{T}$,
  there exists a weak embedding $h$ of $T$ into $B$.  Observe that $h(x')$ must
  be a strict descendant of $v$ for any strict descendant $x'$ of $x$ (indeed,
  $u$ is the earliest node that $x$ could be mapped to and any strict
  descendant of $x$ must be mapped to a tree node strictly descended from this
  point).  So we can define a weak embedding $h_l$ of $T_l$ into $B'$ by
  letting $h_l(r_T) = r$ and $h_l(x') = h(x')$ for every node $x' \ne r_T \in
  T_l$, letting $h_l(r_Ty_l)$ be a path in $B'$ from $r$ to $h_l(y_l)$, and
  letting $h_l(e) = h(e)$ for any other edge $e \in T_l$.  By a similar method,
  we can define a weak embedding $h_r$ of $T_r$ into~$B'$.  Thus, $B'$ is a
  solution for $F_1(\mathcal{T})$, as required.
\end{proof}

\begin{figure}[t]
  \hspace*{\stretch{1}}%
  \subcaptionbox{}{%
    \begin{tikzpicture}[
      node/.style={fill=black,circle,inner sep=0pt,minimum size=3pt,outer sep=0pt},
      dup/.style={fill=red,circle,inner sep=0pt,minimum size=6pt},
      wide node/.style={fill=black,rectangle,inner sep=0pt,minimum height=3pt,rounded corners=1.5pt},
      edge/.style={draw,thick},
      x=0.7cm,y=0.7cm]
      \foreach \i in {0,...,5} {
        \node [node] (l\i) at (\i,0) {};
      }
      \path (l0) +(60:1) node [node] (p01) {};
      \path (l0) +(60:2) node [node] (p02) {};
      \path (l0) +(60:5) node [node] (p05) {};
      \path (l3) +(60:2) node [node] (p35) {};
      \path (l4) +(60:1) node [node] (p45) {};
      \path (p05) +(90:1) node [node] (r) {};
      \path [edge] (l0) -- (p01) -- (l1) (p01) -- (p02) -- (l2)
      (l4) -- (p45) -- (l5) (l3) -- (p35) -- (p45)
      (p02) -- (p05) -- (p35) (p05) -- (r);
      \node [anchor=north,text height=height("$b$")] at (l0) {$a$};
      \node [anchor=north,text height=height("$b$")] at (l1) {$a$};
      \node [anchor=north,text height=height("$b$")] at (l2) {$b$};
      \node [anchor=north,text height=height("$b$")] at (l3) {$a$};
      \node [anchor=north,text height=height("$b$")] at (l4) {$b$};
      \node [anchor=north,text height=height("$b$")] at (l5) {$c$};
      \node [anchor=east,yshift=2pt] at (p02) {$y_l$};
      \node [anchor=west,yshift=2pt] at (p35) {$y_r$};
      \node [anchor=west,yshift=2pt] at (p05) {$x$};
      \node [anchor=west] at (r) {$r_T$};
    \end{tikzpicture}}%
  \hspace*{\stretch{1}}%
  \subcaptionbox{}{%
    \begin{tikzpicture}[
      node/.style={fill=black,circle,inner sep=0pt,minimum size=3pt,outer sep=0pt},
      dup/.style={fill=red,circle,inner sep=0pt,minimum size=6pt},
      wide node/.style={fill=black,rectangle,inner sep=0pt,minimum height=3pt,rounded corners=1.5pt},
      edge/.style={draw,thick},
      x=0.7cm,y=0.7cm]
      \foreach \i/\x in {0/0,1/1.5,2/3.0} {
        \node [node] (l\i) at (\x,0) {};
      }
      \path (l0) +(60:1.5) coordinate (p0);
      \path (l0) +(60:3.0) node [node] (p02) {};
      \path (l1) +(60:1.5) node [node] (p12) {};
      \path (p0) +(240:0.5) node [node] (p0b) {};
      \path (p0) +(60:0.5) node [node] (p0t) {};
      \node [node] (r) at (p02 |- r) {};
      \coordinate (pp02) at (barycentric cs:r=0.5,p02=0.5);
      \path (pp02) +(270:0.5) node [node] (pp02b) {};
      \path (pp02) +(90:0.5) node [node] (pp02t) {};
      \path [name path=path1] (pp02t) -- +(225:1);
      \path [name path=path2] (pp02b) -- +(135:1);
      \path [name intersections={of=path1 and path2}] coordinate (beadc) at (intersection-1);
      \path [edge] (l1) -- (p12) -- (l2) (l0) -- (p0b) (p0t) -- (p02) -- (p12) (p02) -- (pp02b)
      (pp02t) -- (r);
      \path [edge] let \p1 = ($(pp02t) - (beadc)$), \n1 = {veclen(\x1,\y1)} in
      (pp02t) arc (135:225:\n1) arc (-45:45:\n1) -- cycle
      (p0t) arc (105:195:\n1) arc (-75:15:\n1) -- cycle;
      \node [anchor=north,text height=height("$b$")] at (l0) {$a$};
      \node [anchor=north,text height=height("$b$")] at (l1) {$b$};
      \node [anchor=north,text height=height("$b$")] at (l2) {$c$};
      \node [anchor=west] at (r) {$r$};
      \node [anchor=west,yshift=2pt] at (pp02t) {$u$};
      \node [anchor=west,yshift=-2pt] at (pp02b) {$v$};
      \node [anchor=west,yshift=2pt] at (p02) {$z$};
    \end{tikzpicture}}%
  \hspace*{\stretch{1}}%
  \subcaptionbox{}{%
    \begin{tikzpicture}[
      node/.style={fill=black,circle,inner sep=0pt,minimum size=3pt,outer sep=0pt},
      dup/.style={fill=red,circle,inner sep=0pt,minimum size=6pt},
      wide node/.style={fill=black,rectangle,inner sep=0pt,minimum height=3pt,rounded corners=1.5pt},
      edge/.style={draw,thick},
      x=0.7cm,y=0.7cm]
      \foreach \i in {0,...,5} {
        \node [node] (l\i) at (\i,0) {};
      }
      \path (l0) +(60:1) node [node] (p01) {};
      \path (l0) +(60:2) node [node] (p02) {};
      \path (l3) +(60:2) node [node] (p35) {};
      \path (l4) +(60:1) node [node] (p45) {};
      \path (p02) +(90:1) node [node] (rl) {};
      \path (p35) +(90:1) node [node] (rr) {};
      \path [edge] (l0) -- (p01) -- (l1) (p01) -- (p02) -- (l2)
      (l4) -- (p45) -- (l5) (l3) -- (p35) -- (p45)
      (p02) -- (rl) (p35) -- (rr);
      \node [anchor=north,text height=height("$b$")] at (l0) {$a$};
      \node [anchor=north,text height=height("$b$")] at (l1) {$a$};
      \node [anchor=north,text height=height("$b$")] at (l2) {$b$};
      \node [anchor=north,text height=height("$b$")] at (l3) {$a$};
      \node [anchor=north,text height=height("$b$")] at (l4) {$b$};
      \node [anchor=north,text height=height("$b$")] at (l5) {$c$};
      \node [anchor=west,yshift=2pt] at (p02) {$y_l$};
      \node [anchor=west,yshift=2pt] at (p35) {$y_r$};
      \node [anchor=west] at (rl) {$r_T$};
      \node [anchor=west] at (rr) {$r_T$};
    \end{tikzpicture}}%
  \hspace*{\stretch{1}}%
  \subcaptionbox{}{%
    \begin{tikzpicture}[
      node/.style={fill=black,circle,inner sep=0pt,minimum size=3pt,outer sep=0pt},
      dup/.style={fill=red,circle,inner sep=0pt,minimum size=6pt},
      wide node/.style={fill=black,rectangle,inner sep=0pt,minimum height=3pt,rounded corners=1.5pt},
      edge/.style={draw,thick},
      x=0.7cm,y=0.7cm]
      \foreach \i/\x in {0/0,1/1.5,2/3.0} {
        \node [node] (l\i) at (\x,0) {};
      }
      \path (l0) +(60:1.5) coordinate (p0);
      \path (l0) +(60:3.0) node [node] (p02) {};
      \path (l1) +(60:1.5) node [node] (p12) {};
      \path (p0) +(240:0.5) node [node] (p0b) {};
      \path (p0) +(60:0.5) node [node] (p0t) {};
      \path (p02) +(90:1) node [node] (realr) {};
      \coordinate (r) at (p02 |- r);
      \coordinate (pp02) at (barycentric cs:r=0.5,p02=0.5);
      \path (pp02) +(270:0.5) coordinate (pp02b);
      \path (pp02) +(90:0.5) coordinate (pp02t);
      \path [name path=path1] (pp02t) -- +(225:1);
      \path [name path=path2] (pp02b) -- +(135:1);
      \path [name intersections={of=path1 and path2}] coordinate (beadc) at (intersection-1);
      \path [edge] (l1) -- (p12) -- (l2) (l0) -- (p0b) (p0t) -- (p02) -- (p12) (p02) -- (realr);
      \path [edge] let \p1 = ($(pp02t) - (beadc)$), \n1 = {veclen(\x1,\y1)} in
      (p0t) arc (105:195:\n1) arc (-75:15:\n1) -- cycle;
      \node [anchor=north,text height=height("$b$")] at (l0) {$a$};
      \node [anchor=north,text height=height("$b$")] at (l1) {$b$};
      \node [anchor=north,text height=height("$b$")] at (l2) {$c$};
      \node [anchor=west] at (realr) {$r$};
      \node [anchor=west,yshift=2pt] at (p02) {$z$};
    \end{tikzpicture}}%
  \hspace*{\stretch{1}}%
  \caption{Example of Lemma \ref{lem:topBead}.  In all figures, lowercase
    letters on leaves represent labels from the set $X$.  The other labels are
    as described in Lemma \ref{lem:topBead}.  (a) A MUL-tree $T$ on $X =
    \{a,b,c\}$.  (b) A beaded tree $B$ weakly displaying $T$.  (c) The two
    MUL-trees $T_l$ and $T_r$ in $F_1(T)$.  (d) The beaded tree $B'$ derived
    from $B$ by suppressing the nodes of the top bead, which weakly displays
    the trees in~$F_1(T)$.}
  \label{fig:manyFigures}
\end{figure}
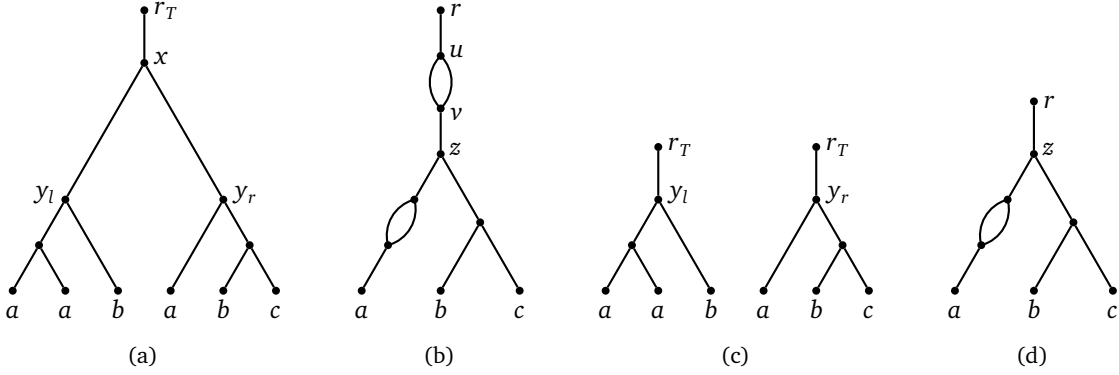

In the same way that Lemma~\ref{lem:BeadedSolution} establishes that
\textsc{Parental Hybridization} always has an optimal solution that is a beaded
tree, the following lemma shows that there always exists an optimal solution to
\textsc{Beaded Tree} of an even more restrictive structure.

\begin{lemma}\label{lem:allOneLineage}
  Every instance $\mathcal{T}$ of \textsc{Beaded Tree} has an optimal solution
  $B$ such that all reticulations are on the same path.
\end{lemma}

\begin{proof}
  Consider an optimal solution $B$ for $\mathcal{T}$.  For each reticulation
  node $z \in B$, let $\lambda_B(z)$ be the number of reticulation nodes
  strictly descended from $z$.  Let $\lambda(B)$ be the sum of
  $\lambda_B(z)$ over all reticulation nodes $z$ in $B$.  Choose $B$ such that
  $\lambda(B)$ is maximized.  Since all optimal solutions for $\mathcal{T}$
  have the same number $b$ of beads and $\lambda(B') \le \binom{b-1}{2}$ for
  any beaded tree $B'$ with $b$ beads, an optimal solution $B$ for
  $\mathcal{T}$ that maximizes $\lambda(B)$ exists.

  If all reticulations in $B$ are on the same path, the lemma holds.  So assume
  that not all reticulations are on the same path; see
  Figure~\ref{fig:OneLineageLocal}.  Then there is some tree node $b$ in $B$
  that is not in a bead and such that both children of $b$ are ancestors of a
  bead.  Let $(u_L,v_L$) be an earliest bead descended from one child of $b$,
  and $(u_R,v_R)$ an earliest bead descended from the other child of $b$.  If
  $u_L$ is not a child of $b$, then let $c_1, \ldots c_l$ be the nodes on the
  path from $b$ to $u_L$.  Similarly, if $u_R$ is not a child of $b$, then let
  $d_1, \ldots d_r$ be the nodes on the path from $b$ to $u_R$.  Note that $c_1,
  \ldots c_l$ and $d_1, \ldots d_r$ are all tree nodes.  Finally let $w_L$ be the
  single child of $v_L$, and $w_R$ the single child of $v_R$.

  Construct a new beaded tree $B'$ from $B$ as follows: Delete the nodes
  $u_L,v_L,u_R,v_R$ and any edges incident to them, as well as the edges $bc_1$
  and $bd_1$.  Now add new nodes $q,u,v,w$ and add a pair of parallel arcs from
  $b$ to $q$ and from $u$ to $v$, as well as arcs $qc_1,c_ld_1, d_ru, vw,
  ww_L$, and $ww_R$.  (Note that this construction assumes that $l,r \geq 1$;
  if this is not the case, then we may produce $B'$ by introducing ``dummy
  nodes'' $c_1$ and $d_1$ and suppressing them after the construction is
  complete.) Observe that this construction ensures that every node $u' \in V'
  = V(B) \setminus \{u_L, v_L, u_R, v_R\}$ is an ancestor of a node $v' \in V'$
  in $B'$ if this is the case in $B$, that every node $v' \in V'$ that is a
  descendant of $u_L$ or $u_R$ in $B$ is a descendant of $u$ in $B'$, and that
  every node $v' \in V'$ that is an ancestor of $u_L$ or $u_R$ in $B$ is an
  ancestor of $u$ in $B'$.

  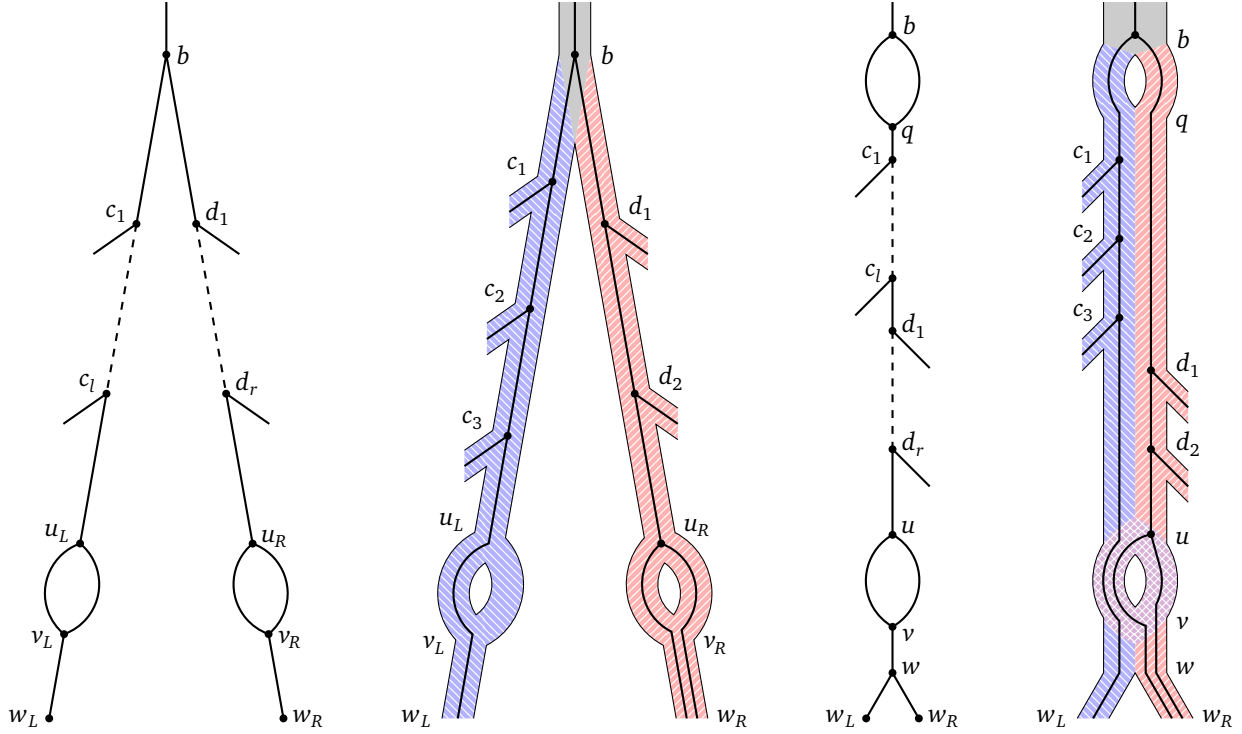
\begin{figure}[t]
    \begin{tikzpicture}[
      node/.style={fill=black,circle,inner sep=0pt,minimum size=3pt,outer sep=0pt},
      dup/.style={fill=red,circle,inner sep=0pt,minimum size=6pt},
      wide node/.style={fill=black,rectangle,inner sep=0pt,minimum height=3pt,rounded corners=1.5pt},
      edge/.style={draw,thick},
      x=0.7cm,y=0.7cm]
      \begin{scope}[overlay]
        \path let \n1 = {(12.75+sin(60))*0.7cm} in coordinate (r) at (0,\n1);
        \coordinate (l1) at (0,0);
        \path [name path=path1] (r) ++(270:1) -- +(180:5);
        \path [name path=path2] (l1) -- +(100:15);
        \path [name intersections={of=path1 and path2}] coordinate (b) at (intersection-1);
        \path (b) +(90:1) coordinate (r);
        \path [name path=path1] (b) -- +(260:15);
        \path [name path=path2] (l1) -- +(180:7);
        \path [name intersections={of=path1 and path2}] coordinate (l0) at (intersection-1);
        \path (l1) +(100:2) coordinate (vr) +(100:3) coordinate (ur);
        \path (l0) +(80:2) coordinate (vl) +(80:3) coordinate (ul);
        \coordinate (c1) at (barycentric cs:b=0.667,ul=0.333);
        \coordinate (c2) at (barycentric cs:b=0.333,ul=0.667);
        \coordinate (d1) at (barycentric cs:b=0.667,ur=0.333);
        \coordinate (d2) at (barycentric cs:b=0.333,ur=0.667);
        \path (c1) +(215:1) coordinate (cc1);
        \path (c2) +(215:1) coordinate (cc2);
        \path (d1) +(-35:1) coordinate (cd1);
        \path (d2) +(-35:1) coordinate (cd2);
        \path [name path=path1] (ul) -- +(215:1.25);
        \path [name path=path2] (vl) -- +(125:1.25);
        \path [name intersections={of=path1 and path2}] coordinate (lbeadcl) at (intersection-1);
        \path [name path=path1] (ul) -- +(-55:1.25);
        \path [name path=path2] (vl) -- +(35:1.25);
        \path [name intersections={of=path1 and path2}] coordinate (lbeadcr) at (intersection-1);
        \path [name path=path1] (ur) -- +(235:1.25);
        \path [name path=path2] (vr) -- +(145:1.25);
        \path [name intersections={of=path1 and path2}] coordinate (rbeadcl) at (intersection-1);
        \path [name path=path1] (ur) -- +(-35:1.25);
        \path [name path=path2] (vr) -- +(55:1.25);
        \path [name intersections={of=path1 and path2}] coordinate (rbeadcr) at (intersection-1);
        \path [name path=path1]
        let \p1 = ($(ul) - (lbeadcr)$), \n1 = {0.21cm+veclen(\x1,\y1)} in
        (lbeadcr) circle (\n1);
        \path [name path=path2] (l0) -- (b);
        \path [name intersections={of=path1 and path2}] coordinate (ulm) at (intersection-1)
        coordinate (vlm) at (intersection-2);
        \path [name path=path1]
        let \p1 = ($(ur) - (rbeadcr)$), \n1 = {0.21cm+veclen(\x1,\y1)} in
        (rbeadcr) circle (\n1);
        \path [name path=path2] (b) -- (l1);
        \path [name intersections={of=path1 and path2}] coordinate (urm) at (intersection-1)
        coordinate (vrm) at (intersection-2);
      \end{scope}
      \path [edge] (r) -- (b);
      \path [edge] (cc1) -- (c1) -- (b) -- (d1) -- (cd1);
      \path [edge] let \p1 = ($(ulm) - (lbeadcr)$), \p2 = ($(vlm) - (lbeadcr)$),
      \n1 = {atan2(\y1,\x1)}, \n2 = {360+atan2(\y2,\x2)}, \n3 = {veclen(\x1,\y1)} in
      (cc2) -- (c2) -- (ulm) arc (\n1:\n2:\n3) -- (l0);
      \path [edge] let \p1 = ($(ulm) - (lbeadcl)$), \p2 = ($(vlm) - (lbeadcl)$),
      \n1 = {atan2(\y1,\x1)}, \n2 = {atan2(\y2,\x2)}, \n3 = {veclen(\x1,\y1)} in
      (ulm) arc (\n1:\n2:\n3);
      \path [edge] (cd2) -- (d2) -- (urm);
      \path [edge] let \p1 = ($(urm) - (rbeadcr)$), \p2 = ($(vrm) - (rbeadcr)$),
      \n1 = {atan2(\y1,\x1)}, \n2 = {360+atan2(\y2,\x2)}, \n3 = {veclen(\x1,\y1)} in
      (urm) arc (\n1:\n2:\n3) -- (l1);
      \path [edge] let \p1 = ($(urm) - (rbeadcl)$), \p2 = ($(vrm) - (rbeadcl)$),
      \n1 = {atan2(\y1,\x1)}, \n2 = {atan2(\y2,\x2)}, \n3 = {veclen(\x1,\y1)} in
      (urm) arc (\n1:\n2:\n3);
      \path [edge,dashed] (c1) -- (c2) (d1) -- (d2);
      \node [node] at (b) {};
      \node [node] at (c1) {};
      \node [node] at (c2) {};
      \node [node] at (d1) {};
      \node [node] at (d2) {};
      \node [node] at (ulm) {};
      \node [node] at (vlm) {};
      \node [node] at (urm) {};
      \node [node] at (vrm) {};
      \node [node] at (l0) {};
      \node [node] at (l1) {};
      \node [anchor=east] at (l0) {$w_L$};
      \node [anchor=west] at (l1) {$w_R$};
      \node [anchor=west,yshift=2pt] at (urm) {$u_R$};
      \node [anchor=west,yshift=-2pt] at (vrm) {$v_R$};
      \node [anchor=west] at (b) {$b$};
      \node [anchor=west,yshift=4pt] at (d1) {$d_1$};
      \node [anchor=west,yshift=4pt] at (d2) {$d_r$};
      \node [anchor=east,yshift=4pt] at (ulm) {$u_L$};
      \node [anchor=east,yshift=-2pt] at (vlm) {$v_L$};
      \node [anchor=east,yshift=4pt] at (c1) {$c_1$};
      \node [anchor=east,yshift=4pt] at (c2) {$c_l$};
    \end{tikzpicture}%
    \hspace*{\stretch{1}}%
    \begin{tikzpicture}[
      node/.style={fill=black,circle,inner sep=0pt,minimum size=3pt,outer sep=0pt},
      dup/.style={fill=red,circle,inner sep=0pt,minimum size=6pt},
      wide node/.style={fill=black,rectangle,inner sep=0pt,minimum height=3pt,rounded corners=1.5pt},
      edge/.style={draw,thick},
      x=0.7cm,y=0.7cm]
      \begin{scope}[overlay]
        \path let \n1 = {(12.75+sin(60))*0.7cm} in coordinate (r) at (0,\n1);
        \coordinate (l1) at (0,0);
        \path [name path=path1] (r) ++(270:1) -- +(180:5);
        \path [name path=path2] (l1) -- +(100:15);
        \path [name intersections={of=path1 and path2}] coordinate (b) at (intersection-1);
        \path (b) +(90:1) coordinate (r);
        \path [name path=path1] (b) -- +(260:15);
        \path [name path=path2] (l1) -- +(180:7);
        \path [name intersections={of=path1 and path2}] coordinate (l0) at (intersection-1);
        \path (l1) +(100:2) coordinate (vr) +(100:3) coordinate (ur);
        \path (l0) +(80:2) coordinate (vl) +(80:3) coordinate (ul);
        \coordinate (c1) at (barycentric cs:b=0.75,ul=0.25);
        \coordinate (c2) at (barycentric cs:b=0.5,ul=0.5);
        \coordinate (c3) at (barycentric cs:b=0.25,ul=0.75);
        \coordinate (d1) at (barycentric cs:b=0.667,ur=0.333);
        \coordinate (d2) at (barycentric cs:b=0.333,ur=0.667);
        \path (c1) +(215:1) coordinate (cc1);
        \path (c2) +(215:1) coordinate (cc2);
        \path (c3) +(215:1) coordinate (cc3);
        \path (d1) +(-35:1) coordinate (cd1);
        \path (d2) +(-35:1) coordinate (cd2);
        \foreach \v in {l0,l1,r,b} {
          \path (\v) +(180:0.1) coordinate (\v 1) +(180:0.3) coordinate (\v l)
          +(0:0.1) coordinate (\v 2) +(0:0.3) coordinate (\v r);
        }
        \foreach \v in {cc1,cc2,cc3,cd1,cd2} {
          \path (\v) +(270:0.3) coordinate (\v b) +(90:0.3) coordinate (\v t);
        }
        \path [name path=path1] (ul) -- +(215:1.25);
        \path [name path=path2] (vl) -- +(125:1.25);
        \path [name intersections={of=path1 and path2}] coordinate (lbeadcl) at (intersection-1);
        \path [name path=path1] (ul) -- +(-55:1.25);
        \path [name path=path2] (vl) -- +(35:1.25);
        \path [name intersections={of=path1 and path2}] coordinate (lbeadcr) at (intersection-1);
        \path [name path=path1] (ur) -- +(235:1.25);
        \path [name path=path2] (vr) -- +(145:1.25);
        \path [name intersections={of=path1 and path2}] coordinate (rbeadcl) at (intersection-1);
        \path [name path=path1] (ur) -- +(-35:1.25);
        \path [name path=path2] (vr) -- +(55:1.25);
        \path [name intersections={of=path1 and path2}] coordinate (rbeadcr) at (intersection-1);
        \path [name path=path1] (l0l) -- (bl);
        \path [name path=path2]
        let \p1 = ($(ul) - (lbeadcr)$), \n1 = {0.42cm+veclen(\x1,\y1)} in
        (lbeadcr) circle (\n1);
        \path [name intersections={of=path1 and path2}] coordinate (ull) at (intersection-1)
        coordinate (vll) at (intersection-2);
        \path [name path=path2] (cc3b) -- +(35:1);
        \path [name intersections={of=path1 and path2}] coordinate (c3b) at (intersection-1);
        \path [name path=path2] (cc3t) -- +(35:1);
        \path [name intersections={of=path1 and path2}] coordinate (c3t) at (intersection-1);
        \path [name path=path2] (cc2b) -- +(35:1);
        \path [name intersections={of=path1 and path2}] coordinate (c2b) at (intersection-1);
        \path [name path=path2] (cc2t) -- +(35:1);
        \path [name intersections={of=path1 and path2}] coordinate (c2t) at (intersection-1);
        \path [name path=path2] (cc1b) -- +(35:1);
        \path [name intersections={of=path1 and path2}] coordinate (c1b) at (intersection-1);
        \path [name path=path2] (cc1t) -- +(35:1);
        \path [name intersections={of=path1 and path2}] coordinate (c1t) at (intersection-1);
        \path [name path=path1] (l0r) -- (br);
        \path [name path=path2]
        let \p1 = ($(ul) - (lbeadcl)$), \n1 = {0.42cm+veclen(\x1,\y1)} in
        (lbeadcl) circle (\n1);
        \path [name intersections={of=path1 and path2}] coordinate (ulr) at (intersection-1)
        coordinate (vlr) at (intersection-2);
        \path [name path=path2] (l1l) -- (bl);
        \path [name intersections={of=path1 and path2}] coordinate (bb) at (intersection-1);
        \path [name path=path1]
        let \p1 = ($(ur) - (rbeadcr)$), \n1 = {0.42cm+veclen(\x1,\y1)} in
        (rbeadcr) circle (\n1);
        \path [name intersections={of=path1 and path2}] coordinate (url) at (intersection-1)
        coordinate (vrl) at (intersection-2);
        \path [name path=path1] (l1r) -- (br);
        \path [name path=path2] (cd1t) -- +(145:1);
        \path [name intersections={of=path1 and path2}] coordinate (d1t) at (intersection-1);
        \path [name path=path2] (cd1b) -- +(145:1);
        \path [name intersections={of=path1 and path2}] coordinate (d1b) at (intersection-1);
        \path [name path=path2] (cd2t) -- +(145:1);
        \path [name intersections={of=path1 and path2}] coordinate (d2t) at (intersection-1);
        \path [name path=path2] (cd2b) -- +(145:1);
        \path [name intersections={of=path1 and path2}] coordinate (d2b) at (intersection-1);
        \path [name path=path2]
        let \p1 = ($(ur) - (rbeadcl)$), \n1 = {0.42cm+veclen(\x1,\y1)} in
        (rbeadcl) circle (\n1);
        \path [name intersections={of=path1 and path2}] coordinate (urr) at (intersection-1)
        coordinate (vrr) at (intersection-2);
        \path [name path=path1]
        let \p1 = ($(ul) - (lbeadcr)$), \n1 = {0.21cm+veclen(\x1,\y1)} in
        (lbeadcr) circle (\n1);
        \path [name path=path2] (l0) -- (b);
        \path [name intersections={of=path1 and path2}] coordinate (ulm) at (intersection-1)
        coordinate (vlm) at (intersection-2);
        \path [name path=path1]
        let \p1 = ($(ur) - (rbeadcr)$), \n1 = {0.21cm+veclen(\x1,\y1)} in
        (rbeadcr) circle (\n1);
        \path [name path=path2]
        (d2) -- (ur);
        \path [name intersections={of=path1 and path2}] coordinate (urm) at (intersection-1);
        \path [name path=path2]
        (l11) -- +(100:2);
        \path [name intersections={of=path1 and path2}] coordinate (vr1) at (intersection-1);
        \path [name path=path1]
        let \p1 = ($(ur) - (rbeadcl)$), \n1 = {0.21cm+veclen(\x1,\y1)} in
        (rbeadcl) circle (\n1);
        \path [name path=path2]
        (l12) -- +(100:2);
        \path [name intersections={of=path1 and path2}] coordinate (vr2) at (intersection-1);
      \end{scope}
      \path [fill=black!20] (rl) -- (bl) -- (bb) -- (br) -- (rr) -- cycle;
      \path [fill=blue!30]
      let \p1 = ($(ull) - (lbeadcr)$), \p2 = ($(vll) - (lbeadcr)$),
      \n1 = {atan2(\y1,\x1)}, \n2 = {360+atan2(\y2,\x2)}, \n3 = {veclen(\x1,\y1)},
      \p4 = ($(vlr) - (lbeadcl)$), \p5 = ($(ulr) - (lbeadcl)$),
      \n4 = {atan2(\y4,\x4)}, \n5 = {atan2(\y5,\x5)}, \n6 = {veclen(\x4,\y4)} in
      (bl) -- (c1t) -- (cc1t) -- (cc1b) -- (c1b) -- (c2t) -- (cc2t)
      -- (cc2b) -- (c2b) -- (c3t) -- (cc3t) -- (cc3b) -- (c3b) -- (ull) arc (\n1:\n2:\n3) -- (l0l)
      -- (l0r) -- (vlr) arc (\n4:\n5:\n6) -- (bb) -- cycle;
      \path [pattern=north west lines,pattern color=white]
      let \p1 = ($(ull) - (lbeadcr)$), \p2 = ($(vll) - (lbeadcr)$),
      \n1 = {atan2(\y1,\x1)}, \n2 = {360+atan2(\y2,\x2)}, \n3 = {veclen(\x1,\y1)},
      \p4 = ($(vlr) - (lbeadcl)$), \p5 = ($(ulr) - (lbeadcl)$),
      \n4 = {atan2(\y4,\x4)}, \n5 = {atan2(\y5,\x5)}, \n6 = {veclen(\x4,\y4)} in
      (bl) -- (c1t) -- (cc1t) -- (cc1b) -- (c1b) -- (c2t) -- (cc2t)
      -- (cc2b) -- (c2b) -- (c3t) -- (cc3t) -- (cc3b) -- (c3b) -- (ull) arc (\n1:\n2:\n3) -- (l0l)
      -- (l0r) -- (vlr) arc (\n4:\n5:\n6) -- (bb) -- cycle;
      \path [fill=red!30]
      let \p1 = ($(url) - (rbeadcr)$), \p2 = ($(vrl) - (rbeadcr)$),
      \n1 = {atan2(\y1,\x1)}, \n2 = {360+atan2(\y2,\x2)}, \n3 = {veclen(\x1,\y1)},
      \p4 = ($(vrr) - (rbeadcl)$), \p5 = ($(urr) - (rbeadcl)$),
      \n4 = {atan2(\y4,\x4)}, \n5 = {atan2(\y5,\x5)}, \n6 = {veclen(\x4,\y4)} in
      (bb) -- (url) arc (\n1:\n2:\n3) -- (l1l) -- (l1r) -- (vrr) arc (\n4:\n5:\n6) -- (d2b)
      -- (cd2b) -- (cd2t) -- (d2t) -- (d1b) -- (cd1b) -- (cd1t) -- (d1t) -- (br) -- cycle;
      \path [pattern=north east lines,pattern color=white]
      let \p1 = ($(url) - (rbeadcr)$), \p2 = ($(vrl) - (rbeadcr)$),
      \n1 = {atan2(\y1,\x1)}, \n2 = {360+atan2(\y2,\x2)}, \n3 = {veclen(\x1,\y1)},
      \p4 = ($(vrr) - (rbeadcl)$), \p5 = ($(urr) - (rbeadcl)$),
      \n4 = {atan2(\y4,\x4)}, \n5 = {atan2(\y5,\x5)}, \n6 = {veclen(\x4,\y4)} in
      (bb) -- (url) arc (\n1:\n2:\n3) -- (l1l) -- (l1r) -- (vrr) arc (\n4:\n5:\n6) -- (d2b)
      -- (cd2b) -- (cd2t) -- (d2t) -- (d1b) -- (cd1b) -- (cd1t) -- (d1t) -- (br) -- cycle;
      \path [draw,fill=white] let \p1 = ($(ul) - (lbeadcr)$), \p2 = ($(vl) - (lbeadcr)$),
      \p3 = ($(vl) - (lbeadcl)$), \p4 = ($(ul) - (lbeadcl)$),
      \n1 = {atan2(\y1,\x1)}, \n2 = {360+atan2(\y2,\x2)}, \n3 = {veclen(\x1,\y1)},
      \n4 = {atan2(\y3,\x3)}, \n5 = {atan2(\y4,\x4)} in
      (ul) arc (\n1:\n2:\n3) arc (\n4:\n5:\n3) -- cycle;
      \path [draw,fill=white] let \p1 = ($(ur) - (rbeadcr)$), \p2 = ($(vr) - (rbeadcr)$),
      \p3 = ($(vr) - (rbeadcl)$), \p4 = ($(ur) - (rbeadcl)$),
      \n1 = {atan2(\y1,\x1)}, \n2 = {360+atan2(\y2,\x2)}, \n3 = {veclen(\x1,\y1)},
      \n4 = {atan2(\y3,\x3)}, \n5 = {atan2(\y4,\x4)} in
      (ur) arc (\n1:\n2:\n3) arc (\n4:\n5:\n3) -- cycle;
      \path [draw] (rl) -- (bl) -- (c1t) -- (cc1t);
      \path [draw] (cc1b) -- (c1b) -- (c2t) -- (cc2t);
      \path [draw] (cc2b) -- (c2b) -- (c3t) -- (cc3t);
      \path [draw] let \p1 = ($(ull) - (lbeadcr)$), \p2 = ($(vll) - (lbeadcr)$),
      \n1 = {atan2(\y1,\x1)}, \n2 = {360+atan2(\y2,\x2)}, \n3 = {veclen(\x1,\y1)} in
      (cc3b) -- (c3b) -- (ull) arc (\n1:\n2:\n3) -- (l0l);
      \path [draw] let \p1 = ($(vlr) - (lbeadcl)$), \p2 = ($(ulr) - (lbeadcl)$),
      \n1 = {atan2(\y1,\x1)}, \n2 = {atan2(\y2,\x2)}, \n3 = {veclen(\x1,\y1)},
      \p4 = ($(url) - (rbeadcr)$), \p5 = ($(vrl) - (rbeadcr)$),
      \n4 = {atan2(\y4,\x4)}, \n5 = {360+atan2(\y5,\x5)}, \n6 = {veclen(\x4,\y4)} in
      (l0r) -- (vlr) arc (\n1:\n2:\n3) -- (bb) -- (url) arc (\n4:\n5:\n6) -- (l1l);
      \path [draw] (rr) -- (br) -- (d1t) -- (cd1t);
      \path [draw] (cd1b) -- (d1b) -- (d2t) -- (cd2t);
      \path [draw] let \p1 = ($(urr) - (rbeadcl)$), \p2 = ($(vrr) - (rbeadcl)$),
      \n1 = {atan2(\y1,\x1)}, \n2 = {atan2(\y2,\x2)}, \n3 = {veclen(\x1,\y1)} in
      (cd2b) -- (d2b) -- (urr) arc (\n1:\n2:\n3) -- (l1r);
      \path [edge] (r) -- (b);
      \path [edge] (c1) -- (b) -- (d1);
      \path [edge] (cc1) -- (c1) -- (c2);
      \path [edge] (cc2) -- (c2) -- (c3);
      \path [edge] let \p1 = ($(ulm) - (lbeadcr)$), \p2 = ($(vlm) - (lbeadcr)$),
      \n1 = {atan2(\y1,\x1)}, \n2 = {360+atan2(\y2,\x2)}, \n3 = {veclen(\x1,\y1)} in
      (cc3) -- (c3) -- (ulm) arc (\n1:\n2:\n3) -- (l0);
      \path [edge] (cd1) -- (d1) -- (d2);
      \path [edge] (cd2) -- (d2) -- (urm);
      \path [edge] let \p1 = ($(urm) - (rbeadcr)$), \p2 = ($(vr1) - (rbeadcr)$),
      \n1 = {atan2(\y1,\x1)}, \n2 = {360+atan2(\y2,\x2)}, \n3 = {veclen(\x1,\y1)} in
      (urm) arc (\n1:\n2:\n3) -- (l11);
      \path [edge] let \p1 = ($(urm) - (rbeadcl)$), \p2 = ($(vr2) - (rbeadcl)$),
      \n1 = {atan2(\y1,\x1)}, \n2 = {atan2(\y2,\x2)}, \n3 = {veclen(\x1,\y1)} in
      (urm) arc (\n1:\n2:\n3) -- (l12);
      \node [node] at (b) {};
      \node [node] at (c1) {};
      \node [node] at (c2) {};
      \node [node] at (c3) {};
      \node [node] at (d1) {};
      \node [node] at (d2) {};
      \node [node] at (urm) {};
      \node [anchor=east] at (l0l) {$w_L$};
      \node [anchor=west] at (l1r) {$w_R$};
      \node [anchor=west,yshift=2pt] at (urr) {$u_R$};
      \node [anchor=west,yshift=-2pt] at (vrr) {$v_R$};
      \node [anchor=west] at (br) {$b$};
      \node [anchor=west,yshift=4pt] at (d1t) {$d_1$};
      \node [anchor=west,yshift=4pt] at (d2t) {$d_2$};
      \node [anchor=east,yshift=4pt] at (ull) {$u_L$};
      \node [anchor=east,yshift=-2pt] at (vll) {$v_L$};
      \node [anchor=east,yshift=4pt] at (c1t) {$c_1$};
      \node [anchor=east,yshift=4pt] at (c2t) {$c_2$};
      \node [anchor=east,yshift=4pt] at (c3t) {$c_3$};
    \end{tikzpicture}%
    \hspace*{\stretch{1}}%
    \begin{tikzpicture}[
      node/.style={fill=black,circle,inner sep=0pt,minimum size=3pt,outer sep=0pt},
      dup/.style={fill=red,circle,inner sep=0pt,minimum size=6pt},
      wide node/.style={fill=black,rectangle,inner sep=0pt,minimum height=3pt,rounded corners=1.5pt},
      edge/.style={draw,thick},
      x=0.7cm,y=0.7cm]
      \begin{scope}[overlay]
        \coordinate (l0) at (0,0);
        \coordinate (l1) at (1,0);
        \coordinate (w) at (60:1);
        \path (w) foreach \i/\v in {1.25/v,2.25/u,4.25/d2,6.5/d1,7.5/c2,9.75/c1,10.75/q,11.75/b,12.75/r} {
          +(90:\i) coordinate (\v)
        };
        \path (c1) +(225:1) coordinate (cc1);
        \path (c2) +(225:1) coordinate (cc2);
        \path (d1) +(315:1) coordinate (cd1);
        \path (d2) +(315:1) coordinate (cd2);
        \path [name path=path1] (v) -- +(135:1);
        \path [name path=path2] (u) -- +(225:1);
        \path [name intersections={of=path1 and path2}] coordinate (bbeadcl) at (intersection-1);
        \path [name path=path1] (v) -- +(45:1);
        \path [name path=path2] (u) -- +(-45:1);
        \path [name intersections={of=path1 and path2}] coordinate (bbeadcr) at (intersection-1);
        \path [name path=path1] (q) -- +(135:1);
        \path [name path=path2] (b) -- +(225:1);
        \path [name intersections={of=path1 and path2}] coordinate (tbeadcl) at (intersection-1);
        \path [name path=path1] (q) -- +(45:1);
        \path [name path=path2] (b) -- +(-45:1);
        \path [name intersections={of=path1 and path2}] coordinate (tbeadcr) at (intersection-1);
        \path [name path=path1] (w) -- (r);
        \path [name path=path2] let \p1 = ($(u) - (bbeadcl)$), \n1 = {0.21cm+veclen(\x1,\y1)} in
        (bbeadcl) circle (\n1);
        \path [name intersections={of=path1 and path2}] coordinate (um) at (intersection-1)
        coordinate (vm) at (intersection-2);
        \path [name path=path2] let \p1 = ($(b) - (tbeadcl)$), \n1 = {0.21cm+veclen(\x1,\y1)} in
        (tbeadcl) circle (\n1);
        \path [name intersections={of=path1 and path2}] coordinate (bm) at (intersection-1)
        coordinate (qm) at (intersection-2);
      \end{scope}
      \path [edge] (r) -- (bm) (qm) -- (c1) -- (cc1) (cc2) -- (c2) -- (d1) -- (cd1)
      (cd2) -- (d2) -- (um) (vm) -- (w) (l0) -- (w) -- (l1);
      \path [edge] let \p1 = ($(bm) - (tbeadcr)$), \p2 = ($(qm) - (tbeadcr)$),
      \p4 = ($(qm) - (tbeadcl)$), \p5 = ($(bm) - (tbeadcl)$),
      \n1 = {atan2(\y1,\x1)}, \n2 = {360+atan2(\y2,\x2)}, \n3 = {veclen(\x1,\y1)},
      \n4 = {atan2(\y4,\x4)}, \n5 = {atan2(\y5,\x5)}, \n6 = {veclen(\x4,\y4)} in
      (bm) arc (\n1:\n2:\n3) arc (\n4:\n5:\n6);
      \path [edge] let \p1 = ($(um) - (bbeadcr)$), \p2 = ($(vm) - (bbeadcr)$),
      \p4 = ($(vm) - (bbeadcl)$), \p5 = ($(um) - (bbeadcl)$),
      \n1 = {atan2(\y1,\x1)}, \n2 = {360+atan2(\y2,\x2)}, \n3 = {veclen(\x1,\y1)},
      \n4 = {atan2(\y4,\x4)}, \n5 = {atan2(\y5,\x5)}, \n6 = {veclen(\x4,\y4)} in
      (um) arc (\n1:\n2:\n3) arc (\n4:\n5:\n6);
      \path [edge,dashed] (c1) -- (c2) (d1) -- (d2);
      \node [node] at (c1) {};
      \node [node] at (c2) {};
      \node [node] at (d1) {};
      \node [node] at (d2) {};
      \node [node] at (bm) {};
      \node [node] at (qm) {};
      \node [node] at (um) {};
      \node [node] at (vm) {};
      \node [node] at (w)  {};
      \node [node] at (l0) {};
      \node [node] at (l1) {};
      \node [anchor=east] at (l0) {$w_L$};
      \node [anchor=west] at (l1) {$w_R$};
      \node [anchor=west,yshift=2pt] at (w) {$w$};
      \node [anchor=west,yshift=-3pt] at (vm) {$v$};
      \node [anchor=west,yshift=3pt] at (um) {$u$};
      \node [anchor=west,yshift=2pt] at (d2) {$d_r$};
      \node [anchor=west,yshift=2pt] at (d1) {$d_1$};
      \node [anchor=east,yshift=2pt] at (c2) {$c_l$};
      \node [anchor=east,yshift=2pt] at (c1) {$c_1$};
      \node [anchor=west,yshift=4pt] at (bm) {$b$};
      \node [anchor=west,yshift=-3pt] at (qm) {$q$};
    \end{tikzpicture}%
    \hspace*{\stretch{1}}%
    \begin{tikzpicture}[
      node/.style={fill=black,circle,inner sep=0pt,minimum size=3pt,outer sep=0pt},
      dup/.style={fill=red,circle,inner sep=0pt,minimum size=6pt},
      wide node/.style={fill=black,rectangle,inner sep=0pt,minimum height=3pt,rounded corners=1.5pt},
      edge/.style={draw,thick},
      x=0.7cm,y=0.7cm]
      \begin{scope}[overlay]
        \coordinate (l0) at (0,0);
        \coordinate (l1) at (1,0);
        \coordinate (w) at (60:1);
        \path (w) foreach \i/\v in {1.25/v,2.25/u,4.25/d2,5.75/d1,6.75/c3,8.25/c2,9.75/c1,10.75/q,11.75/b,12.75/r} {
          +(90:\i) coordinate (\v)
        };
        \foreach \v in {l0,l1,w,v,u,d2,d1,c3,c2,c1,q,b,r} {
          \path (\v) +(180:0.2) coordinate (\v 2) +(180:0.4) coordinate (\v 1) +(180:0.6) coordinate (\v l)
          +(0:0.2) coordinate (\v 3) +(0:0.4) coordinate (\v 4) +(0:0.6) coordinate (\v r);
        }
        \coordinate (l0m) at (barycentric cs:l01=0.5,l02=0.5);
        \coordinate (wm) at (barycentric cs:w1=0.5,w2=0.5);
        \coordinate (c1m) at (barycentric cs:c11=0.5,c12=0.5);
        \coordinate (c2m) at (barycentric cs:c21=0.5,c22=0.5);
        \coordinate (c3m) at (barycentric cs:c31=0.5,c32=0.5);
        \coordinate (d1m) at (barycentric cs:d13=0.5,d14=0.5);
        \coordinate (d2m) at (barycentric cs:d23=0.5,d24=0.5);
        \path (c1m) +(225:1) coordinate (cc1);
        \path (c2m) +(225:1) coordinate (cc2);
        \path (c3m) +(225:1) coordinate (cc3);
        \path (d1m) +(315:1) coordinate (cd1);
        \path (d2m) +(315:1) coordinate (cd2);
        \foreach \v in {cc1,cc2,cc3,cd1,cd2} {
          \path (\v) +(90:0.3) coordinate (\v t) +(270:0.3) coordinate (\v b);
        }
        \path [name path=path1] (l0) -- +(60:1.2);
        \path [name path=path2] (l1) -- +(120:1.2);
        \path [name intersections={of=path1 and path2}] coordinate (wb) at (intersection-1);
        \path [name path=path1] (v) -- +(135:1);
        \path [name path=path2] (u) -- +(225:1);
        \path [name intersections={of=path1 and path2}] coordinate (bbeadcl) at (intersection-1);
        \path [name path=path1] (v) -- +(45:1);
        \path [name path=path2] (u) -- +(-45:1);
        \path [name intersections={of=path1 and path2}] coordinate (bbeadcr) at (intersection-1);
        \path [name path=path1] (q) -- +(135:1);
        \path [name path=path2] (b) -- +(225:1);
        \path [name intersections={of=path1 and path2}] coordinate (tbeadcl) at (intersection-1);
        \path [name path=path1] (q) -- +(45:1);
        \path [name path=path2] (b) -- +(-45:1);
        \path [name intersections={of=path1 and path2}] coordinate (tbeadcr) at (intersection-1);
        \path [name path=path1] (wm) -- +(90:13);
        \path [name path=path2] let \p1 = ($(u) - (bbeadcr)$), \n1 = {0.28cm+veclen(\x1,\y1)}
        in (bbeadcr) circle (\n1);
        \path [name intersections={of=path1 and path2}] coordinate (v1) at (intersection-2)
        coordinate (u1) at (intersection-1);
        \path [name path=path2] let \p1 = ($(b) - (tbeadcr)$), \n1 = {0.21cm+veclen(\x1,\y1)}
        in (tbeadcr) circle (\n1);
        \path [name intersections={of=path1 and path2}] coordinate (q1) at (intersection-2);
        \path [name path=path1] (b) -- +(90:1);
        \path [name intersections={of=path1 and path2}] coordinate (bm) at (intersection-1);
        \path [name path=path1] (d2m) -- +(270:2);
        \path [name path=path2] let \p1 = ($(u) - (bbeadcl)$), \n1 = {0.14cm+veclen(\x1,\y1)}
        in (bbeadcr) circle (\n1);
        \path [name intersections={of=path1 and path2}] coordinate (u3) at (intersection-1);
        \path [name path=path1] (w3) -- +(90:1);
        \path [name intersections={of=path1 and path2}] coordinate (v3) at (intersection-1);
        \path [name path=path2] let \p1 = ($(u) - (bbeadcl)$), \n1 = {0.21cm+veclen(\x1,\y1)}
        in (bbeadcl) circle (\n1);
        \path [name path=path1] (w4) -- +(90:1.5);
        \path [name intersections={of=path1 and path2}] coordinate (v4) at (intersection-1);
        \path [name path=path1] let \p1 = ($(u3) - (bbeadcl)$), \p2 = ($(v4) - (bbeadcl)$),
        \n1 = {veclen(\x2,\y2)}, \n2 = {veclen(\x1,\y1)}, \n3 = {atan2(\y1,\x1) - acos(\n1 / \n2)} in
        (bbeadcl) -- +(\n3:1.5);
        \path [name intersections={of=path1 and path2}] coordinate (u4t) at (intersection-1);
        \path [name path=path1] (d1m) -- +(90:6);
        \path [name path=path2] let \p1 = ($(b) - (tbeadcl)$), \n1 = {0.21cm+veclen(\x1,\y1)}
        in (tbeadcl) circle (\n1);
        \path [name intersections={of=path1 and path2}] coordinate (q3) at (intersection-1);
        \path [name path=path1] (wl) -- +(90:13);
        \path [name path=path2] let \p1 = ($(v) - (bbeadcr)$), \n1 = {0.42cm+veclen(\x1,\y1)} in
        (bbeadcr) circle (\n1);
        \path [name intersections={of=path1 and path2}] coordinate (vl) at (intersection-2)
        coordinate (ul) at (intersection-1);
        \path [name path=path2] (cc3b) -- +(45:1);
        \path [name intersections={of=path1 and path2}] coordinate (c3b) at (intersection-1);
        \path [name path=path2] (cc3t) -- +(45:1);
        \path [name intersections={of=path1 and path2}] coordinate (c3t) at (intersection-1);
        \path [name path=path2] (cc2b) -- +(45:1);
        \path [name intersections={of=path1 and path2}] coordinate (c2b) at (intersection-1);
        \path [name path=path2] (cc2t) -- +(45:1);
        \path [name intersections={of=path1 and path2}] coordinate (c2t) at (intersection-1);
        \path [name path=path2] (cc1b) -- +(45:1);
        \path [name intersections={of=path1 and path2}] coordinate (c1b) at (intersection-1);
        \path [name path=path2] (cc1t) -- +(45:1);
        \path [name intersections={of=path1 and path2}] coordinate (c1t) at (intersection-1);
        \path [name path=path2] let \p1 = ($(q) - (tbeadcr)$), \n1 = {0.42cm+veclen(\x1,\y1)} in
        (tbeadcr) circle (\n1);
        \path [name intersections={of=path1 and path2}] coordinate (ql) at (intersection-2)
        coordinate (bl) at (intersection-1);
        \path [name path=path1] (wr) -- +(90:13);
        \path [name path=path2] let \p1 = ($(v) - (bbeadcl)$), \n1 = {0.42cm+veclen(\x1,\y1)} in
        (bbeadcl) circle (\n1);
        \path [name intersections={of=path1 and path2}] coordinate (vr) at (intersection-2)
        coordinate (ur) at (intersection-1);
        \path [name path=path2] (cd2b) -- +(135:1);
        \path [name intersections={of=path1 and path2}] coordinate (d2b) at (intersection-1);
        \path [name path=path2] (cd2t) -- +(135:1);
        \path [name intersections={of=path1 and path2}] coordinate (d2t) at (intersection-1);
        \path [name path=path2] (cd1b) -- +(135:1);
        \path [name intersections={of=path1 and path2}] coordinate (d1b) at (intersection-1);
        \path [name path=path2] (cd1t) -- +(135:1);
        \path [name intersections={of=path1 and path2}] coordinate (d1t) at (intersection-1);
        \path [name path=path2] let \p1 = ($(q) - (tbeadcl)$), \n1 = {0.42cm+veclen(\x1,\y1)} in
        (tbeadcl) circle (\n1);
        \path [name intersections={of=path1 and path2}] coordinate (qr) at (intersection-2)
        coordinate (br) at (intersection-1);
        \path [name path=path1] let \p1 = ($(u) - (bbeadcl)$), \n1 = {0.42cm+veclen(\x1,\y1)} in
        (bbeadcl) circle (\n1);
        \path [name path=path2] (wb) -- (q);
        \path [name intersections={of=path1 and path2}] coordinate (um) at (intersection-1)
        coordinate (vm) at (intersection-2);
      \end{scope}
      \path [fill=black!20] (rl) -- (bl) -- (b) -- (br) -- (rr) -- cycle;
      \path [fill=blue!30]
      let \p1 = ($(ql) - (tbeadcr)$), \p2 = ($(bl) - (tbeadcr)$),
      \n1 = {360+atan2(\y1,\x1)}, \n2 = {atan2(\y2,\x2)}, \n3 = {veclen(\x1,\y1)} in
      (b) -- (um) -- (ul) -- (c3b) -- (cc3b) -- (cc3t) -- (c3t) -- (c2b) -- (cc2b)
      -- (cc2t) -- (c2t) -- (c1b) -- (cc1b) -- (cc1t) -- (c1t) -- (ql) arc (\n1:\n2:\n3)
      -- cycle
      (vl) -- (vm) -- (wb) -- (l0) -- (l0l) -- (wl) -- cycle;
      \path [pattern=north west lines,pattern color=white]
      let \p1 = ($(ql) - (tbeadcr)$), \p2 = ($(bl) - (tbeadcr)$),
      \n1 = {360+atan2(\y1,\x1)}, \n2 = {atan2(\y2,\x2)}, \n3 = {veclen(\x1,\y1)} in
      (b) -- (um) -- (ul) -- (c3b) -- (cc3b) -- (cc3t) -- (c3t) -- (c2b) -- (cc2b)
      -- (cc2t) -- (c2t) -- (c1b) -- (cc1b) -- (cc1t) -- (c1t) -- (ql) arc (\n1:\n2:\n3)
      -- cycle
      (vl) -- (vm) -- (wb) -- (l0) -- (l0l) -- (wl) -- cycle;
      \path [fill=red!30]
      let \p1 = ($(qr) - (tbeadcl)$), \p2 = ($(br) - (tbeadcl)$),
      \n1 = {atan2(\y1,\x1)}, \n2 = {atan2(\y2,\x2)}, \n3 = {veclen(\x1,\y1)} in
      (b) -- (um) -- (ur) -- (d2b) -- (cd2b) -- (cd2t) -- (d2t) -- (d1b) -- (cd1b) -- (cd1t)
      -- (d1t) -- (qr) arc (\n1:\n2:\n3) -- cycle
      (vm) -- (wb) -- (l1) -- (l1r) -- (wr) -- (vr) -- cycle;
      \path [pattern=north east lines,pattern color=white]
      let \p1 = ($(qr) - (tbeadcl)$), \p2 = ($(br) - (tbeadcl)$),
      \n1 = {atan2(\y1,\x1)}, \n2 = {atan2(\y2,\x2)}, \n3 = {veclen(\x1,\y1)} in
      (b) -- (um) -- (ur) -- (d2b) -- (cd2b) -- (cd2t) -- (d2t) -- (d1b) -- (cd1b) -- (cd1t)
      -- (d1t) -- (qr) arc (\n1:\n2:\n3) -- cycle
      (vm) -- (wb) -- (l1) -- (l1r) -- (wr) -- (vr) -- cycle;
      \path [fill=violet!30]
      let \p1 = ($(um) - (bbeadcr)$), \p2 = ($(vm) - (bbeadcr)$),
      \n1 = {atan2(\y1,\x1)}, \n2 = {360+atan2(\y2,\x2)}, \n3 = {veclen(\x1,\y1)},
      \p4 = ($(vm) - (bbeadcl)$), \p5 = ($(um) - (bbeadcl)$),
      \n4 = {atan2(\y4,\x4)}, \n5 = {atan2(\y5,\x5)}, \n6 = {veclen(\x4,\y4)} in
      (um) arc(\n1:\n2:\n3) arc (\n4:\n5:\n6) -- cycle;
      \path [pattern=crosshatch,pattern color=white]
      let \p1 = ($(um) - (bbeadcr)$), \p2 = ($(vm) - (bbeadcr)$),
      \n1 = {atan2(\y1,\x1)}, \n2 = {360+atan2(\y2,\x2)}, \n3 = {veclen(\x1,\y1)},
      \p4 = ($(vm) - (bbeadcl)$), \p5 = ($(um) - (bbeadcl)$),
      \n4 = {atan2(\y4,\x4)}, \n5 = {atan2(\y5,\x5)}, \n6 = {veclen(\x4,\y4)} in
      (um) arc(\n1:\n2:\n3) arc (\n4:\n5:\n6) -- cycle;
      \path [draw,fill=white] let \p1 = ($(u) - (bbeadcl)$), \n1 = {veclen(\x1,\y1)}
      in (u) arc (135:225:\n1) arc (-45:45:\n1) -- cycle
      (b) arc (135:225:\n1) arc (-45:45:\n1) -- cycle;
      \path [draw] (l0) -- (wb) -- (l1);
      \path [draw] let \p1 = ($(ul) - (bbeadcr)$), \p2 = ($(vl) - (bbeadcr)$),
      \n1 = {atan2(\y1,\x1)}, \n2 = {360+atan2(\y2,\x2)}, \n3 = {veclen(\x1,\y1)} in
      (cc3b) -- (c3b) -- (ul) arc (\n1:\n2:\n3) -- (wl) -- (l0l);
      \path [draw] (cc2b) -- (c2b) -- (c3t) -- (cc3t);
      \path [draw] (cc1b) -- (c1b) -- (c2t) -- (cc2t);
      \path [draw] let \p1 = ($(bl) - (tbeadcr)$), \p2 = ($(ql) - (tbeadcr)$),
      \n1 = {atan2(\y1,\x1)}, \n2 = {360+atan2(\y2,\x2)}, \n3 = {veclen(\x1,\y1)} in
      (rl) -- (bl) arc (\n1:\n2:\n3) -- (c1t) -- (cc1t);
      \path [draw] let \p1 = ($(ur) - (bbeadcl)$), \p2 = ($(vr) - (bbeadcl)$),
      \n1 = {atan2(\y1,\x1)}, \n2 = {atan2(\y2,\x2)}, \n3 = {veclen(\x1,\y1)} in
      (cd2b) -- (d2b) -- (ur) arc (\n1:\n2:\n3) -- (wr) -- (l1r);
      \path [draw] (cd1b) -- (d1b) -- (d2t) -- (cd2t);
      \path [draw] let \p1 = ($(br) - (tbeadcl)$), \p2 = ($(qr) - (tbeadcl)$),
      \n1 = {atan2(\y1,\x1)}, \n2 = {atan2(\y2,\x2)}, \n3 = {veclen(\x1,\y1)} in
      (rr) -- (br) arc (\n1:\n2:\n3) -- (d1t) -- (cd1t);

      \path [edge] (r) -- (bm);
      \path [edge] let \p1 = ($(u1) - (bbeadcr)$), \p2 = ($(v1) - (bbeadcr)$),
      \n1 = {atan2(\y1,\x1)}, \n2 = {360+atan2(\y2,\x2)}, \n5 = {veclen(\x1,\y1)},
      \p3 = ($(bm) - (tbeadcr)$), \p4 = ($(q1) - (tbeadcr)$),
      \n3 = {atan2(\y3,\x3)}, \n4 = {360+atan2(\y4,\x4)}, \n6 = {veclen(\x3,\y3)} in
      (bm) arc (\n3:\n4:\n6) -- (c1m) -- (c2m) -- (c3m) -- (u1) arc (\n1:\n2:\n5) -- (wm) -- (l0m)
      (c1m) -- (cc1) (c2m) -- (cc2) (c3m) -- (cc3);
      \path [edge] let \p1 = ($(u3) - (bbeadcr)$), \p2 = ($(v3) - (bbeadcr)$),
      \n1 = {atan2(\y1,\x1)}, \n2 = {360+atan2(\y2,\x2)}, \n3 = {veclen(\x1,\y1)} in
      (cd2) -- (d2m) -- (u3) arc (\n1:\n2:\n3) -- (w3) -- (l13)
      (cd1) -- (d1m) -- (d2m);
      \path [edge] let \p1 = ($(u4t) - (bbeadcl)$), \p2 = ($(v4) - (bbeadcl)$),
      \n1 = {atan2(\y1,\x1)}, \n2 = {atan2(\y2,\x2)}, \n3 = {veclen(\x1,\y1)} in
      (u3) -- (u4t) arc (\n1:\n2:\n3) -- (w4) -- (l14);
      \path [edge] let \p1 = ($(bm) - (tbeadcl)$), \p2 = ($(q3) - (tbeadcl)$),
      \n1 = {atan2(\y1,\x1)}, \n2 = {atan2(\y2,\x2)}, \n3 = {veclen(\x1,\y1)} in
      (bm) arc (\n1:\n2:\n3) -- (d1m);
      \node [node] at (c1m) {};
      \node [node] at (c2m) {};
      \node [node] at (c3m) {};
      \node [node] at (d1m) {};
      \node [node] at (d2m) {};
      \node [node] at (bm)  {};
      \node [node] at (u3)  {};
      \node [anchor=east] at (l0l) {$w_L$};
      \node [anchor=west] at (l1r) {$w_R$};
      \node [anchor=west,yshift=2pt] at (wr) {$w$};
      \node [anchor=west,yshift=-2pt] at (vr) {$v$};
      \node [anchor=west,yshift=2pt] at (ur) {$u$};
      \node [anchor=west,yshift=2pt] at (d2t) {$d_2$};
      \node [anchor=west,yshift=2pt] at (d1t) {$d_1$};
      \node [anchor=east,yshift=2pt] at (c3t) {$c_3$};
      \node [anchor=east,yshift=2pt] at (c2t) {$c_2$};
      \node [anchor=east,yshift=2pt] at (c1t) {$c_1$};
      \node [anchor=west,yshift=2pt] at (br) {$b$};
      \node [anchor=west,yshift=-2pt] at (qr) {$q$};
    \end{tikzpicture}%
    \caption{Given a beaded tree $B$ that weakly displays a set of MUL-trees
      $\mathcal{T}$ and in which not all reticulations are on one path, we can
      produce a beaded tree with the same number of reticulations, more
      reticulations on the same path, and which also weakly displays the
      MUL-trees in~$\mathcal{T}$.}
    \label{fig:OneLineageLocal}
  \end{figure}

  To show that any MUL-tree weakly displayed by $B$ is also weakly displayed by
  $B'$, let $T$ be a MUL-tree weakly displayed by $B$ and let $h$ be a weak
  embedding of $T$ into $B$.  We define a weak embedding $h'$ of $T$ into $B'$
  as follows: For any node $x \in V(T)$, let
  \begin{equation*}
    h'(x) = \begin{cases}
      u    & \text{if } h(x) \in \{u_L, u_R\}\\
      h(x) & \text{otherwise}
    \end{cases}.
  \end{equation*}
  Note that this ensures that $h'(x) \in V' \cup \{u\}$ because $h(x)$ is a
  tree node for all $x \in V(T)$, that is, $h(x) \notin \{v_L, v_R\}$.  This
  definition of $h'$ ensures that there exists a path $h'(xy)$ from $h'(x)$ to
  $h'(y)$ for every edge $xy$ of $T$.  Indeed, there exists a path $h(xy)$ from
  $h(x)$ to $h(y)$ in $B$ because $h$ is a weak embedding of $T$ into $B$.  If
  $h(x), h(y) \in V'$, then $h'(x) = h(x)$, $h'(y) = h(y)$, and we observed
  above that every descendant of $h'(x)$ in $B$ that belongs to $V'$ is also a
  descendant of $h'(x)$ in $B'$, that is, there exists a path $h'(xy)$ from
  $h'(x)$ to $h'(y)$.  If $h(x) \in \{u_L, u_R\}$, then $h(y) \in V'$, $h'(x) =
  u$, and $h'(y) = h(y)$.  As observed above, every descendant of $u_L$ or
  $u_R$ in $B$ that belongs to $V'$ is a descendant of $u$ in $B'$.  Thus,
  there exists a path $h'(xy)$ from $h'(x)$ to $h'(y)$ also in this case.
  Finally, if $h(y) \in \{u_L, u_R\}$, then $h(x) \in V'$, $h'(x) = h(x)$, and
  $h'(y) = u$.  As observed above, every ancestor of $u_L$ or $u_R$ in $B$ that
  belongs to $V'$ is an ancestor of $u$ in $B'$.  Thus, there exists a path
  $h'(xy)$ from $h'(x)$ to $h'(y)$ once again.  It remains to show that these
  paths can be chosen so that the two paths $h'(xy)$ and $h'(xy')$ corresponding
  to the edges $xy$ and $xy'$ between a node $x \in V(T)$ and its two children
  $y$ and $y'$ in $T$ begin with different out-edges of~$h'(x)$.

  So consider any node $x$ of $T$ and its two children $y$ and $y'$.  If
  $h'(x)$ is the top node of a bead, then the two paths $h'(xy)$ and $h'(xy')$
  can be chosen to start with different edges of this bead.  If $h'(x)$ is not
  the top node of a bead, then $h'(x) = h(x) \in V'$ and $h(x)$ is not the top
  node of a bead in $B$ either.  Since $h$ is a weak embedding of $T$ into $B$,
  $h(y)$ is a descendant of one child $z$ of $h(x)$ and $h(y')$ is a descendant
  of the other child $z'$ of $h(x)$.  Moreover, one of these two children, say
  $z$, is also a child of $h'(x)$ in $B'$.  As observed above, since $h(y)$ is
  a descendant of $z$ in $B$, $h'(y)$ is also a descendant of $z$ in $B'$, so
  we can choose the path $h'(xy)$ to start with the edge $h'(x)z$.  If $z'$ is
  a child of $h'(x)$ in $B'$, then, by an analogous argument, we can choose the
  path $h'(xy')$ to start with the edge $h'(x)z'$, so the two paths $h'(xy)$
  and $h'(xy')$ start with different out-edges of $h'(x)$.  If $z'$ is not a
  child of $h'(x)$, then $h(x) \in \{c_l, d_r\}$, $z' \in \{u_L, u_R\}$, and
  $z$ is the child of $h(x)$ not on the path from $h(x)$ to $u_L$ or $u_R$.
  In this case, $u$~is a descendant of $h'(x)$ and $h'(y')$ is a descendant of
  $u$.  Thus, we can choose $h'(xy')$ to be the concatenation of two paths from
  $h'(x)$ to $u$ and from $u$ to $h'(y')$.  Since $z$ does not belong to this
  path, the two paths $h'(xy)$ and $h'(xy')$ once again start with different
  edges.

  Since we have just shown that any MUL-tree weakly displayed by $B$ is also
  weakly displayed by $B'$, $B'$~is a solution for $\mathcal{T}$.  Moreover,
  $B'$ has the same number of beads as $B$ and, since $\lambda_{B'}(v) =
  \lambda_B(v_L) + \lambda_B(v_R)$, $\lambda_{B'}(q) = \lambda_B(v_L) +
  \lambda_B(v_R) + 1$, and $\lambda_{B'}(z) = \lambda_B(z)$ for any
  reticulation node $z \in V'$, $\lambda(B') > \lambda(B)$.  This contradicts
  the choice of $B$, so $B$ has all its beads on a single path.
\end{proof}

In what follows, we use $\mathcal{T}|_S$ and $\mathcal{T} \setminus S$ to
denote the sets $\{T_1|_{S}, \ldots, T_t|_{S}\}$ and $\{T_1 \setminus S,
  \ldots, T_t \setminus S\}$, respectively, for any set of trees $\mathcal{T} =
\{T_1, \ldots, T_t\}$ and any label set $S$.  If any tree in $\mathcal{T}|_S$
or $\mathcal{T} \setminus S$ is empty, it is removed from the set.

The following definitions and lemmas describe the structure of an optimal
solution for $\mathcal{T}$ in terms of optimal solutions for $\mathcal{T}|_{S}$
and $\mathcal{T}\setminus S$. These structural results will make it easy to
design an algorithm for \textsc{Beaded Tree}.

\begin{definition}\label{def:splitPartition}
  Given a set of MUL-trees $\mathcal{T} = \{T_1, \ldots, T_t\}$, with each
  MUL-tree $T_i$ having label set $X_i \subseteq X$, the \emph{split partition}
  $\mathcal{S} = \{S_1, \ldots, S_s\}$ of $\{T_1, \ldots, T_t\}$ is the partition
  of $X$ into minimal sets such that any two labels of the same MUL-tree in
  $F_1(\mathcal{T})$ belong to the same set in $\mathcal{S}$.
\end{definition}

\begin{definition}
  Given two phylogenetic networks $N_1$ on $X_1$ and $N_2$ on $X_2$ with $X_1
  \cap X_2 = \emptyset$, the process of \emph{joining} $N_1$ with $N_2$
  consists of identifying the root $r_1$ of $N_1$ and the root $r_2$ of $N_2$
  into a single node $u$ and making $u$ the child of a new root node $r$.
\end{definition}

\begin{observation}
  If $N$ is obtained by joining $N_1$ and $N_2$, then any MUL-tree weakly
  displayed by $N_1$ or $N_2$ is also weakly displayed by $N$.
\end{observation}

The following lemma immediately suggests a strategy for constructing an
optimal beaded tree for a collection $\mathcal{T}$ of MUL-trees.

\begin{lemma}\label{lem:structureFromSpitPartition}
  Given an instance $\mathcal{T}$ of \textsc{Beaded Tree}, if $|X|= 1$ and
  $\max_{1 \le i \le t}|L(T_i)|=1$, then the optimal solution is the tree with
  a single leaf on $X$. Otherwise, let $\mathcal{S} = \{S_1, \ldots, S_s\}$ be
  the split partition of $\mathcal{T}$.  If for some $S_i$, there exists a tree
  $T$ weakly displaying the MUL-trees in $\mathcal{T}|_{S_i}$, then there
  exists an optimal solution $B$ that is obtained by joining $T$ with an
  optimal solution to $\mathcal{T}\setminus S_i$.  If no such tree $T$ exists, there exists an
  optimal solution $B$ with a bead $(u,v)$ at the root and such that $v$ is the
  root of an optimal solution for $F_1(\mathcal{T})$.
\end{lemma}

\begin{proof}
  If $|X|= 1$ and $\max_{1 \le i \le t}|L(T_i)|=1$, then the optimal solution
  clearly is the tree with a single leaf on $X$. So suppose that $|X| > 1$ and
  assume first that there exists a set $S_i \in \mathcal{S}$ such that the
  MUL-trees in $\mathcal{T}|_{S_i}$ are weakly displayed by some tree~$T$.  If
  some tree $T'$ weakly displays the MUL-trees in $\mathcal{T}$, then $s \ge 2$
  (since any tree in $F_1(\mathcal{T})$ has its leaf set contained within the
  leaf set of one of the trees in $F_1(T')$ and we can assume w.l.o.g.\ that
  not all trees in $F_1(\mathcal{T})$ are displayed by the same tree in
  $F_1(T')$).  In particular, $S_i \neq X$.  If no such tree $T'$ exists,
  then $S_i \neq X$ because $T$ weakly displays all MUL-trees in
  $\mathcal{T}|_{S_i}$.  Since $S_i \neq X$ in both cases, it follows that
  $X\setminus S_i \ne \emptyset$. Now consider any optimal solution $B'$ for
  $\mathcal{T}$.  Observe that $B' \setminus S_i$ weakly displays all MUL-trees
  in $\mathcal{T} \setminus S_i$.  Moreover, $B' \setminus S_i$ has
  reticulation number at most that of $B'$.

  Construct a network $B$ by joining $B' \setminus S_i$ with $T$.  Any MUL-tree
  $T_j \in \mathcal{T}$ with no leaves in $S_i$ is weakly displayed by $B'
  \setminus S_i$ and therefore by~$B$.  Similarly, if every leaf of $T_j$ is in
  $S_i$, then $T_j$ is weakly displayed by $T$ and therefore by $B$.  So
  suppose $T_j$ has leaves in both $S_i$ and $X \setminus S_i$.  Since
  $F_1(T_j)$ consists of two MUL-trees and $\mathcal{S}$ is a split partition
  of $\mathcal{T}$, we must have $F_1(T_j) = \{T_j|_{S_i}, T_j \setminus S_i\}$.
  Since $T$ weakly displays $T_j|_{S_i}$ and $B' \setminus S_i$ weakly displays
  $T_j \setminus S_i$, it follows that $B$ weakly displays $T_j$.  This shows
  that $B$ displays all MUL-trees in $\mathcal{T}$.  Since $B$ has reticulation
  number at most that of $B'$, $B$ is therefore an optimal solution for
  $\mathcal{T}$.

  It remains to observe that $B' \setminus S_i$ is an optimal solution to
  $\mathcal{T} \setminus S_i$, as otherwise we could obtain a solution
  for $\mathcal{T}$ that is better than $B$ by joining $T$ with an optimal solution for
  $\mathcal{T} \setminus S_i$.  Thus, the lemma holds for the case when there
  exists a tree $T$ weakly displaying all MUL-trees in $\mathcal{T}|_{S_i}$ for
  some $S_i \in \mathcal{S}$.

  Now suppose that there is no tree weakly displaying the MUL-trees in
  $\mathcal{T}|_{S_i}$ for any $S_i \in \mathcal{S}$.  By
  Lemma~\ref{lem:allOneLineage}, there exists an optimal solution $B$ with all
  reticulations on one path. Suppose that $B$ does not have a bead at the root.
  Then the child $a$ of the root is a tree node which is the root of two
  otherwise disjoint beaded trees, and at least one of these beaded trees is a
  tree $T$ (without beads).  Let $S$ be the leaves of this tree $T$.  Since we
  can assume that at least one MUL-tree in $\mathcal{T}$ has a leaf in $S$,
  there exists a set $S_i \in \mathcal{S}$ such that $S_i \cap S \ne
  \emptyset$.  Any such set $S_i$ must be a subset of $S$ because otherwise
  there exists a MUL-tree $T' \in F_1(\mathcal{T})$ that has leaves in both $S$
  and $X \setminus S$; since $a$ is a tree node that is not part of a bead,
  $T'$ would have to be weakly displayed by either $T$ or $B \setminus S$,
  which is impossible.

  So consider such a set $S_i \subseteq S$ in $\mathcal{S}$.  $B|_{S_i}$ weakly
  displays the MUL-trees in $\mathcal{T}|_{S_i}$ and is a tree because
  $B|_{S_i} = T|_{S_i}$.  Since we assumed that no tree displaying all
  MUL-trees in $\mathcal{T}|_{S_i}$ exists, $B$ must in fact have a bead at the
  root, as claimed.  By Lemma~\ref{lem:topBead}, we also have that the bottom
  part of the bead is the root of a solution $B'$ for $F_1(\mathcal{T})$ with
  reticulation number $k-1$, where $k$ is the reticulation number of $B$.
  Moreover, $B'$ must be an optimal solution for $F_1(\mathcal{T})$ because
  otherwise we could obtain a solution for $\mathcal{T}$ that is better than $B$
  by adding a bead at the root of an optimal solution for $F_1(\mathcal{T})$.
  This proves the lemma for the case when there is no tree $T$ displaying
  all MUL-trees in $\mathcal{T}|_{S_i}$ for any $S_i \in \mathcal{S}$.
\end{proof}

The next two lemmas show that not only does there exist an optimal solution to
\textsc{Beaded Tree} with all reticulations on one path, but in fact
\emph{any} optimal solution must be quite close to such a structure.

\begin{lemma}\label{lem:noDoublePairs}
  Given two beads in any optimal solution to an instance $\mathcal{T}$ of
  \textsc{Beaded Tree} such that neither bead is a descendant of the other, at
  least one of these beads has no beads strictly descended from it.
\end{lemma}

\begin{proof}
  The proof is similar to the proof of Lemma~\ref{lem:allOneLineage}.

  Consider an optimal solution $B$ and suppose for the sake of contradiction
  that the claim does not hold for $B$.  Then let $(p_L, q_L)$, $(p_R, q_R)$,
  $(u_L,v_L)$ and $(u_R,v_R)$ be four distinct beads such that $(p_L, q_L)$ is
  not an ancestor of $(p_R, q_R)$ and $(p_R, q_R)$ is not an ancestor of $(p_L,
  q_L)$ , but $(p_L, q_L)$ is an ancestor of $(u_L,v_L)$ and $(p_R,q_R)$ is an
  ancestor of $(u_R,v_R)$.  See Figure~\ref{fig:noDoublePairs}.  Moreover,
  assume that $(p_L, q_L)$, $(p_R, q_R)$, $(u_L,v_L)$ and $(u_R,v_R)$ are the
  earliest such beads, that is, the condition is not satisfied if we replace
  any one of these beads with one of its strict ancestors.  This implies that
  there are no beads on the path between $q_L$ and $u_L$, on the path between
  $q_R$ and $u_R$ or on the path from the least common ancestor of $p_L$ and
  $p_R$ to either $p_L$ or $p_R$.

  Let $x$ be the least common ancestor of $p_L$ and $p_R$.  If $p_L$ is not a
  child of $x$, then let $a_1, \ldots a_s$ be the nodes on the path from $x$ to
  $p_L$.  Similarly, if $p_R$ is not a child of $x$, then let $b_1, \ldots b_t$
  be the nodes on the path from $x$ to $p_R$.  If $u_L$ is not a child of
  $q_L$, then let $c_1, \ldots c_l$ be the nodes on the path from $q_L$ to
  $u_L$.  Similarly, if $u_R$ is not a child of~$q_R$, then let $d_1, \ldots
  d_r$ be the nodes on the path from $q_R$ to $u_R$.  Note that $a_1, \ldots
  a_s, b_1, \ldots b_t, c_1, \ldots c_l, d_1, \ldots d_r$ are all tree nodes.
  Finally let $w_L$ be the single child of $v_L$ and $w_R$ the single child of
  $v_R$.

  Construct $B'$ from $B$ as follows:  Delete the nodes $p_L$, $q_L$, $u_L$,
  $v_L$, $p_R$, $q_R$, $u_R$, and $v_R$ and any edges incident to them, as well
  as the edges $xa_1$ and $xb_1$.  Now add new nodes $y$, $p$, $q$, $u$, $v$,
  and $w$, and add a pair of parallel arcs from $x$ to $y$, from $p$ to $q$,
  and from $u$ to $v$, as well as arcs $ya_1$, $a_sb_1$, $b_tp$, $qc_1$,
  $c_ld_1$, $d_ru$, $vw$, $ww_L$, and $ww_R$ (see
  Figure~\ref{fig:noDoublePairs}.) (Note that this construction assumes that
  $s,t,l,r \geq 1$; if this is not the case, then we can produce $B'$ by
  introducing ``dummy nodes'' $a_1$, $b_1$, $c_1$, and $d_1$ and suppressing
  them after the construction is complete.)

  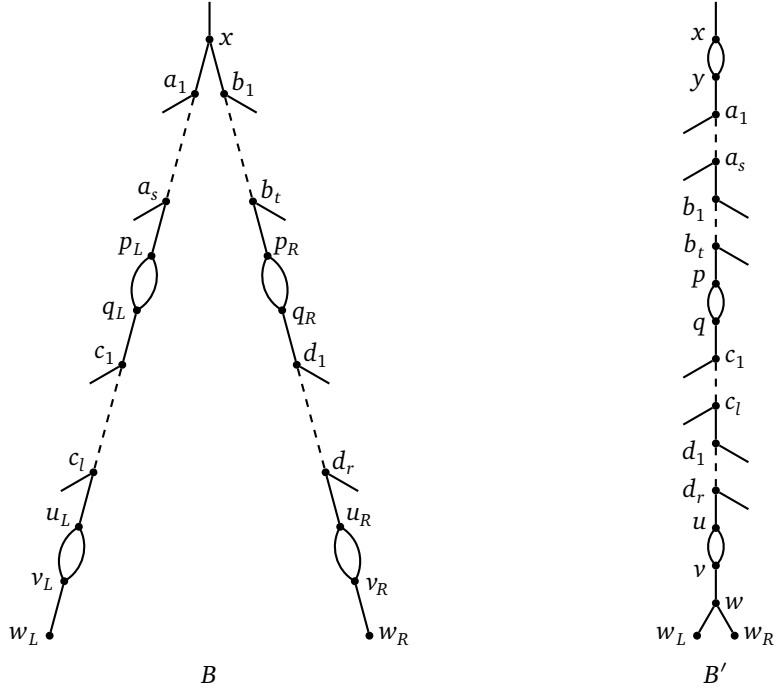
\begin{figure}[t]
    \hspace*{\stretch{1}}%
    \begin{tikzpicture}[
      node/.style={fill=black,circle,inner sep=0pt,minimum size=3pt,outer sep=0pt},
      dup/.style={fill=red,circle,inner sep=0pt,minimum size=6pt},
      wide node/.style={fill=black,rectangle,inner sep=0pt,minimum height=3pt,rounded corners=1.5pt},
      edge/.style={draw,thick},
      x=0.5cm,y=0.5cm]
      \begin{scope}[overlay]
        \path let \n1 = {0.5cm*(sin(60)+16)} in coordinate (r) at (0,\n1) {};
        \path (r) +(270:1) node [node] (x) {};
        \coordinate (bot) at (-5,0);
        \path [name path=path1] (bot) -- +(0:10);
        \path [name path=path2] (x) -- +(255:20);
        \path [name intersections={of=path1 and path2}] node [node] (wl) at (intersection-1) {};
        \path [name path=path2] (x) -- +(285:20);
        \path [name intersections={of=path1 and path2}] node [node] (wr) at (intersection-1) {};
      \end{scope}
      \path (wl) +(75:1.5) node [node] (vl) {};
      \path (wl) +(75:3) node [node] (ul) {};
      \path (wl) +(75:4.5) node [node] (cl) {};
      \path (wr) +(105:1.5) node [node] (vr) {};
      \path (wr) +(105:3) node [node] (ur) {};
      \path (wr) +(105:4.5) node [node] (dr) {};
      \path (x) +(255:1.5) node [node] (a1) {};
      \path (x) +(285:1.5) node [node] (b1) {};
      \path let \p1 = ($(a1) - (cl)$), \n1 = {(veclen(\x1,\y1) - 2.25cm) / 2} in
      (a1) +(255:\n1) node [node] (as) {}
      (b1) +(285:\n1) node [node] (bt) {};
      \path (as) +(255:1.5) node [node] (pl) {};
      \path (as) +(255:3) node [node] (ql) {};
      \path (as) +(255:4.5) node [node] (c1) {};
      \path (bt) +(285:1.5) node [node] (pr) {};
      \path (bt) +(285:3) node [node] (qr) {};
      \path (bt) +(285:4.5) node [node] (d1) {};
      \foreach \v in {a1,as,c1,cl} {
        \path (\v) +(210:1) coordinate (c\v);
      }
      \foreach \v in {b1,bt,d1,dr} {
        \path (\v) +(-30:1) coordinate (c\v);
      }
      \path [name path=path1] (pl) -- +(210:1.25);
      \path [name path=path2] (ql) -- +(120:1.25);
      \path [name intersections={of=path1 and path2}] coordinate (tlbeadcl) at (intersection-1);
      \path [name path=path1] (pl) -- +(300:1.25);
      \path [name path=path2] (ql) -- +(30:1.25);
      \path [name intersections={of=path1 and path2}] coordinate (tlbeadcr) at (intersection-1);
      \path [name path=path1] (pr) -- +(240:1.25);
      \path [name path=path2] (qr) -- +(150:1.25);
      \path [name intersections={of=path1 and path2}] coordinate (trbeadcl) at (intersection-1);
      \path [name path=path1] (pr) -- +(330:1.25);
      \path [name path=path2] (qr) -- +(60:1.25);
      \path [name intersections={of=path1 and path2}] coordinate (trbeadcr) at (intersection-1);
      \path [edge] (r) -- (x) (ca1) -- (a1) -- (x) -- (b1) -- (cb1)
      (cas) -- (as) -- (pl) (cbt) -- (bt) -- (pr)
      (ql) -- (c1) -- (cc1) (qr) -- (d1) -- (cd1)
      (ccl) -- (cl) -- (ul) (cdr) -- (dr) -- (ur)
      (vl) -- (wl) (vr) -- (wr);
      \path [edge,dashed] (a1) -- (as) (c1) -- (cl) (b1) -- (bt) (d1) -- (dr);
      \path [edge] let \p1 = ($(pl) - (tlbeadcr)$), \p2 = ($(ql) - (tlbeadcr)$),
      \p4 = ($(ql) - (tlbeadcl)$), \p5 = ($(pl) - (tlbeadcl)$),
      \n1 = {atan2(\y1,\x1)}, \n2 = {360+atan2(\y2,\x2)}, \n3 = {veclen(\x1,\y1)},
      \n4 = {atan2(\y4,\x4)}, \n5 = {atan2(\y5,\x5)}, \n6 = {veclen(\x4,\y4)} in
      (pl) arc (\n1:\n2:\n3) arc (\n4:\n5:\n6)
      (ul) arc (\n1:\n2:\n3) arc (\n4:\n5:\n6);
      \path [edge] let \p1 = ($(pr) - (trbeadcr)$), \p2 = ($(qr) - (trbeadcr)$),
      \p4 = ($(qr) - (trbeadcl)$), \p5 = ($(pr) - (trbeadcl)$),
      \n1 = {atan2(\y1,\x1)}, \n2 = {360+atan2(\y2,\x2)}, \n3 = {veclen(\x1,\y1)},
      \n4 = {atan2(\y4,\x4)}, \n5 = {atan2(\y5,\x5)}, \n6 = {veclen(\x4,\y4)} in
      (pr) arc (\n1:\n2:\n3) arc (\n4:\n5:\n6)
      (ur) arc (\n1:\n2:\n3) arc (\n4:\n5:\n6);
      \node [anchor=east,shift=(75:4pt)] at (a1) {$a_1$};
      \node [anchor=east,shift=(75:4pt)] at (as) {$a_s$};
      \node [anchor=east,shift=(75:4pt)] at (c1) {$c_1$};
      \node [anchor=east,shift=(75:4pt)] at (cl) {$c_l$};
      \node [anchor=east,shift=(75:4pt)] at (pl) {$p_L$};
      \node [anchor=east] at (ql) {$q_L$};
      \node [anchor=east,shift=(75:4pt)] at (ul) {$u_L$};
      \node [anchor=east] at (vl) {$v_L$};
      \node [anchor=east] at (wl) {$w_L$};
      \node [anchor=west,shift=(105:4pt)] at (b1) {$b_1$};
      \node [anchor=west,shift=(105:4pt)] at (bt) {$b_t$};
      \node [anchor=west,shift=(105:4pt)] at (d1) {$d_1$};
      \node [anchor=west,shift=(105:4pt)] at (dr) {$d_r$};
      \node [anchor=west,shift=(105:4pt)] at (pr) {$p_R$};
      \node [anchor=west,shift=(285:2pt)] at (qr) {$q_R$};
      \node [anchor=west,shift=(105:4pt)] at (ur) {$u_R$};
      \node [anchor=west,shift=(285:2pt)] at (vr) {$v_R$};
      \node [anchor=west] at (wr) {$w_R$};
      \node [anchor=west] at (x) {$x$};
      \node [anchor=north,text height=height("$B'$")] at (current bounding box.south) {$B$};
    \end{tikzpicture}%
    \hspace*{\stretch{1}}%
    \begin{tikzpicture}[
      node/.style={fill=black,circle,inner sep=0pt,minimum size=3pt,outer sep=0pt},
      dup/.style={fill=red,circle,inner sep=0pt,minimum size=6pt},
      wide node/.style={fill=black,rectangle,inner sep=0pt,minimum height=3pt,rounded corners=1.5pt},
      edge/.style={draw,thick},
      x=0.5cm,y=0.5cm]
      \node [node] (wl) at (0,0) {};
      \node [node] (wr) at (1,0) {}; 
      \node [node] (w) at (60:1) {};
      \path (w) foreach \v/\y in {v/1,u/2,dr/3,d1/4.25,cl/5.25,c1/6.5,q/7.5,p/8.5,bt/9.5,b1/10.75,
        as/11.75,a1/13,y/14,x/15} {
        +(90:\y) node [node] (\v) {}
      } +(90:16) coordinate (r);
      \foreach \v in {a1,as,c1,cl} {
        \path (\v) +(210:1) coordinate (c\v);
      }
      \foreach \v in {b1,bt,d1,dr} {
        \path (\v) +(330:1) coordinate (c\v);
      }
      \path [name path=path1] (x) -- +(225:1);
      \path [name path=path2] (y) -- +(135:1);
      \path [name intersections={of=path1 and path2}] coordinate (beadcl) at (intersection-1);
      \path [name path=path1] (x) -- +(-45:1);
      \path [name path=path2] (y) -- +(45:1);
      \path [name intersections={of=path1 and path2}] coordinate (beadcr) at (intersection-1);
      \path [edge] let \p1 = ($(x) - (beadcr)$), \p2 = ($(y) - (beadcr)$),
      \p4 = ($(y) - (beadcl)$), \p5 = ($(x) - (beadcl)$),
      \n1 = {atan2(\y1,\x1)}, \n2 = {360+atan2(\y2,\x2)}, \n3 = {veclen(\x1,\y1)},
      \n4 = {atan2(\y4,\x4)}, \n5 = {atan2(\y5,\x5)}, \n6 = {veclen(\x4,\y4)} in
      (r) -- (x) (y) -- (a1) -- (ca1) (cas) -- (as) -- (b1) -- (cb1)
      (cbt) -- (bt) -- (p) (q) -- (c1) -- (cc1) (ccl) -- (cl) -- (d1) -- (cd1)
      (cdr) -- (dr) -- (u) (v) -- (w) (wl) -- (w) -- (wr)
      (x) arc (\n1:\n2:\n3) arc (\n4:\n5:\n6)
      (p) arc (\n1:\n2:\n3) arc (\n4:\n5:\n6)
      (u) arc (\n1:\n2:\n3) arc (\n4:\n5:\n6);
      \path [edge,dashed] (a1) -- (as) (b1) -- (bt) (c1) -- (cl) (d1) -- (dr);
      \node [anchor=east]                 at (wl) {$w_L$};
      \node [anchor=west]                 at (wr) {$w_R$};
      \node [anchor=west]                 at (w)  {$w$};
      \node [anchor=east,yshift=-2pt]     at (v)  {$v$};
      \node [anchor=east,yshift=2pt]      at (u)  {$u$};
      \node [anchor=east]                 at (dr) {$d_r$};
      \node [anchor=east,yshift=-0.125cm] at (d1) {$d_1$};
      \node [anchor=west]                 at (cl) {$c_l$};
      \node [anchor=west]                 at (c1) {$c_1$};
      \node [anchor=east,yshift=-2pt]     at (q)  {$q$};
      \node [anchor=east,yshift=2pt]      at (p)  {$p$};
      \node [anchor=east]                 at (bt) {$b_t$};
      \node [anchor=east,yshift=-0.125cm] at (b1) {$b_1$};
      \node [anchor=west]                 at (as) {$a_s$};
      \node [anchor=west]                 at (a1) {$a_1$};
      \node [anchor=east,yshift=-2pt]     at (y)  {$y$};
      \node [anchor=east,yshift=2pt]      at (x)  {$x$};
      \node [anchor=north,text height=height("$B'$")] at (current bounding box.south) {$B'$};
    \end{tikzpicture}%
    \hspace*{\stretch{1}}%
    \caption{Given a beaded tree $B$ with a double pair of reticulations,
      $\{(p_L,q_L),(u_L,v_L)\}, \{(p_R,q_R),(u_R,v_R)\}$, where neither pair is
      an ancestor of the other, we can construct a beaded tree $B'$ with fewer
      reticulations.}
    \label{fig:noDoublePairs}
  \end{figure}

  We now show that any MUL-tree weakly displayed by $B$ is also weakly
  displayed by $B'$.  Let $T$ be a MUL-tree weakly displayed by $B$, and let
  $h$ be a weak embedding of $T$ into $B$.  Then we define a weak embedding
  $h'$ of $T$ into $B'$ as follows:  For any node $z \in V(T)$, we set
  \begin{equation*}
    h'(z) = \begin{cases}
      p & \text{if } h(z) \in \{p_L, p_R\}\\
      u & \text{if } h(z) \in \{u_L, u_R\}\\
      h(z) & \text{otherwise}
    \end{cases}.
  \end{equation*}
  Note that $h(z) \notin \{q_L, q_R, v_L, v_R\}$ because $h(z)$ is a tree node
  for all $z \in V(T)$.  Observe that, if there is a path from $u'$ to $v'$ in
  $B$, for any two nodes $u',v' \in V(B) \setminus
  \{p_L,q_L,u_L,v_L,p_R,q_R,u_R,v_R\}$, then there is a path from $u'$ to $v'$
  in $B'$.  Moreover, if there is a path in $B$ from $h(x')$ to $h(y')$, for
  any two nodes $x',y' \in V(T)$, then there is a path $h'(x'y')$ in $B'$ from
  $h'(x')$ to $h'(y')$.  It remains to verify that these paths can be chosen
  such that, for any tree node $x' \in V(T)$ with children $y'$ and $z'$, the
  two paths $h'(x'y')$ and $h'(x'z')$ start with different out-edges of
  $h'(x')$.

  So consider any tree node $x' \in V(T)$ and its two children $y'$ and $z'$.
  Since $h$ is a weak embedding, $h(x')$ is a tree node and, by construction,
  so is $h'(x')$.  If $h'(x')$ is the top node of a bead, then $h'(x'y')$ and
  $h'(x'z')$ can be chosen to start with different parallel arcs of this bead.
  So assume $h'(x')$ is not the top of a bead in $B'$.  Then, by construction,
  $h(x')$ is not the top part of a bead in $B$ and $h(y')$ and $h(z')$ are
  descendend from different children of $h(x')$ in $B$.  If no out-arcs of
  $h(x')$ were deleted in the construction of $B'$, then the children of
  $h'(x')$ are the same as the children of $h(x')$, and these children are
  still ancestors of $h'(y')$ and $h'(z')$. Thus paths $h'(x'y')$ and
  $h'(x'z')$ can still be chosen to start with different out-edges of $h'(x')$.
  The final case is when $h'(x')$ is not the top of a bead and at least one
  out-arc of $h(x')$ was deleted in the construction of $B'$.  In this case,
  $h'(x') \in \{a_s, b_t, c_l, d_r\}$.  It is easy to check that in each of
  these cases, $h'(y')$ and $h'(z')$ are still descendants of different
  children of $h'(x')$.

  This completes the proof that any MUL-tree weakly displayed by $B$ is also
  weakly displayed by $B'$.  Moreover, $B'$ has fewer beads than $B$ (as we
  replaced the four beads $(p_L,q_L), (p_R,q_R), (u_L,v_L), (u_R,v_R)$ with the
  three beads $(x,y), (p,q), (u,v)$), contradicting the optimality of $B$.
  Thus, there is no optimal solution $B$ for $\mathcal{T}$ that does not
  satisfy the lemma.
\end{proof}

Using Lemmas~\ref{lem:allOneLineage} and~\ref{lem:noDoublePairs}, we can show
the following lemma.  Intuitively speaking, it says that any optimal solution
to  \textsc{Beaded Tree} must have ``almost all reticulations on one path'', in
the sense that most reticulations exist on a single path, and any branch coming
off of this path leads to at most one reticulation.

\begin{lemma}\label{lem:almostAllOneLineage}
  Given any optimal solution $B$ to an instance $\mathcal{T}$ of \textsc{Beaded
    Tree}, there exists a path from the root to a leaf of $B$ such that any
  node not on this path has at most one strict descendant that is a
  reticulation.
\end{lemma}

\begin{proof}
  Suppose for the sake of contradiction that the claim does not hold, that is,
  for any path $P$ in $B$, there exists a node $u$ not in $P$ that has at least
  two reticulations among its strict descendants.  In particular, this implies
  that there exist two nodes $a,b$ such that $a$ is not an ancestor of $b$, $b$
  is not an ancestor of $a$, and each of $a$ and $b$ is a strict ancestor of at
  least two reticulations.  Let $B_a$ be the part of $B$ descended from $a$ and
  let $B_b$ be the part of $B$ descended from~$b$.  By
  Lemma~\ref{lem:allOneLineage}, there exist beaded trees $B_a'$ and $B_b'$
  such that $B_a'$ weakly displays every MUL-tree weakly displayed by $B_a$,
  $B_b'$ weakly displays every MUL-tree weakly displayed by $B_b$, $B_a'$ has
  no more reticulations than $B_a$, $B_b'$ has no more reticulations than
  $B_b$, and both $B_a'$ and $B_b'$ have all their reticulations on a single
  path.  By replacing $B_a$ and $B_b$ with $B_a'$ and $B_b'$, respectively, in
  $B$, we obtain a beaded tree $B'$ that weakly displays all MUL-trees in
  $\mathcal{T}$ and has no more reticulations than $B$.  If $B_a'$ or $B_b'$
  has only one bead, then $B'$ has fewer reticulations than $B$, contradicting
  $B$'s optimality.  Thus, $B_a'$ has a bead $(p_a, q_a)$ that is an ancestor
  of another bead $(u_a, v_a)$ in $B_a'$ and $B_b'$ has a bead $(p_b, q_b)$
  that is an ancestor of another bead $(u_b, v_b)$ in $B_b'$.  Since neither
  $(p_a, q_a)$ nor $(p_b, q_b)$ is an ancestor of the other, $B'$ cannot be an
  optimal solution for $\mathcal{T}$, by Lemma~\ref{lem:noDoublePairs}.  Thus,
  since $B'$ has no more reticulations than $B$, $B$ is not an optimal solution
  for $\mathcal{T}$ either, a contradiction.
\end{proof}

\section{Beaded Tree Algorithm}\label{sec:algorithm}

In what follows, we let \textsc{Supertree} denote an algorithm that takes as
input a set of MUL-trees $\mathcal{T}$, and returns either a tree $T$ weakly
displaying all MUL-trees in $\mathcal{T}$ or the value \textsc{None} if no such
tree exists.  The algorithm of~\cite{Aho1981InferringAT} achieves this in
$O(|X|n)$ time, where $n = \sum_{i = 1}^t|T_i|$ and $|T_i|$ is the total number
of nodes in $T_i$.  We note that the algorithm of~\cite{Aho1981InferringAT} is
designed only for MUL-trees with at most one copy of each label, for the simple
reason that there is no tree weakly displaying a MUL-tree with multiple copies
of some label.  Fortunately, the fix for this is straightforward: we just let
\textsc{Supertree} return \textsc{None} whenever $\mathcal{T}$ contains a
MUL-tree with two or more copies of some label.  By the following lemma, an
optimal solution for any instance of \textsc{Beaded Tree} can be found in
polynomial time using Algorithm~\ref{alg:greedyParentalHybrid}.
An example of a network produced by this algorithm is shown in
Figure~\ref{fig:algorithm}.

\begin{algorithm}[t]
  \KwIn{A set of MUL-trees $\mathcal{T} = \{T_1, \ldots, T_t\}$}
  \KwOut{A beaded tree $B$ with the minimum number of reticulations
    that weakly displays the MUL-trees in $\mathcal{T}$}
  \BlankLine
  \uIf{$|X| = 1$ and $\max_{1 \le i \le t}|L(T_i)| = 1$}{
    \Return a tree with $1$ leaf on $X$\;
  }
  \Else{
    Calculate the split partition $\{S_1, \ldots, S_s\}$ of $\mathcal{T}$\;
    \For{$i \gets 1$ \KwTo $s$}{
      $T \gets \textsc{Supertree}(\mathcal{T}|_{S_i})$\;
      \If{$T \ne \textsc{None}$}{
        $B' \gets \textsc{Beaded-Tree}(\mathcal{T}\setminus S_i)$\; 
        Construct $B$ by joining $B'$ and $T$\;
        \Return $B$\;
      }
    }
    $B' \gets \textsc{Beaded-Tree}(F_1(\mathcal{T})$)\;
    Construct $B$ by adding a bead whose child is the root of $B'$\;
    \Return $B$\;
  }
  \caption{Algorithm \textsc{Beaded-Tree}($\mathcal{T}$)}
  \label{alg:greedyParentalHybrid}
\end{algorithm}

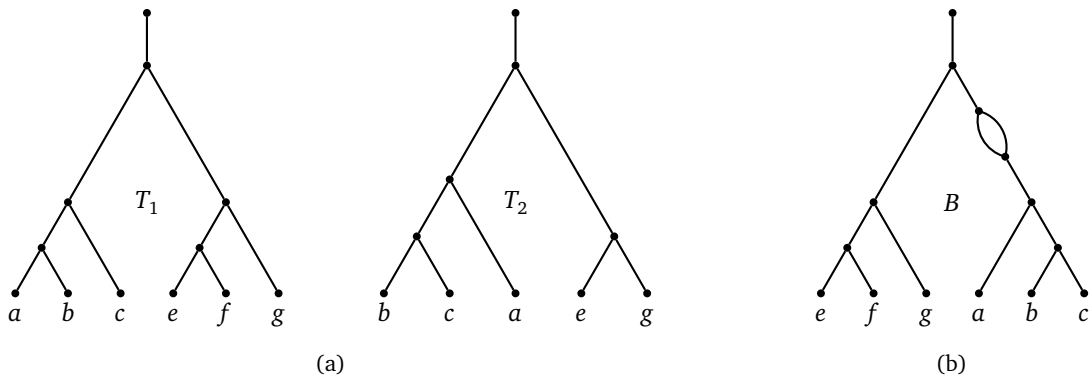
\begin{figure}[t]
  \hspace*{\stretch{1}}%
  \subcaptionbox{}{%
    \begin{tikzpicture}[
      node/.style={fill=black,circle,inner sep=0pt,minimum size=3pt,outer sep=0pt},
      dup/.style={fill=red,circle,inner sep=0pt,minimum size=6pt},
      wide node/.style={fill=black,rectangle,inner sep=0pt,minimum height=3pt,rounded corners=1.5pt},
      edge/.style={draw,thick},
      x=0.7cm,y=0.7cm]
      \begin{scope}
        \foreach \i in {0,...,5} {
          \node [node] (l\i) at (\i,0) {};
        }
        \path (l0) +(60:1) node [node] (p01) {};
        \path (l0) +(60:2) node [node] (p02) {};
        \path (l0) +(60:5) node [node] (p05) {};
        \path (l3) +(60:1) node [node] (p34) {};
        \path (l3) +(60:2) node [node] (p35) {};
        \path (p05) +(90:1) node [node] (r) {};
        \node [anchor=north,text height=height("$f$")] at (l0) {$a$};
        \node [anchor=north,text height=height("$f$")] at (l1) {$b$};
        \node [anchor=north,text height=height("$f$")] at (l2) {$c$};
        \node [anchor=north,text height=height("$f$")] at (l3) {$e$};
        \node [anchor=north,text height=height("$f$")] at (l4) {$f$};
        \node [anchor=north,text height=height("$f$")] at (l5) {$g$};
        \path [edge] (l0) -- (p01) -- (l1) (p01) -- (p02) -- (l2) (l3) -- (p34) -- (l4)
        (p34) -- (p35) -- (l5) (p02) -- (p05) -- (p35) (p05) -- (r);
        \node (T1) at (barycentric cs:p02=0.5,p35=0.5) {$T_1$};
      \end{scope}
      \begin{scope}[xshift=4.9cm,x=0.875cm,y=0.875cm]
        \foreach \i in {0,...,4} {
          \node [node] (l\i) at (\i,0) {};
        }
        \node [anchor=north,text height=height("$f$")] at (l0) {$b$};
        \node [anchor=north,text height=height("$f$")] at (l1) {$c$};
        \node [anchor=north,text height=height("$f$")] at (l2) {$a$};
        \node [anchor=north,text height=height("$f$")] at (l3) {$e$};
        \node [anchor=north,text height=height("$f$")] at (l4) {$g$};
        \path (l0) +(60:1) node [node] (p01) {};
        \path (l0) +(60:2) node [node] (p02) {};
        \path (l0) +(60:4) node [node] (p04) {};
        \path (l3) +(60:1) node [node] (p34) {};
        \path (p04) +(90:0.8) node [node] (r) {};
        \path [edge] (l0) -- (p01) -- (l1) (p01) -- (p02) -- (l2) (l3) -- (p34) -- (l4)
        (p02) -- (p04) -- (p34) (p04) -- (r);
        \node at (T1 -| r) {$T_2$};
      \end{scope}
    \end{tikzpicture}}%
  \hspace*{\stretch{2}}%
  \subcaptionbox{}{%
    \begin{tikzpicture}[
      node/.style={fill=black,circle,inner sep=0pt,minimum size=3pt,outer sep=0pt},
      dup/.style={fill=red,circle,inner sep=0pt,minimum size=6pt},
      wide node/.style={fill=black,rectangle,inner sep=0pt,minimum height=3pt,rounded corners=1.5pt},
      edge/.style={draw,thick},
      x=0.7cm,y=0.7cm]
      \begin{scope}
        \foreach \i in {0,...,5} {
          \node [node] (l\i) at (\i,0) {};
        }
        \path (l0) +(60:1) node [node] (p01) {};
        \path (l0) +(60:2) node [node] (p02) {};
        \path (l0) +(60:5) node [node] (p05) {};
        \path (l3) +(60:2) node [node] (p35) {};
        \path (l4) +(60:1) node [node] (p45) {};
        \path (p05) +(90:1) node [node] (r) {};
        \node [node] (u) at (barycentric cs:p05=0.667,p35=0.333) {};
        \node [node] (v) at (barycentric cs:p05=0.333,p35=0.667) {};
        \path [name path=path1] (u) -- +(-15:1);
        \path [name path=path2] (v) -- +(75:1);
        \path [name intersections={of=path1 and path2}] coordinate (beadcr) at (intersection-1);
        \path [name path=path1] (u) -- +(255:1);
        \path [name path=path2] (v) -- +(165:1);
        \path [name intersections={of=path1 and path2}] coordinate (beadcl) at (intersection-1);
        \node [anchor=north,text height=height("$f$")] at (l0) {$e$};
        \node [anchor=north,text height=height("$f$")] at (l1) {$f$};
        \node [anchor=north,text height=height("$f$")] at (l2) {$g$};
        \node [anchor=north,text height=height("$f$")] at (l3) {$a$};
        \node [anchor=north,text height=height("$f$")] at (l4) {$b$};
        \node [anchor=north,text height=height("$f$")] at (l5) {$c$};
        \path [edge] (l0) -- (p01) -- (l1) (p01) -- (p02) -- (l2) (l3) -- (p35) -- (p45)
        (l4) -- (p45) -- (l5) (p02) -- (p05) -- (u) (v) -- (p35) (p05) -- (r);
        \path [edge] let \p1 = ($(u) - (beadcr)$), \p2 = ($(v) - (beadcr)$),
        \p4 = ($(v) - (beadcl)$), \p5 = ($(u) - (beadcl)$),
        \n1 = {atan2(\y1,\x1)}, \n2 = {360+atan2(\y2,\x2)}, \n3 = {veclen(\x1,\y1)},
        \n4 = {atan2(\y4,\x4)}, \n5 = {atan2(\y5,\x5)}, \n6 = {veclen(\x4,\y4)} in
        (u) arc (\n1:\n2:\n3) arc (\n4:\n5:\n6);
        \node at (barycentric cs:p02=0.5,p35=0.5) {$B$};
      \end{scope}
    \end{tikzpicture}}%
  \hspace*{\stretch{1}}%
  \caption{(a) An instance $\mathcal{T} = \{T_1, T_2\}$ of \textsc{Beaded Tree}.  (b)
    The beaded tree $B$ constructed by running algorithm \textsc{Beaded-Tree} on
    $\mathcal{T}$. Initially, the split partition is $\{a,b,c\}, \{e,f,g\}$. As
    \textsc{Supertree} returns a tree on $\{e,f,g\}$, the top tree node of $B$ has
    that tree as one of its children. To construct the other side of $B$, we
    run \textsc{Beaded-Tree} on $\mathcal{T}|_{\{a,b,c\}}$, and \textsc{Supertree} does
    not return a tree on this set. Therefore, this side of $B$ begins with a
    bead.}
  \label{fig:algorithm}
\end{figure}

\begin{lemma}\label{lem:beadedTreeRunTime}
  Let $\mathcal{T} = \{T_1, \ldots, T_t\}$ be an instance of \textsc{Beaded
    Tree}, let $n = \sum_{i = 1}^t|T_i|$, and let $k$ be the reticulation
  number of an optimal solution for $\mathcal{T}$.
  Algorithm~\ref{alg:greedyParentalHybrid} finds an optimal solution for
  $\mathcal{T}$ in $O((|X|^2 + |X|k)n)$ time.
\end{lemma}

\begin{proof}
  The correctness of the algorithm follows from
  Lemma~\ref{lem:structureFromSpitPartition}.  To analyze the running time,
  observe that each recursive call of \textsc{Beaded-Tree} acts on an instance
  $\mathcal{T'}$ on leaf set $X'$ such that either $|X'| < |X|$ and an optimal
  solution for $\mathcal{T'}$ has at most as many reticulations as an optimal
  solution for $\mathcal{T}$, or $X' = X$ and an optimal solution for
  $\mathcal{T'}$ has fewer reticulations than an optimal solution for
  $\mathcal{T}$.  It follows that the algorithm makes at most $k + |X| + 1$
  recursive calls of \textsc{Beaded-Tree}, where $k$ is the reticulation number
  of an optimal solution to $\mathcal{T}$.

  To determine the cost of a single invocation of \textsc{Beaded-Tree}, observe
  that line~14 clearly takes constant time and line~4 takes $O(n)$ time.
  Indeed, it takes $O(|X|) = O(n)$ time to construct a graph $G = (X,
  \emptyset)$.  Then, for each tree~$T_i$, we compute the connected components
  of its depth-1 forest in $O(|T_i|)$ time.  For each such component~$C$, we
  choose one of its leaves as the ``representative leaf'' $\ell$ of the
  component and add an edge $(\ell, x)$ to $G$ for every leaf $x$ in $C$.  This
  also takes $O(|T_i|)$ time.  Doing this for all trees in $\mathcal{T}$ takes
  $O(\sum_{i=1}^t |T_i|) = O(n)$ time.  The split partition of $\mathcal{T}$ is
  now easily seen to be the partition of $X$ into the vertex sets of the
  connected components of~$G$, which can be computed in $O(n)$ time.  Each
  iteration of the for-loop in lines 5--12, excluding lines~6 and~8 takes
  constant time.  In line~6, the construction of $\mathcal{T}|_{S_i}$ is easily
  accomplished in $O(|S_i|n)$ time and the call to \textsc{Supertree} takes
  $O(|S_i|n)$ time.  Thus, excluding the cost of line~8, the total cost of all
  iterations of the for-loop is $O(\sum_{i=1}^s |S_i|n) = O(|X|n)$ and the
  entire invocation of \textsc{Beaded-Tree} takes $O(|X|n)$ time.

  Since the algorithm makes at most $k + |X| + 1$ invocations, its total cost is
  thus $O((k + |X| + 1) |X|n) = O((|X|^2 + |X|k)n)$.
\end{proof}

\section{Minimizing Bead Depth}\label{sec:depth}

Lemma~\ref{lem:almostAllOneLineage} implies that any optimal solution to an
instance of \textsc{Beaded Tree} has a very restrictive structure.  Informally
speaking, there is a single path in the beaded tree that may contain any number
of reticulations, and any ``branch'' coming off this path can contain at most
one reticulation.  Because of the close relationship between \textsc{Beaded
  Tree} and \textsc{Unrestricted Minimal Episodes Inference} (described in Lemma~\ref{lem:WGDequivalence}), the same
structural properties apply to optimal solutions for the latter problem: there
is one main path containing any number of duplication episodes, and any path
branching off from the main path contains at most one duplication episode.

This structure is quite unusual.  Furthermore, it is not  clear why, from a
biological perspective, it should be the case that most duplications occur on a
single path.  For this reason, we now consider the
problems \textsc{Unrestricted Minimal Episodes Depth Inference} and
\textsc{Beaded Tree Depth}.

\medskip

\noindent \textsc{Unrestricted Minimal Episodes Depth Inference}\\
\noindent \textbf{Input}: A set $\mathcal{T} = \{T_1, \ldots, T_t\}$ of
MUL-trees with label sets $X_1, \ldots, X_t \subseteq X$.\\
\noindent \textbf{Output}:  A duplication tree $D$ on $X$ with the minimum
number of duplication nodes on any path from the root to a leaf and such that
$D$ is consistent with each of~$T_1,\ldots, T_t$.

\medskip

\noindent \textsc{Beaded Tree Depth}\\
\noindent \textbf{Input}: A set $\mathcal{T} = \{T_1, \ldots, T_t\}$ of
MUL-trees with label sets $X_1, \ldots, X_t \subseteq X$.\\
\noindent \textbf{Output}:  A beaded tree $B$ on $X$ with the minimum number of
beads on any directed path and such that $B$ weakly displays each
of~$T_1,\ldots, T_t$.

\medskip

By a similar argument to the proof of Lemma~\ref{lem:WGDequivalence}, these
two problems are equivalent.

\textsc{Unrestricted Minimal Episodes Depth Inference} is loosely based on the
following two assumptions: separate lineages accumulate duplications
independently; there is a maximal duplication rate that does not vary too much
between lineages. Given that $d$ duplication episodes happened on one path,
these assumptions make it reasonable to expect at most $d$ duplication episodes
on any other path disjoint from it (with same evolutionary length). In
particular, this holds for all paths (lineages) starting at the root, which
justifies the maximum depth formulation.  These assumptions seem close to those
of evolutionary models. However, this does not make the \textsc{Unrestricted
  Minimal Episodes Depth Inference} problem model-based.  The problem is still
one of parsimony: we \emph{minimize} the maximum depth or, equivalently, the
duplication rate.

Note that solutions to \textsc{Unrestricted Minimal Episodes Depth Inference}
explicitly do not contain unnecessarily highly placed duplications: where the
proof of Lemma~\ref{lem:allOneLineage} ``zipped'' duplication episodes at much
as possible, we now ``unzip'' them to avoid ``stacking'' duplications as in the
proof of Lemma~\ref{lem:allOneLineage}.  Hence, this new problem is
biologically motivated and it has more reasonable solutions than
\textsc{Unrestricted Minimal Episodes Inference}.

Fortunately, it turns out that a similar algorithm to that for \textsc{Beaded
  Tree} can be used to solve \textsc{Beaded Tree Depth}. The difference between
the two algorithms may be summed up as follows: Both algorithms begin by
considering the split partition of the set of MUL-trees under consideration. If
any set in this partition can be ``solved'' using a tree, then for both
problems it is always optimal to assume that the solution does not start with a
bead, but instead includes such a tree as a child of the top tree node.  If the
split partition consists of a single set (and there is more than one leaf),
then any possible solution (even a non-optimal solution) must begin with a
bead.  For the remaining cases, we essentially have a choice; there exist
solutions that begin with a bead and solutions that don't.  The algorithm for
\textsc{Beaded Tree} always introduces a bead in these cases; the algorithm for
\textsc{Beaded Tree Depth} never does.  The following lemma is the basis for
our algorithm to solve \textsc{Beaded Tree Depth} and establishes its
correctness.

\begin{lemma}\label{lem:depthStructureFromSpitPartition}
  Let $\mathcal{T}$ be an instance of \textsc{Beaded Tree Depth}, and let
  $\{S_1, \ldots, S_s\}$ be the split partition of $\mathcal{T}$.  If $|X|= 1$
  and $\max_{1 \le i \le t} |L(T_i)| = 1$, then the optimal solution is the
  tree with a single leaf on $X$. Otherwise, if $s = 1$, then every optimal
  solution $B$ has a bead $(u,v)$ at the root and the child of $v$ is the root
  of an optimal solution for $F_1(\mathcal{T})$.  If $s > 1$, then any network~$B$ obtained by joining an optimal solution for~$\mathcal{T}|_{S_s}$ with an optimal solution for $\mathcal{T}\setminus S_s$ is optimal for~$\mathcal{T}$. Such a network~$B$ always exists.
\end{lemma}

\begin{proof}
  If $|X| = 1$ and $\max_{1 \le i \le t}|L(T_i)| = 1$, then the optimal
  solution clearly is the tree with a single leaf on~$X$. So assume that $|X| >
  1$ and assume first that the split partition of $\mathcal{T}$ is trivial ($s
  = 1$).  Consider an optimal solution $B$.  We prove first that $B$ must have
  a bead at the root.  Assume the contrary.  Since either $|X| > 1$ or $\max_{1
    \le i \le t}|L(T_i)|>1$, $B$ is not a tree with a single leaf.  Therefore,
  the child of the root of $B$ is a \emph{split node} $a$, that is, a tree node
  that is not in a bead. Let $b_1$ and $b_2$ be the children of $a$ and let $S$
  and $X \setminus S$ be the disjoint leaf sets descended from $b_1$ and $b_2$,
  respectively.  Since both $b_1$ and $b_2$ have non-empty sets of descendant
  leaves, $S$ is a non-empty proper subset of $X$.

  Since the split partition is trivial, there exists at least one MUL-tree $T
  \in \mathcal{T}$ such that some MUL-trees $T' \in F_1(T)$ has a leaf $\ell_1
  \in S$ and a leaf $\ell_2 \in X \setminus S$.  Let $r_T$ be the root of $T$,
  $x$ the child of $r_T$, and $y_l$ and $y_r$ the children of $x$. Without loss
  of generality, $T'$ is the tree obtained from $T$ by deleting $y_r$ and all
  its descendants, and suppressing $x$.  We show that $B$ does not weakly
  display $T$, which is the desired contradiction.  So consider any weak
  embedding $h$ of $T$ into $B$.  If $h(y_l)$ is a proper descendant of $a$,
  then either $h(\ell_1)$ or $h(\ell_2)$ is not a descendant of $h(y_l)$, a
  contradiction because $y_l$ is an ancestor of both $\ell_1$ and $\ell_2$ in
  $T$.  Thus, $h(y_l) \in \{r, a\}$.  $h(y_l) = r$ is impossible because $h(x)$
  must be a proper ancestor of $h(y_l)$.  Thus, $h(y_l) = a$ and $h(x) = r$.
  Since $a$ is the only child of $r$, this implies that both paths $h(xy_l)$
  and $h(xy_r)$ start with the edge $ra$, again a contradiction.  This proves
  that $B$ must have a bead at the root.

  The part of $B$ descended from this bead must be an optimal solution
  to $F_1(\mathcal{T})$ because otherwise we could obtain a solution
  for $\mathcal{T}$ that is better than $B$ by constructing an optimal solution for
  $F_1(\mathcal{T})$ and adding a bead at its root.
  This proves the lemma for the case when $s = 1$.

  Finally, assume that $\mathcal{T}$ does not have a trivial split partition,
  that is, $s > 1$.  For any collection $\mathcal{T}'$ of MUL-trees, let
  $\OPT(\mathcal{T'})$ denote an optimal solution to $\mathcal{T}'$.
  For any beaded tree $B$, let $d(B)$ be the maximum number of beads along
  any root-to-leaf path in $B$.
  We show first that the beaded tree $B$ defined in the lemma weakly displays
  all trees in $\mathcal{T}$.

  Any MUL-tree $T \in \mathcal{T}$ with no leaves in $S_s$ is weakly displayed
  by $\OPT(\mathcal{T} \setminus S_s)$ and therefore by $B$.  Similarly, if all
  leaves of $T$ belong to $S_s$, then $T$ is weakly displayed by
  $\OPT(\mathcal{T}|_{S_s})$ and therefore by $B$.  So suppose $T$ has leaves
  in both $S_s$ and $X \setminus S_s$.  Since $F_1(T)$ consists of two
  MUL-trees and $\{S_1, \ldots, S_s\}$ is a split partition of $\mathcal{T}$,
  we must have $F_1(T) = \{T|_{S_s}, T \setminus S_s\}$.  Since $T|_{S_s} \in
  \mathcal{T}_{S_s}$ and $T \setminus S_s \in \mathcal{T} \setminus S_s$, the
  former is weakly displayed by $\OPT(\mathcal{T}_{S_s})$ and the latter is
  weakly displayed by $\OPT(\mathcal{T} \setminus S_s)$.  Thus, $T$ is once
  again weakly displayed by $B$.  This shows that $B$ weakly displays all trees
  in $\mathcal{T}$.

  Now, since $\OPT(\mathcal{T})$ weakly displays all MUL-trees in
  $\mathcal{T}|_{S_s}$ and $\mathcal{T} \setminus S_s$ and $B$ is obtained by
  joining $\OPT(\mathcal{T}|_{S_s})$ and $\OPT(\mathcal{T} \setminus S_s)$, we
  have $d(B) = \max(d(\OPT(\mathcal{T}|_{S_s})), d(\OPT(\mathcal{T} \setminus
  S_s))) \le d(\OPT(\mathcal{T}))$, that is, $B$ is an optimal solution
  for~$\mathcal{T}$.
\end{proof}

\begin{algorithm}[t]
  \KwIn{A set of MUL-trees $\mathcal{T} = \{T_1, \ldots, T_t\}$}
  \KwOut{A beaded tree $B$ with the minimum bead depth that weakly displays all
    MUL-trees in $\mathcal{T}$}
  \BlankLine
  \uIf{$|X|= 1$ and $\max_{1 \le i\le t}|L(T_i)|=1$}{
    \Return a tree with $1$ leaf on $X$\;
  }
  \Else{
    Calculate the split partition $\{S_1, \ldots, S_s\}$ of $\mathcal{T}$\;
    \uIf{$s = 1$}{
      $B' \gets \textsc{Bead-Depth}(F_1(\mathcal{T})$)\;
      Construct $B$ by adding a bead whose child is the root of $B'$\;
    }
    \Else{
      $B' \gets \textsc{Bead-Depth}(\mathcal{T}|_{S_s})$\;
      $B'' \gets \textsc{Bead-Depth}(\mathcal{T}\setminus S_s)$\;
      Construct $B$ by joining $B'$ and $B''$\;
    }
    \Return $B$\;
  }
  \caption{Algorithm \textsc{Bead-Depth}($\mathcal{T}$)}
  \label{alg:greedyBeadDepth}
\end{algorithm}

\begin{lemma}
  Let $\mathcal{T} = \{T_1, \ldots, T_t\}$ be an instance of \textsc{Beaded
    Tree Depth}.  Algorithm~\ref{alg:greedyBeadDepth} finds an optimal solution
  for $\mathcal{T}$ in $O((|X|^2 + |X|k)n)$ time, where $n = \sum_{i =
    1,}^t|T_i|$ and $k$ is the reticulation number of the computed solution.
\end{lemma}

\begin{proof}
  The correctness of the algorithm follows from
  Lemma~\ref{lem:depthStructureFromSpitPartition}.  To analyze the running
  time, the cost per invocation of \textsc{Bead-Depth} is $O(|X|n)$, by the
  same analysis as in the proof of Lemma~\ref{lem:beadedTreeRunTime}.
  To bound the number of recursive calls, observe that the input to the
  recursive call in line~6 has label set $X$ and has an optimal solution
  with $k-1$ reticulations.
  The inputs to the recursive calls in lines~9 and~10 have label sets
  $S_s$ and $X \setminus S_s$ and have optimal solutions with $k_1$ and
  $k_2$ reticulations, respectively, where $k_1 + k_2 = k$.
  Thus, if $S(x, k)$ is the number of recursive calls made on an input
  with $x = |X|$ and having $k$ reticulations in the optimal solution,
  we have $S(x, k) = 1 + \min(S(x, k-1), S(x_1, k_1) + S(x_2, k_2))$,
  where $x_1 + x_2 = x$ and $k_1 + k_2 = k$.
  This recurrence has the solution $S(x, k) = 2(x + k) - 1$.
  Thus, the running time of the algorithm is $O((2|X| + 2k - 1) \cdot |X|n)
  = O((|X|^2 + |X|k)n)$.
\end{proof}

\section{Concluding Remarks}\label{sec:conclusion}

Although we have shown that the \textsc{Unrestricted Minimal Episodes
  Inference} and \textsc{Parental Hybridization} problems are solvable in
polynomial time, we have also shown that the phylogenies produced by solving
these problems have a severely restricted structure.

The optimal phylogenetic network that our algorithm produces for the
\textsc{Parental Hybridization} problem is always a phylogenetic tree with
``beads'', where a bead consists of a speciation directly followed by a
reticulation. Such solutions are not necessarily the most realistic or likely
ones since they contain a lot of ``extra lineages'', that is, multiple lineages
of an input tree travelling through the same branch of the phylogenetic
network.  Minimizing the total number of extra lineages, the \emph{XL-score}~\cite{yu2011coalescent},
irrespective of the reticulation number, is also not ideal, since there always
exists a solution with zero extra lineages and possibly a very high
reticulation number. Therefore, the most relevant open problem that needs to be
solved is to find a phylogenetic network that minimizes a weighted sum of the
XL-score and the reticulation number of the network. Another alternative
problem formulation that seems reasonable is to minimize the total number of
parental trees that the constructed phylogenetic network has in addition to the
input trees.

Another option would be to completely exclude beads in the solutions. However,
although this is an interesting theoretical open problem, we do not see a
reason why the resulting optimal solutions would by any more realistic, or why
it would be reasonable to assume that a speciation cannot be followed by a
reticulation. 

Regarding \textsc{Unrestricted Minimal Episodes Inference}, the situation is in
some sense even worse. We have shown that \emph{all} optimal solutions have
a very specific structure: there is one main path from the root to a taxon
  containing potentially many duplication episodes, while each path branching
  off this main path contains at most one duplication episode. Although such
scenarios are not to be excluded (for example see the eukaryotic species
phylogeny from~\cite{guigo1996reconstruction}), it is unrealistic to expect all
phylogenies to look like this (see for example Figure~\ref{fig:WGDExample} for
a phylogeny where the duplication episodes are significantly more spread out).
Therefore, we have proposed an alternative problem that minimizes the
``duplication depth'': the maximum number of duplication episodes that lie on
any directed path. This problem can also be solved in polynomial time and we
expect it to produce more realistic solutions. Moreover, note that, although the
problem definition does not exclude unnecessary duplication episodes as long as
they do not increase the duplication depth, our algorithm will not create such
redundant duplication episodes. Nevertheless, to properly assess the two
algorithms, it is necessary to implement both algorithms and extensively test
them on simulated and real biological data sets.


Finally, it would be interesting to study more general problem variants, which
simultaneously take different processes into account, such as duplication
episodes, hybridization, and gene loss and transfers. Although such problems
have been studied in a reconciliation setting where the species tree is
(assumed to be) known, there has been less work on variants where the species
tree or network needs to be inferred. Although such problems seem daunting, we
have shown here that not knowing the species tree can actually make
computational problems easier.

\bibliographystyle{alpha}
\bibliography{parentalHybrid}

\end{document}